\long\def\ignore#1{}
\def\cut{\ignore}

\cut{65 + 24 = 89 figure files:
f41.pdf f42.pdf f10.pdf f13.pdf f35.pdf 
f14.pdf f37.pdf f40.pdf f38.pdf f39.pdf 
graph1.pdf f36.pdf f1.pdf f8.pdf f9.pdf 
f11.pdf f12.pdf f15.pdf f30.pdf f45.pdf 
f46.pdf f31.pdf f32.pdf f33.pdf f34.pdf 
f19.pdf f20.pdf f43.pdf f44.pdf f23.pdf 
f24.pdf f21.pdf f22.pdf f51.pdf f52.pdf 
f53.pdf f54.pdf f55.pdf f56.pdf f57.pdf 
f58.pdf f59.pdf f60.pdf f61.pdf f62.pdf 
f63.pdf f64.pdf f65.pdf f66.pdf f67.pdf 
f68.pdf f69.pdf f70.pdf f47.pdf f48.pdf 
f84.pdf f85.pdf f49.pdf f80.pdf f81.pdf 
f86.pdf f87.pdf f88.pdf f102.pdf f105.pdf

graph2.pdf graph3.pdf graph4.pdf graph5.pdf graph6.pdf
graph7.pdf graph8.pdf graph9.pdf graph10.pdf graph11.pdf
graph12.pdf graph13.pdf graph14.pdf graph15.pdf graph16.pdf
graph17.pdf graph18.pdf graph19.pdf graph20.pdf graph21.pdf
graph22.pdf graph23.pdf graph24.pdf graph25.pdf
} 

\cut{list of lemmas:
group-with-at-least-10
group-with-at-least-11
group-with-exactly-11
group-with-at-least-12
group-with-at-least-13
V-has-10
V-has-11
V-has-12
V-has-13
} 

\documentclass[11pt]{article}
\usepackage{graphicx}
\usepackage{graphics}
\usepackage{amssymb}
\usepackage{amsmath}
\usepackage{amsthm}
\usepackage{psfrag}
\usepackage{enumerate}
\usepackage{epsfig}
\usepackage{epsf}
\usepackage{amsfonts}
\usepackage{latexsym}
\usepackage{fancyvrb}
\usepackage[small]{caption}
\usepackage{url}
\usepackage{array,multirow}
\usepackage{longtable}
\usepackage{float}
\usepackage{verbatim}
\usepackage{listings}
\usepackage{color}
\usepackage{tcolorbox}
\usepackage{hyperref}

\newtheorem{theorem}{Theorem}
\newtheorem{lemma}{Lemma}

\def\etal{{et~al.}}
\def\ie{{i.e.}}

\def\wrt{{w.r.t.~}}

\newcommand{\old}[1]{{}}
\newcommand{\later}[1]{{}}

\newcommand{\NN}{\mathbb{N}}

\newcommand{\RR}{\mathbb{R}}

\def\A{{\mathcal A}}
\def\D{{\mathcal D}}
\def\U{{\mathcal U}}

\def\In{{\tt In}}
\def\Out{{\tt Out}}

\definecolor{dkgreen}{rgb}{0,0.4,0}
\definecolor{gray}{rgb}{0.5,0.5,0.5}
\definecolor{mauve}{rgb}{0.58,0,0.82}
\definecolor{orange}{rgb}{0.8,0.1,0}

\graphicspath{{FIGURES/}}

\title{\textsc{Monotone Paths in Geometric Triangulations}\thanks{An extended abstract
of this paper appeared in the \emph{Proceedings of the 27th International Workshop
on Combinatorial Algorithms (IWOCA 2016)}, LNCS~9843, pp.~411--422, Springer
International Publishing, 2016.}}

\author{Adrian Dumitrescu\thanks{Department of Computer Science,
University of Wisconsin--Milwaukee, USA\@.
Email:~\texttt{dumitres@uwm.edu}.}
\and
Ritankar Mandal\thanks{Department of Computer Science,
University of Wisconsin--Milwaukee, USA\@.
Email:~\texttt{rmandal@uwm.edu}.}
\and
Csaba D. T\'oth\thanks{Department of Mathematics, California State
University, Northridge, Los Angeles, CA, USA; and
Department of Computer Science, Tufts University, Medford, MA, USA.
Supported in part by the NSF awards CCF-1422311 and CCF-1423615.
Email:  \texttt{cdtoth@acm.org}.
}}

\begin{document}

\maketitle

\begin{abstract}
(I) We prove that the (maximum) number of monotone paths in a geometric
triangulation of $n$ points in the plane is $O(1.7864^n)$. This improves an
earlier upper bound of $O(1.8393^n)$; the current best lower bound is
$\Omega(1.7003^n)$.

(II) Given a planar geometric graph $G$ with $n$ vertices, we show that
the number of monotone paths in $G$ can be computed in $O(n^2)$ time.

\end{abstract}

\section{Introduction} \label{sec:intro}

A directed polygonal path $\xi$ in $\mathbb{R}^d$ is \emph{monotone} if
there exists a nonzero vector $\mathbf{u}\in \mathbb{R}^d$ that
has a positive inner product with every directed edge of $\xi$.
The study of combinatorial properties of monotone paths is motivated
by the classical simplex algorithm in linear programming, which finds
an optimal solution by tracing a monotone path in the $1$-skeleton
of a $d$-dimensional polytope of feasible solutions.
It remains an elusive open problem whether there is a pivoting rule for the simplex
method that produces a monotone path whose length is polynomial in $d$ and $n$~\cite{APR14}.

Let $S$ be a set of $n$ points in the plane.
A {\em geometric graph} $G$ is a graph drawn in the plane so
that the vertex set consists of the points in $S$ and the edges
are drawn as straight line segments between the corresponding points in $S$.
A \emph{plane geometric graph} is one in which edges intersect only at common endpoints.
In this paper, we are interested in the maximum number of monotone paths over all
plane geometric graphs with $n$ vertices; it is easy to see that
triangulations maximize the number of such paths (since adding edges can only increase
the number of monotone paths).

\paragraph{Our results.}
We first show that the number of monotone paths in a triangulation of
$n$ points in the plane is $O(1.8193^n)$, using a fingerprinting
technique in which incidence patterns of $8$ vertices are analyzed.
We then give a sharper bound of $O(1.7864^n)$ using the same strategy,
by enumerating fingerprints of $11$ vertices using a computer program.

\begin{theorem} \label{thm:upper}
The number of monotone paths in a geometric triangulation on $n$ vertices
in the plane is $O(1.7864^n)$.
\end{theorem}

It is often challenging to determine the number of configurations (\ie, count)
faster than listing all such configurations (\ie, enumerate).
In Section~\ref{sec:compute} we show that monotone paths can be counted
in polynomial time in plane graphs.
\begin{theorem}\label{thm:compute}
Given a plane geometric graph $G$ with $n$ vertices,
the number of monotone paths in $G$ can be computed in $O(n^2)$ time.
The monotone paths can be enumerated in an additional $O(1)$-time per edge,
\ie, in $O(n^2+K)$ time, where $K$ is the sum of the lengths of all monotone paths.
\end{theorem}

\paragraph{Related previous work.}
We derive a new upper bound on the maximum number of monotone paths in
geometric triangulations of $n$ points in the plane.
Analogous problems have been studied for cycles, spanning cycles,
spanning trees, and matchings~\cite{BKK07} in $n$-vertex edge-maximal
planar graphs, which are defined in purely graph theoretic terms.
In contrast, the monotonicity of a path depends on the embedding of
the point set in the plane, \ie, it is a \emph{geometric} property.
The number of geometric configurations contained (as a subgraph)
in a triangulation of $n$ points have been considered only recently.
The maximum number of \emph{convex polygons} is known to be between
$\Omega(1.5028^n)$ and $O(1.5029^n)$~\cite{DT14,KLP12}. For the number
of \emph{monotone paths}, Dumitrescu~\etal~\cite{DLST16} gave an upper bound
of $O(1.8393^n)$; we briefly review their proof in Section~\ref{sec:prelim}.
A lower bound of $\Omega(1.7003^n)$ is established in the same paper.
It can be deduced from the following construction illustrated in
Fig.~\ref{fig:lb-constr}.
Let $n = 2^\ell + 2$ for an integer $\ell \in \NN$;
the plane graph $G$ has $n$ vertices $V = \{v_1,\ldots,v_n\}$, it contains
the Hamiltonian path $\xi_0=(v_1,\ldots,v_n)$, and it has edge $(v_i, v_{i+2^k})$,
for $1 \leq i \leq n-2^k$, iff $i-1$ or $i-2$ is a multiple of $2^k$.
\begin{figure}[htpb]
\centering
\includegraphics[width=0.39\textwidth]{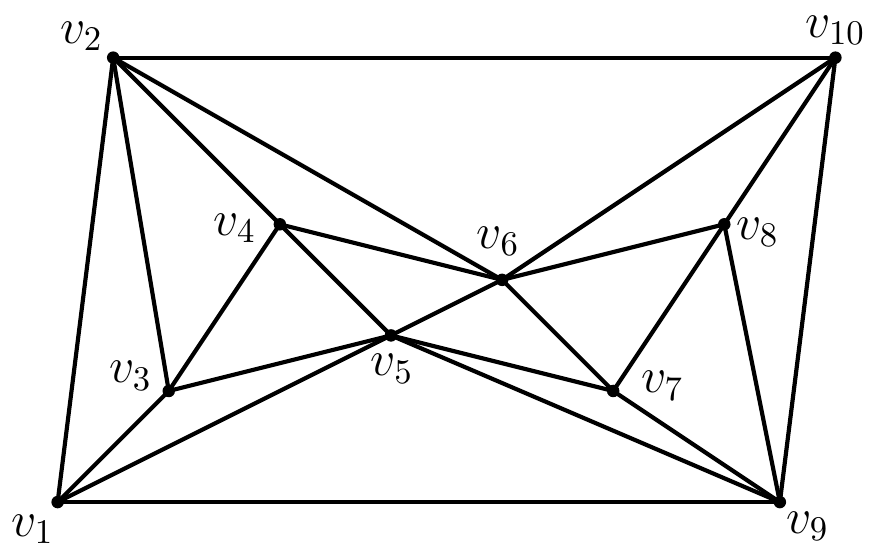}
\hspace{.5in}
\includegraphics[width=0.41\textwidth]{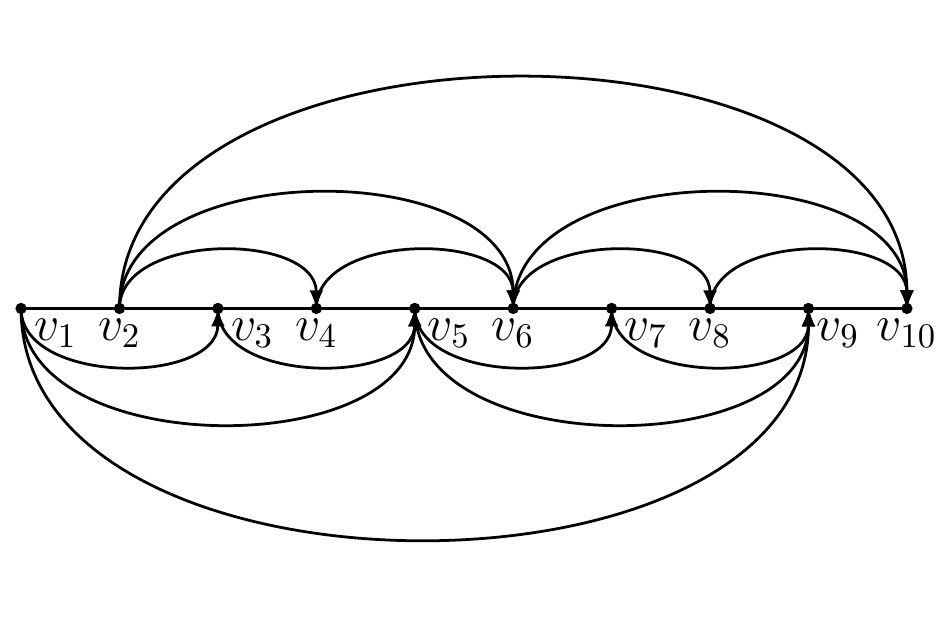}
\caption{Left: a graph on $n=2^\ell+2$ vertices (here $\ell=3$), that
  contains $\Omega(1.7003^n)$ monotone paths, for $n$ sufficiently large.
Right: an isomorphic plane monotone graph where corresponding vertices are in
 the same order by $x$-coordinate; and edges above (resp., below) $\xi_0$ remain
 above (resp., below) $\xi_0$.}
\label{fig:lb-constr}
\end{figure}

Every $n$-vertex triangulation contains $\Omega(n^2)$ monotone paths,
since there is a monotone path between any two vertices (by a straightforward
adaptation of~\cite[Lemma~1]{DRT12} from convex subdivisions to triangulations).
The \emph{minimum} number of monotone paths in an $n$-vertex
triangulation lies between $\Omega(n^2)$ and $O(n^{3.17})$~\cite{DLST16}.

The number of crossing-free structures (matchings, spanning trees,
spanning cycles, triangulations) on a set of $n$ points in the plane
is known to be exponential~\cite{ACNS82,DSST13,GNT00,RSW08,SS11,SS13,SSW12,SW06a};
see also~\cite{DT12,She14}.
Early upper bounds in this area were obtained by multiplying an upper
bound on the maximum number of triangulations on $n$ points
with an upper bound on the maximum number of desired configurations in
an $n$-vertex triangulation; valid upper bounds result since
every plane geometric graph can be augmented into a triangulation.

The efficiency of the simplex algorithms and its variants hinges on
extremal bounds on the \emph{length} of a monotone paths in the 1-skeleton
of a polytope in $\mathbb{R}^d$. For example, the \emph{monotone Hirsch conjecture}~\cite{Zi94}
states that for every $\mathbf{u} \in \mathbb{R}^d\setminus \{\mathbf{0}\}$,
the $1$-skeleton of every $d$-dimensional polytope with $n$ facets contains
a $\mathbf{u}$-monotone path with at most $n-d$ edges from any vertex to
a $\mathbf{u}$-maximal vertex. Klee~\cite{Kle65} verified the conjecture
for $3$-dimensional polytopes, but counterexamples have been found
in dimensions $d\geq 4$~\cite{Tod80} (see also~\cite{San12}).
Kalai~\cite{Kal92,Kal04} gave a subexponential upper bound for
the length of a \emph{shortest} monotone path between any two vertices.
However, even in $\RR^3$, no deterministic pivot rule is known
to find a monotone path of length $n-3$~\cite{KMSZ05}, and the
expected length of a path found by randomized pivot rules requires
averaging over all $\mathbf{u}$-monotone paths~\cite{GK07,MS06}.
See also~\cite{San13} for a summary of results of the \texttt{polymath~3} project
on the \emph{polynomial Hirsch conjecture}.

\section{Preliminaries} \label{sec:prelim}

A polygonal path $\xi=(v_1,v_2,\ldots, v_t)$ in $\mathbb{R}^d$ is
\emph{monotone in direction $\mathbf{u}\in \mathbb{R}^d\setminus \{\mathbf{0}\}$}
if every directed edge of $\xi$ has a positive inner product with $\mathbf{u}$, that is,
$\langle \overrightarrow{v_iv_{i+1}},\mathbf{u}\rangle>0$ for $i=1,\ldots, t-1$;
here $\mathbf{0}$ is the origin. A path $\xi=(v_1,v_2,\ldots, v_t)$ is \emph{monotone}
if it is monotone in some direction $\mathbf{u}\in \mathbb{R}^d\setminus \{\mathbf{0}\}$.
A path $\xi$ in the plane is \emph{$x$-monotone}, if it is monotone with respect to the
positive direction of the $x$-axis, \ie, monotone in direction $\mathbf{u}=(1,0)$.

Let $S$ be a set of $n$ points in the plane. A (geometric) \emph{triangulation} of $S$
is a plane geometric graph with vertex set $S$ such that the bounded
faces are triangles that jointly tile of the convex hull of $S$. Since
a triangulation has at most $3n-6$ edges for $n\geq 3$, and the
$\mathbf{u}$-monotonicity of an edge $(a,b)$ depends on the sign of
$\langle \overrightarrow{ab},\mathbf{u}\rangle$, it is enough to
consider monotone paths in at most $2(3n-6)=6n-12$ directions (one direction
between any two consecutive unit normal vectors of the edges). In the remainder of the paper,
we fix a direction and obtain an upper on the number of monotone paths in that direction.
We may assume that this is the $x$-axis after a suitable rotation.

Let $G=(S,E)$ be a plane geometric graph with $n$ vertices. An $x$-monotone path
$\xi=(v_1,v_2,\ldots , v_t)$ in $G$ is \emph{maximal} if it is not a proper subpath of
any other $x$-monotone path in $G$. Every $x$-monotone path in $G$ contains at most $n$
vertices, hence it contains at most ${n \choose 2}$ $x$-monotone subpaths.
Conversely, every $x$-monotone paths can be extended to a maximal $x$-monotone path.
Let $\lambda_n$ (resp., $\mu_n$) denote the maximum number of monotone paths
(resp., maximal $x$-monotone paths) in an $n$-vertex triangulation.
As such, we have
$$ \lambda_n \leq (6n-12){n\choose 2}\cdot \mu_n = O(n^3 \mu_n). $$

We prove an upper bound for a broader class of graphs, \emph{plane monotone graphs},
in which every edge is an $x$-monotone Jordan arc.
Consider a plane monotone graph $G$ on $n$ vertices with a maximum number of $x$-monotone paths.
We may assume that the vertices have distinct $x$-coordinates;
otherwise we can perturb the vertices without decreasing the number of
$x$-monotone paths. Since inserting new edges can only increase the number
of $x$-monotone paths, we may also assume that $G$ is fully triangulated~\cite[Lemma~3.1]{PT11},
\ie, it is an edge-maximal planar graph. Conversely, every plane monotone
graph is isomorphic to a plane geometric graph in which the $x$-coordinates
of the corresponding vertices are the same~\cite[Theorem~2]{PT11}. Consequently,
the number of maximal $x$-monotone paths in $G$ equals $\mu_n$.

Denote the vertex set of $G$ by $W=\{w_1,w_2,\ldots,w_n\}$, ordered by increasing
$x$-coordinates; and direct each edge $w_i w_j \in E(G)$ from $w_i$ to $w_j$ if $i<j$;
we thereby obtain a directed graph $G$.
By~\cite[Lemma~3]{DLST16}, all edges $w_i w_{i+1}$ must be present, \ie,
$G$ contains a Hamiltonian path $\xi_0 = (w_1,w_2,\ldots,w_n)$.
If $T(i)$ denotes the number of maximal (\wrt inclusion) 
$x$-monotone paths in such a maximizing graph starting at vertex $w_{n-i+1}$,
then $T(i)$ satisfies the recurrence \linebreak
$T(i) \leq T(i - 1) + T(i - 2) + T(i - 3)$ for $i\geq 4$,
with initial values $T(1) = T(2) = 1$ and $T(3) = 2$
(one-vertex paths are also counted).
Its solution is $T(n) = O(\alpha^n)$, where $\alpha = 1.8392\ldots$ is
the unique real root of the cubic equation $x^3-x^2-x-1=0$.
Consequently, any $n$-vertex geometric triangulation admits
at most $O(n^3 \, T(n))=O(1.8393^n)$ monotone paths.
Theorem~\ref{thm:upper} improves this bound to $O(1.7864^n)$.

\paragraph{Fingerprinting technique.}
An $x$-monotone path can be represented uniquely by the subset of
visited vertices. This unique representation gives the trivial
upper bound of $2^n$ for the number of $x$-monotone paths.
For a set of $k$ vertices $V \subseteq W$, an \emph{incidence
pattern} of $V$ (\emph{pattern}, for short) is a subset of $V$ that
appears in a monotone path $\xi$ (\ie, the intersection between $V$
and a monotone path $\xi$).
Denote by $I(V)$ the set of all incidence patterns of $V$;
see Fig.~\ref{fig:f42}. For instance, $v_1 v_3\in I(V)$ implies
that there exists a monotone path $\xi$ in $G$ that is incident
to $v_1$ and $v_3$ in $V$, but no other vertices in $V$. The
incidence pattern $\emptyset\in I(V)$ denotes an empty intersection
between $\xi$ and $V$, \ie, a monotone path that has no vertices in $V$.

We now describe a \emph{divide \& conquer} application of the fingerprinting
technique we use in our proof.
For $k\in \mathbb{N}$, let $p_k=\max_{|V|=k} |I(V)|$ denote the maximum
number of incidence patterns for a set $V$ of $k$ consecutive vertices
in a plane monotone triangulation. We trivially have $p_k \leq 2^k$, and
it immediately follows from the definition that $p_k \leq p_i p_j$ for all $i,j \geq 1$
with $i+j=k$; in particular, we have $p_{2k} \leq p_k^2$.
Assuming that $n$ is a multiple of $k$, the product rule yields $\mu_n\leq p_k^{n/k}$.
For arbitrary $n$ and constant $k$, we obtain
$\mu_n \leq p_k^{\lfloor n/k\rfloor} 2^{n-k\lfloor n/k\rfloor}
\leq p_k^{\lfloor n/k\rfloor} 2^k = O\left(p_k^{n/k}\right)$.
Table~\ref{tab:1} summarizes the upper bounds obtained by this approach.

\renewcommand*{\arraystretch}{1.1}
\begin{table}[hbtp]
\begin{center}
\begin{tabular}{|l|r|l|l|}
\hline
$k$ & $p_k$ & $\mu_n = O\left(p_k^{n/k}\right)$ &
$\lambda_n = O(n^3 \mu_n)$ \\
\hline\hline
$2$ & $4$ & $2^n$ & $O(n^3 \, 2^n)$ \\
\hline
$3$ & $7$ & $O(7^{n/3})$ & $O(n^3 \, 7^{n/3}) = O(1.913^n)$\\
\hline
$4$ & $13$ & $O(13^{n/4})$ & $O(n^3 \, 13^{n/4}) = O(1.8989^n)$\\
\hline
$5$ & $23$ & $O(23^{n/5})$  & $O(n^3 \, 23^{n/5}) = O(1.8722^n)$\\
\hline
$6$ & $41$ & $O(41^{n/6})$ & $O(n^3 \, 41^{n/6}) = O(1.8570^n)$\\
\hline
$7$ & $70$ & $O(70^{n/7})$ & $O(n^3 \, 70^{n/7}) = O(1.8348^n)$\\
\hline
$8$ & $120$ & $O(120^{n/8})$ & $O(n^3 \, 120^{n/8}) = O(1.8193^n)$\\
\hline
$9$ & $201$ & $O(201^{n/9})$ & $O(n^3 \, 201^{n/9}) = O(1.8027^n)$\\
\hline
$10$ & $346$ & $O(346^{n/10})$ & $O(n^3 \, 346^{n/10}) = O(1.7944^n)$\\
\hline
$11$ & $591$ & $O(591^{n/11})$ & $O(n^3 \, 591^{n/11}) = O(1.7864^n)$\\
\hline
\end{tabular}
\caption{Upper bounds obtained via the fingerprinting technique for $k \leq 11$.}
\label{tab:1}
\end{center}
\end{table}
\renewcommand*{\arraystretch}{1}

\vspace{-\baselineskip}
It is clear that $p_1=2$ and $p_2=4$, and it is not difficult to see that $p_3=7$
(note that $p_3<p_1 p_2$). We prove $p_4=13$ (and so $p_4<p_2^2 = 16$)
by analytic methods (Section~\ref{sec:groups-of-4});
this yields the upper bounds $\mu_n = O(13^{n/4})$ and
consequently $\lambda_n = O(n^3 13^{n/4}) = O(1.8989^n)$.
A careful analysis of the edges between two consecutive groups of $4$ vertices
shows that $p_8=120$, and so $p_8$ is significantly smaller than $p_4^2 = 13^2 = 169$
(Lemma~\ref{lem:summary-group-of-8}), hence $\mu_n = O(120^{n/8})$ and
$\lambda_n = O(n^3 120^{n/8}) = O(1.8193^n)$.
Computer search shows that $p_{11}=591$, and so
$\mu_n = O(591^{n/11})$ and $\lambda_n = O(n^3 591^{n/11}) = O(1.7864^n)$
(Section~\ref{sec:groups-of-9-10}).

The analysis of $p_k$, for $k \geq 12$, using the same technique
is expected to yield further improvements. Handling incidence patterns on $12$ or $13$ vertices
is still realistic (although time consuming), but working with larger groups is currently
prohibitive, both by analytic methods and with computer search.
Significant improvement over our results may require new ideas.

\paragraph{Definitions and notations for a single group.}
Let $G$ be a directed plane monotone triangulation that contains a Hamiltonian path
$\xi_0=(w_1,w_2,\ldots , w_n)$. Denote by $G^-$ (resp., $G^+)$ the path $\xi_0$ together
with all edges below (resp., above) $\xi_0$.
Let $V=\{v_1,\ldots , v_k\}$ be a set of $k$ consecutive vertices of $\xi_0$.
For the purpose of identifying the edges relevant for the incidence patterns of $V$,
the edges between a vertex $v_i\in V$ and any vertex preceding $V$
(resp., succeeding $V$) are equivalent.
We therefore apply a graph homomorphism $\varphi$ on $G^-$ and $G^+$,
respectively, that maps all vertices preceding $V$ to a new node $v_0$,
and all vertices succeeding $V$ to a new node $v_{k+1}$.
The path $\xi_0$ is mapped to a new path $(v_0,v_1,\ldots,v_k,v_{k+1})$.
Denote the edges in $\varphi(G^- \setminus \xi_0)$ and $\varphi(G^+ \setminus \xi_0)$,
respectively, by $E^-(V)$ and $E^+(V)$; they are referred to as
the \emph{upper side} and the \emph{lower side}; and let $E(V)=E^-(V) \cup E^+(V)$.
The incidence pattern of the vertex set $V$ is determined by the triple $(V,E^-(V),E^+(V))$.
We call this triple the \emph{group} induced by $V$, or simply the \emph{group} $V$.

\begin{figure}[htbp]
 \centering
 \includegraphics[scale=0.9]{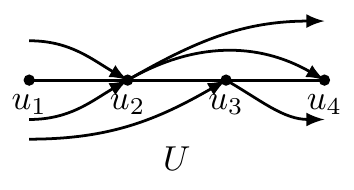}
 \hspace{2.5cm}
 \includegraphics[scale=0.9]{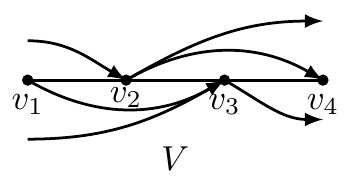}
\caption{Left: a group $U$ with incidence patterns
$I(U)=\{\emptyset$, $u_1 u_2$, $u_1 u_2 u_3$, $u_1 u_2 u_3 u_4$, $u_1 u_2 u_4$,
$u_2, u_2 u_3$, $u_2 u_3 u_4$, $u_2 u_4$, $u_3, u_3 u_4\}$.
Right: a group $V$ with
$I(V)=\{\emptyset$, $v_1 v_2$, $v_1 v_2 v_3$, $v_1 v_2 v_3 v_4$, $v_1 v_2 v_4$,
$v_1 v_3$, $v_1 v_3 v_4$, $v_2$, $v_2 v_3$, $v_2 v_3 v_4$, $v_2 v_4$, $v_3$, $v_3 v_4\}$.}
\label{fig:f42}
\end{figure}

The edges $v_iv_j\in E(V)$, $1\leq i<j\leq k$, are called \emph{inner edges}.
The edges $v_0v_i$, $1\leq i\leq k$, are called \emph{incoming edges} of $v_i\in V$;
and the edges $v_iv_{k+1}$, $1\leq i\leq k$, are \emph{outgoing edges} of $v_i\in V$
(note that $v_0$ and $v_{k+1}$ are not in $V$).
An incoming edge $v_0v_i$ for $1<i\leq k$ (resp., and outgoing edge
$v_iv_{k+1}$ for $1\leq i<k$) may be present in both $E^-(V)$ and
$E^+(V)$. Denote by $\In(v)$ and $\Out(v)$, respectively, the number
of incoming and outgoing edges of a vertex $v\in V$; and note that
$\In(v)$ and $\Out(v)$ can be $0$, $1$ or $2$.

For $1 \leq i \leq k$, let $V_{*i}$ denote the set of incidence patterns in the group $V$
ending at $i$. For example in Fig.~\ref{fig:f42}\,(right),
$V_{*3} = \{ v_1 v_2 v_3, v_1 v_3, v_2 v_3, v_3 \}$.
By definition we have $|V_{*i}| \leq 2^{i-1}$.
Similarly $V_{i*}$ denotes the set of incidence patterns in the group $V$
starting at $i$. In  Fig.~\ref{fig:f42}\,(left), $U_{2*} = \{ u_2, u_2 u_3, u_2 u_3 u_4, u_2 u_4 \}$.
Observe that $|V_{i*}| \leq 2^{k-i}$. Note that
\begin{equation}\label{eq:outpatterns}
|I(V)| = 1+\sum_{i=1}^k|V_{*i}|
\hspace{1cm}\mbox{ \rm and }\hspace{1cm}
|I(V)| = 1+\sum_{i=1}^k|V_{i*}|.
\end{equation}
Reflecting all components of a triple $(V,E^-(V),E^+(V))$ with respect
to the $x$-axis generates a new group denoted by $V^R$. By definition, both $V$ and $V^R$
have the same set of incidence patterns.

\paragraph{Remark.}
Our counting arguments pertain to maximal $x$-monotone paths.
Suppose that a maximal $x$-monotone path $\xi$ has an incidence pattern in $V_{*i}$,
for some $1 \leq i<k$.
By the maximality of $\xi$, $\xi$ must leave the group after $v_i$,
and so $v_i$ must be incident to an outgoing edge.
Similarly, the existence of a pattern in $V_{i*}$ for $1<i\leq k$, implies
that $v_i$ is incident to an incoming edge.

\section{Groups of 4 vertices} \label{sec:groups-of-4}

In this section we analyze the incidence patterns of groups with 4 vertices.
We prove that $p_4= 13$ and find the only two groups with $4$ vertices that have $13$ patterns
(Lemma~\ref{lem:group-with-at-least-13}).
We also prove important properties of groups that have exactly $11$ or $12$ patterns, respectively
(Lemmata~\ref{lem:group-with-at-least-11}, \ref{lem:group-with-exactly-11}
and~\ref{lem:group-with-at-least-12}).

\begin{lemma}\label{lem:group-with-at-least-10}
Let $V$ be a group of $4$ vertices with at least $10$ incidence patterns.
Then there is
\begin{itemize} \itemsep 0pt
\item[{\rm (i)}] an outgoing edge from $v_2$ or $v_3$; and
\item[{\rm (ii)}] an incoming edge into $v_2$ or $v_3$.
\end{itemize}
\end{lemma}
\begin{proof}
\begin{figure}[htbp]
\centering
 \includegraphics[scale=0.9]{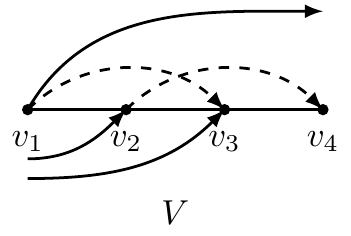}
 \caption{$v_1$ cannot be the last vertex with an outgoing edge from
 a group $V=\{v_1,v_2,v_3,v_4\}$ with at least $10$ incidence patterns.}
 \label{fig:f10}
\end{figure}
(i) There is at least one outgoing edge from $\{v_1, v_2, v_3\}$,
since otherwise $V_{*1} = V_{*2}  = V_{*3} = \emptyset$ implying $|I(V)| = |V_{*4}| + 1 \leq 9$.
Assume there is no outgoing edge from $v_2$ and $v_3$;
then $V_{*1} = \{v_1\}$ and $V_{*2} = V_{*3} = \emptyset$.
From~\eqref{eq:outpatterns}, we have $|V_{*4}| = 8$ and this implies
$\{ v_1 v_3 v_4, v_2 v_4, v_3 v_4\} \subset V_{*4}$.
The patterns $v_1 v_3 v_4$ and $v_2 v_4$, respectively, imply that
$v_1 v_3, v_2 v_4 \in E(V)$. The patterns $v_2 v_4$ and $v_3 v_4$,
respectively, imply there are incoming edges into $v_2$ and $v_3$.
Refer to Fig.~\ref{fig:f10}. Without loss of generality, an
outgoing edge from $v_1$ is in $E^+(V)$. By planarity,
an incoming edges into $v_2$ and $v_3$ have to be in $E^-(V)$.
Then $v_1 v_3$ and $v_2 v_4$ both have to be in $E^+(V)$ which by planarity is impossible.

\smallskip\noindent
(ii) By symmetry in a vertical axis, there is an incoming edge into $v_2$ or $v_3$.
\end{proof}

\begin{lemma}\label{lem:group-with-at-least-11}
Let $V$ be a group of $4$ vertices with at least $11$ incidence patterns.
Then there is
\begin{itemize} \itemsep 0pt
	\item[{\rm (i)}] an incoming edge into $v_2$; and
	\item[{\rm (ii)}] an outgoing edge from $v_3$.
\end{itemize}
\end{lemma}
\begin{proof}
(i) Assume $\In(v_2)=0$. Hence $|V_{2*}| = 0$. By
Lemma~\ref{lem:group-with-at-least-10}\,(ii), we have $\In(v_3)>0$.
By definition $|V_{3*}| \leq 2$. We distinguish two cases.

\medskip
\emph{Case 1: $\In(v_4)=0$.} In this case, $|V_{4*}| = 0$.
Refer to Fig.~\ref{fig:f14}\,(left). By planarity, the edge $v_1 v_4$
and an outgoing edge from $v_2$ cannot coexist with an incoming edge into
$v_3$. So either $v_1 v_4$ or $v_1 v_2$ is not in $V_{1*}$,
which implies $|V_{1*}| < 8$. Therefore, \eqref{eq:outpatterns} yields
$|I(V)| = |V_{1*}| + |V_{3*}| + 1 < 8 + 2 + 1 = 11$, which is a contradiction.

\emph{Case 2: $\In(v_4)>0$.} In this case, $|V_{4*}| = 1$.
If the incoming edges into $v_3$ and $v_4$ are on opposite sides
(see Fig.~\ref{fig:f14}\,(center)), then by planarity there are outgoing
edges from neither $v_1$ nor $v_2$, which implies that the patterns
$v_1$ and $v_1 v_2$ are not in $V_{1*}$, and so $|V_{1*}| \leq 8 - 2 = 6$.
If the incoming edges into $v_3$ and $v_4$ are on the same side
(see Fig.~\ref{fig:f14}\,(right)), then by planarity either the edges
$v_1 v_4$ and $v_2 v_4$ or an outgoing edge from $v_3$ cannot exist,
which implies that either $v_1 v_4$ and $v_1 v_2 v_4$ are not in $V_{1*}$
or $v_1 v_3$ and $v_1 v_2 v_3$ are not in $V_{1*}$. In either case,
$|V_{1*}| \leq 8 - 2 = 6$.

\begin{figure}[htbp]
 \centering
 \includegraphics[scale=0.9]{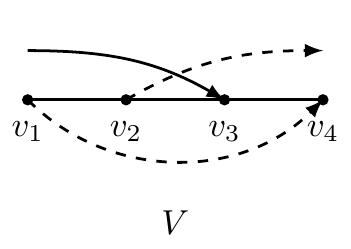}
 \hspace{1cm}
 \includegraphics[scale=0.9]{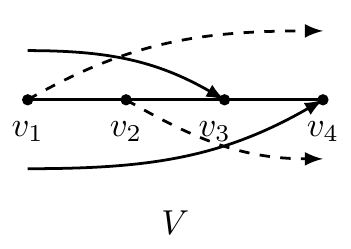}
 \hspace{1cm}
 \includegraphics[scale=0.9]{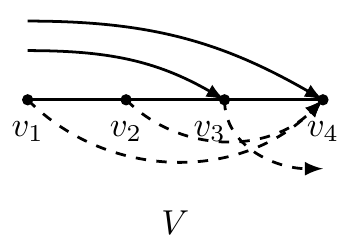}
 \caption{Left: an incoming edge arrives into $v_3$, but not into $v_4$.
   Center and right: incoming edges arrive into both $v_3$ and $v_4$; either
   on the same or on opposite sides of $\xi_0$.}
 \label{fig:f14}
\end{figure}

Therefore, irrespective of the relative position of the incoming edges
into $v_3$ and $v_4$, \eqref{eq:outpatterns} yields
$|I(V)| = |V_{1*}| + |V_{3*}| + |V_{4*}| + 1 \leq 6 + 2 + 1 + 1 = 10$,
which is a contradiction.

\smallskip\noindent
(ii) By symmetry in a vertical axis, $\Out(v_3)>0$.
\end{proof}

\begin{lemma}\label{lem:group-with-exactly-11}
Let $V$ be a group of $4$ vertices with exactly $11$ incidence patterns.
Then the following hold.
\begin{itemize} \itemsep 0pt
\item[{\rm (i)}] If $\In(v_3)=0$, then all the incoming edges into
$v_2$ are on the same side, $|V_{1*}| \geq 5$ and $|V_{2*}| \geq 3$.
\item[{\rm (ii)}] If $\In(v_3)>0$, then
all the incoming edges into $v_3$ are on the same side,
$|V_{1*}| \geq 4$, $|V_{2*}| \geq 2$ and $|V_{3*}| = 2$.
\end{itemize}
\end{lemma}
\begin{proof}
By Lemma~\ref{lem:group-with-at-least-11}, $\In(v_2) \neq 0$ and
$\Out(v_3) \neq 0$. Therefore $\{ v_2 v_3, v_2 v_3 v_4 \} \subseteq V_{2*}$,
implying $|V_{2*}| \geq 2$. By definition $|V_{4*}| \leq 1$.

\medskip
(i) Assume $\In(v_3)=0$. Then we have $|V_{3*}| = 0$. Using~\eqref{eq:outpatterns},
$|V_{1*}|+|V_{2*}| \geq 9$. By definition $|V_{2*}| \leq 4$, implying
$|V_{1*}| \geq 5$. All incoming edges into $v_2$ are on the same
side, otherwise the patterns $\{ v_1, v_1 v_3, v_1 v_3 v_4, v_1 v_4 \}$
cannot exist, which would imply $|V_{1*}| < 5$. If $|V_{2*}| < 3$, then
$v_2$ and $v_2 v_4$ are not in $V_{2*}$ implying that $v_1 v_2$ and $v_1 v_2 v_4$
are not in $V_{1*}$; hence $|V_{1*}| \leq 6$ and thus
$|V_{1*}|+|V_{2*}| < 9$, which is a contradiction.
We conclude that $|V_{2*}| \geq 3$.

\medskip
(ii) Assume $\In(v_3)>0$.
Then we have $\{ v_3, v_3 v_4 \} \subseteq V_{3*}$, hence $|V_{3*}| = 2$.
By~\eqref{eq:outpatterns}, we obtain $|V_{1*}| + |V_{2*}| \geq 7$.
If $|V_{1*}| < 4$, then $|V_{2*}| \geq 4$ and so
$\{ v_2, v_2 v_3, v_2 v_4, v_2 v_3 v_4 \} \subseteq V_{2*}$.
This implies $\{ v_1 v_2, v_1 v_2 v_3, v_1 v_2 v_4,
v_1 v_2 v_3 v_4 \} \subseteq V_{1*}$,
hence $|V_{1*}| \geq 4$ and $|V_{1*}| + |V_{2*}| \geq 4+4 = 8$
which is a contradiction. We conclude $|V_{1*}| \geq 4$. All incoming
edges into $v_3$ are on the same side, otherwise the patterns
$\{ v_1, v_1 v_2, v_1 v_2 v_4, v_1 v_4, v_2, v_2 v_4 \}$
cannot exist, and thus $|I(V)| \leq 10$, which is a contradiction.
\end{proof}

\begin{lemma}\label{lem:group-with-at-least-12}
Let $V$ be a group of $4$ vertices with exactly $12$ incidence patterns.
Then the following hold.
\begin{itemize} \itemsep 0pt
\item[{\rm (i)}] For $i=1,2,3$, all outgoing edges from $v_i$,
if any, are on one side of $\xi_0$.

\item[{\rm (ii)}] If $V$ has outgoing edges from exactly one vertex, then
this vertex is $v_3$ and we have $|V_{*3}| = 4$ and $|V_{*4}| = 7$.
Otherwise there are outgoing edges from $v_2$ and $v_3$, and we have
$|V_{*2}| = 2$, $|V_{*3}| \geq 3$ and $|V_{*4}| \geq 5$.

\item[{\rm (iii)}] For $i=2,3,4$, all incoming edges into $v_i$,
if any, are on one side of $\xi_0$.

\item[{\rm (iv)}] If $V$ has incoming edges into exactly one vertex, then
this vertex is $v_2$ and we have $|V_{2*}| = 4$ and $|V_{1*}| = 7$.
Otherwise there are incoming edges into $v_3$ and $v_2$, and we have
$|V_{3*}| = 2$, $|V_{2*}| \geq 3$ and $|V_{1*}| \geq 5$.
\end{itemize}
\end{lemma}
\begin{proof}
(i) By Lemma~\ref{lem:group-with-at-least-11}\,(i), there is an incoming edge into
$v_2$. So by planarity, all outgoing edges from $v_1$, if any, are on one side of $\xi_0$.

If there are outgoing edges from $v_2$ on both sides, then by planarity the edges $v_1 v_3$,
$v_1 v_4$ and any incoming edge into $v_3$ cannot exist, implying the five
patterns $\{ v_1 v_3, v_1 v_3 v_4, v_1 v_4, v_3, v_3 v_4 \}$ are not in $I(V)$
and thus $|I(V)| \leq 16 - 5 = 11$, which is a contradiction.

If there are outgoing edges from $v_3$ on both sides (see Fig.~\ref{fig:f37}\,(a)),
then by planarity the edges $v_1 v_4$, $v_2 v_4$ and an incoming edge into $v_4$ cannot exist
hence the four patterns $\{ v_1 v_2 v_4, v_1 v_4, v_2 v_4, v_4 \}$ are not in $I(V)$.
Without loss of generality, an incoming edge into $v_2$ is in $E^+(V)$.
Then by planarity, any outgoing edge of $v_1$
and the edge $v_1 v_3$ (which must be present) are in $E^-(V)$.
Then by planarity either an incoming edge into $v_3$ or an outgoing
edge from $v_2$ cannot exist. So either the patterns $\{ v_3, v_3 v_4\}$
or the patterns $\{ v_1 v_2, v_2\}$ are not in $I(V)$. Hence
$|I(V)| \leq 16 - (4+2) = 10$, which is a contradiction.
Consequently, all outgoing edges of $v_i$ are on the same side of $\xi_0$,
for $i = 1,2,3$.

\begin{figure}[htbp]
 \centering
 \includegraphics[scale=0.9]{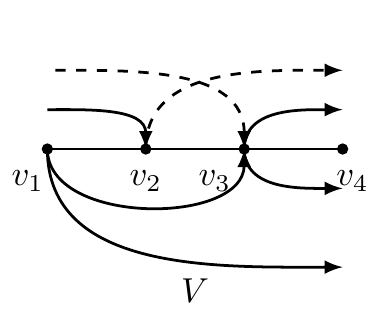}
 \hspace{0.5cm}
 \includegraphics[scale=0.9]{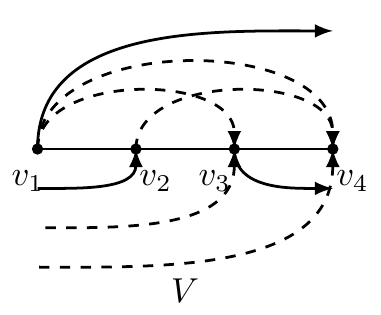}
 \hspace{0.5cm}
 \includegraphics[scale=0.9]{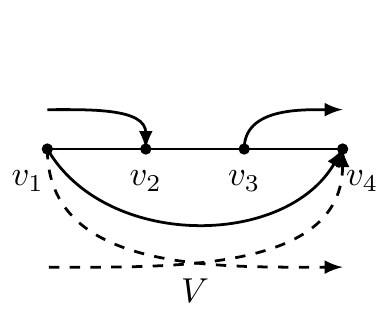}
 \hspace{0.5cm}
 \includegraphics[scale=0.9]{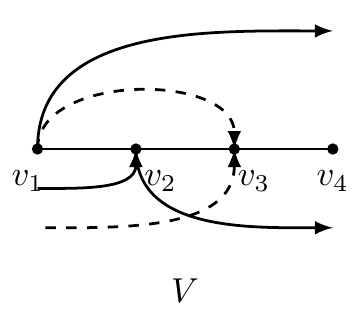}
 \caption{(a) Having outgoing edges from $v_3$ on both sides is impossible.
(b) Existence of outgoing edges only from $\{ v_1 v_3 \}$ is impossible.
(c) $|V_{*3}| \geq 3$.
(d) $|V_{*4}| \geq 5$.}
 \label{fig:f37}
\end{figure}

(ii) If $V$ has outgoing edges from exactly one vertex, then
by Lemma~\ref{lem:group-with-at-least-11}\,(ii), this vertex is $v_3$.
Consequently, $V_{*1} = V_{*2} = \emptyset$. Using~\eqref{eq:outpatterns},
$|V_{*3}| + |V_{*4}| = 11$. Therefore $|V_{*4}| \geq 7$, since by
definition $|V_{*3}| \leq 4$. If $|V_{*4}| = 8$, then we have
$\{ v_1 v_2 v_3 v_4, v_1 v_3 v_4, v_2 v_3 v_4, v_3 v_4 \} \subset V_{*4}$.
Existence of these four patterns along with an outgoing edge from
$v_3$ implies $\{ v_1 v_2 v_3, v_1 v_3, v_2 v_3, v_3 \} \subseteq V_{*3}$
and thus $|V_{*3}| + |V_{*4}| = 4+8 = 12$, which is a contradiction.
Therefore $|V_{*4}| = 7$ and $|V_{*3}| = 4$.

If $V$ has outgoing edges from more than one vertex, the the possible
vertex sets with outgoing edges are $\{ v_1, v_3 \}$,
$\{ v_2, v_3 \}$, and $\{ v_1, v_2, v_3 \}$.
We show that it is impossible that all outgoing edges are from $\{ v_1, v_3 \}$,
which will imply that there are outgoing edges from both $v_2$ and $v_3$.

If there are outgoing edges from $\{ v_1, v_3 \}$ only, we may assume
the ones from $v_1$ are in $E^+(V)$ and then by planarity all incoming
edges into $v_2$ are in $E^-(V)$, see Fig.~\ref{fig:f37}\,(b). Then
by planarity, either $v_1 v_3$ or $v_2 v_4$ or an incoming edge into
$v_3$ cannot exist implying that $\{ v_1 v_3, v_1 v_3 v_4 \}$ or
$\{ v_1 v_2 v_4, v_2 v_4\}$ or $\{ v_3, v_3 v_4 \}$ is not in $I(V)$.
By the same token, depending on the side the outgoing edges from $v_3$
are on, either the edge $v_1 v_4$ or an incoming edge into $v_4$
cannot exist, implying that either $v_1 v_4$ or $v_4$ is not in $I(V)$.
Since $V_{*2} = \emptyset$, $\{ v_1 v_2, v_2\}$ are not in $I(V)$.
So $|I(V)| \leq 16 - (2+1+2) = 11$, which is a contradiction.
Therefore the existence of outgoing edges only from $v_1$
and $v_3$ is impossible.

If there are outgoing edges from only $\{ v_2, v_3 \}$
or only $\{ v_1, v_2, v_3 \}$, then we have
$\{ v_1 v_2, v_2\} \subseteq V_{*2}$ and $\{ v_1 v_2 v_3, v_2 v_3 \} \subseteq V_{*3}$,
since $\In(v_2) \neq 0$ and $\Out(v_3) \neq 0$ by
Lemma~\ref{lem:group-with-at-least-11}. Therefore $|V_{*2}| = 2$
and $|V_{*3}| \geq 2$. If $|V_{*3}| < 3$, then $v_1 v_3, v_3 \notin V_{*3}$,
which implies $v_1 v_3$ and that an incoming edge into $v_3$ are
not in $E(V)$. Consequently, $v_1 v_3 v_4$, $v_3 v_4$ $\notin$ $I(V)$.
Observe Fig.~\ref{fig:f37}\,(c). By planarity the edge $v_1 v_4$, an
incoming edge into $v_4$ and an outgoing edge from $v_1$ cannot exist
together with an incoming edge into $v_2$ and an outgoing edge from
$v_3$. So at least one of the patterns $\{v_1, v_1 v_4, v_4\}$ is
missing implying $ |I(V)| \leq 16 - (2 + 2 + 1) = 11$, which is a
contradiction. So $|V_{*3}| \geq 3$.
If $|V_{*4}| < 5$, then \eqref{eq:outpatterns} yields
$|V_{*3}| = 4$, $|V_{*2}| = 2$ and $|V_{*1}| = 1$. We may assume that all
outgoing edges from $v_1$ are in $E^+(V)$; see Fig.~\ref{fig:f37}\,(d).
By planarity, the incoming edges into $v_2$ are in $E^-(V)$.
Depending on the side the outgoing edges from $v_2$ are on, either
$v_1 v_3$ or an incoming edge into $v_3$ cannot exist, implying
that either $v_1 v_3$ or $v_3$ is not in $V_{*3}$, therefore $|V_{*3}| < 4$,
creating a contradiction. We conclude that $|V_{*4}| \geq 5$.

\medskip
(iii) By symmetry, (iii) immediately follows from (i).

\medskip
(iv) By symmetry, (iv) immediately follows from (ii).
\end{proof}

\begin{lemma}\label{lem:group-with-at-least-13}
Let $V$ be a group of $4$ vertices. Then $V$ has at most $13$ incidence patterns.
If $V$ has 13 incidence patterns, then $V$ is either $A$ or $A^R$ in Fig.~\ref{fig:f1}.
Consequently, $p_4=13$.
\end{lemma}
\begin{figure}[htbp]
 \centering
  \includegraphics[scale=0.9]{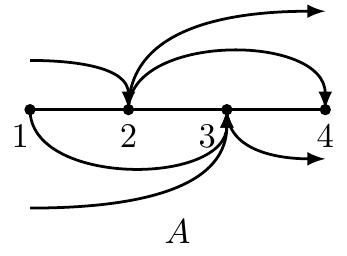}
  \hspace{2cm}
  \includegraphics[scale=0.9]{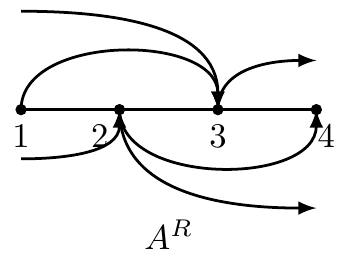}
  \caption{$I(A)=I(A^R)=\emptyset,12,123,1234,124,13,134,2,23,234,24,3,34$.
    $A$ and $A^R$ are the only groups with $13$ incidence patterns.}
 \label{fig:f1}
\end{figure}
\begin{proof}
Observe that group $A$ in Fig.~\ref{fig:f1} has $13$ patterns.
Let $V$ be a group of 4 vertices with at least $13$ patterns.
We first prove that $V$ has an incoming edge into $v_3$ and an outgoing edge from $v_2$.
Their existence combined with Lemma~\ref{lem:group-with-at-least-11}
implies that $\{ v_3 v_4, v_3 \} \subset I(V)$
and $\{ v_1 v_2, v_2\} \subset I(V)$, respectively. At least one of
these two edges has to be in $E(V)$, otherwise $V$ has at most
$16 - (2+2) = 12$ patterns. Assume that one of the two, without
loss of generality, the outgoing edge from $v_2$ is not in $E(V)$.
Then $\{ v_1 v_3, v_2 v_4 \} \subseteq E(V)$, otherwise either
patterns $\{v_1 v_3, v_1 v_3 v_4\}$ or $\{v_1 v_2 v_4, v_2 v_4\}$
are not in $I(V)$ and there are at most $16 - (2 + 2) = 12$ patterns.
By Lemma~\ref{lem:group-with-at-least-11}, there is an incoming
edge into $v_2$ and an outgoing edge from $v_3$. Without loss of
generality, the outgoing edge from $v_3$ is in $E^-(V)$. So by
planarity $v_2 v_4$ is in $E^+(V)$ which implies that $v_1 v_3$ and
the incoming edge into $v_3$ are in $E^-(V)$. By the same token,
the incoming edge into $v_2$ is in $E^+(V)$. So by planarity
the edge $v_1 v_4$ and an outgoing edge from $v_1$ cannot be in $E(V)$.
Then the patterns $\{v_1 v_4$, $v_1\}$ are not in $I(V)$,
thus $V$ has at most $16 - (2 + 2) = 12$ patterns which is a contradiction.

We may assume that the incoming edge into $v_3$ is in $E^-(V)$.
By planarity, the outgoing edge from $v_2$ is in $E^+(V)$.
If the outgoing edge from $v_1$ is in $E(V)$, then by planarity
it has to be in $E^+(V)$, which implies incoming edge into $v_2$
is in $E^-(V)$ and the edge $v_1 v_3$ is not in $E(V)$. Since
outgoing edge from $v_1$ implies only one pattern $v_1$ where
the edge $v_1 v_3$ implies two patterns $\{v_1 v_3, v_1 v_3 v_4\}$,
outgoing edge from $v_1$ cannot be in $E(V)$ but the edge $v_1 v_3$
is in $E^-(V)$. By a similar argument we show that the incoming edge
into $v_4$ cannot be in $E(V)$ and the edge $v_2 v_4$ is in $E^+(V)$.
Therefore $V$ is $A$ and has $13$ patterns.

If the incoming edge into $v_3$ is in $E^+(V)$, then $V$ is $A^R$
(again with $13$ patterns).
\end{proof}

\section{Groups of $8$ vertices} \label{sec:groups-of-8}

In this section, we analyze two consecutive groups, $U$ and $V$,
each with 4 vertices, and show that $p_8=120$ (Lemma~\ref{lem:summary-group-of-8}).
Let $U = \{ u_1, u_2, u_3, u_4 \}$ and $V = \{ v_1, v_2, v_3, v_4 \}$,
and put $UV= U \cup V$ for short. We may assume that $|I(V)| \leq |I(U)|$
(by applying a reflection in the vertical axis if necessary), and
we have $|I(U)| \leq 13$ by Lemma~\ref{lem:group-with-at-least-13}.
This yields a trivial upper bound $|I(UV)| \leq |I(U)| \cdot |I(V)| \leq 13^2=169$.
It is enough to consider cases in which $10 \leq |I(V)| \leq |I(U)|\leq 13$,
otherwise the trivial bound is already less than $120$.

In all cases where $|I(U)|\cdot |I(V)|>120$, we improve on the trivial bound
by finding edges between $U$ and $V$ that cannot be present in the group $UV$.
If edge $u_iv_j$ is not in $E(UV)$, then any of the $|U_{*i}|\cdot |V_{j*}|$
patterns that contain $u_iv_j$ is excluded. Since every maximal $x$-monotone path
has at most one edge between $U$ and $V$, distinct edges $u_iv_j$ exclude
disjoint sets of patterns, and we can use the sum rule to count the excluded patterns.
We continue with a case analysis.

\begin{lemma}\label{lem:V-has-10}
Consider a group $UV$ consisting of two consecutive groups of $4$ vertices,
where $|I(U)| \geq 10$ and $|I(V)|=10$. Then $UV$ allows at most $120$ incidence patterns.
\end{lemma}
\begin{proof}
If $U$ has at most $12$ patterns, then $UV$ has at most $12 \times 10 = 120$ patterns,
and the proof is complete. We may thus assume that $U$ has $13$ patterns.
By Lemma~\ref{lem:group-with-at-least-13}, $U$ is either
$A$ or $A^R$. We may assume, by reflecting $UV$ around the
horizontal axis if necessary, that $U$ is $A$. Refer to Fig.~\ref{fig:f9}\,(left).
Therefore $|U_{*2}| = 2$, $|U_{*3}| = 4$ and $|U_{*4}| = 6$, according
to Figure~\ref{fig:f1}. The cross product of the patterns of $U$ and $V$ produce
$13 \times 10 = 130$ possible patterns. We show that at least $10$ of them are
incompatible in each case. It follows that $|I(UV)|\leq 130 - 10 = 120$.
Let $v_i$ denote the first vertex with an incoming edge in $E(V)$, where $i \neq 1$.
By Lemma~\ref{lem:group-with-at-least-10}\,(ii), $i = 2$ or $3$.

\medskip
\emph{Case 1:} $(u_4,v_i) \in E(UV)$. We first show that $|V_{1*}| \geq 3$.
By definition $|V_{3*}| \leq 2$ and $|V_{4*}| \leq 1$. By~\eqref{eq:outpatterns},
$|V_{1*}| + |V_{2*}| \geq 9-(2+1) = 6$.
If $|V_{2*}| \leq 3$, then $|V_{1*}| \geq 3$. Otherwise $|V_{2*}| = 4$
implying $V_{2*} = \{v_2 v_3 v_4, v_2 v_3, v_2 v_4, v_2\}$.
This implies there are outgoing edges from $v_2$ and $v_3$ in $E(V)$.
Therefore $\{v_1 v_2 v_3 v_4, v_1 v_2 v_3, v_1 v_2\} \subset V_{1*}$
and $|V_{1*}| \geq 3$.

\medskip
\emph{Case 1.1:} $(u_4,v_i) \in  E^-(UV)$; see Fig.~\ref{fig:f9}\,(right).
As $i = 2$ or $3$, by planarity
$(u_3,v_1) \notin E(UV)$. Hence at least $|U_{*3}|\ |V_{1*}| \geq 4 \times 3 = 12$
combinations are incompatible.
\begin{figure}[htbp]
 \centering
 \includegraphics[scale=0.83]{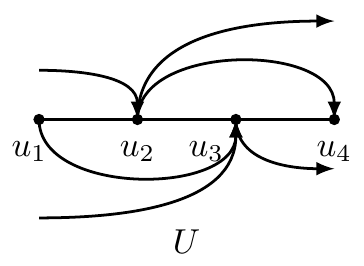}
 \hspace{0.1cm}
 \includegraphics[scale=0.83]{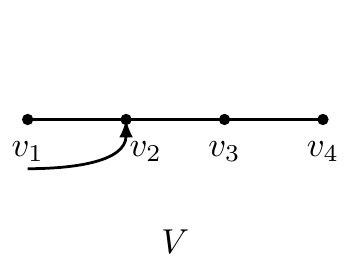}
 \hspace{.1cm}
 \includegraphics[scale=0.83]{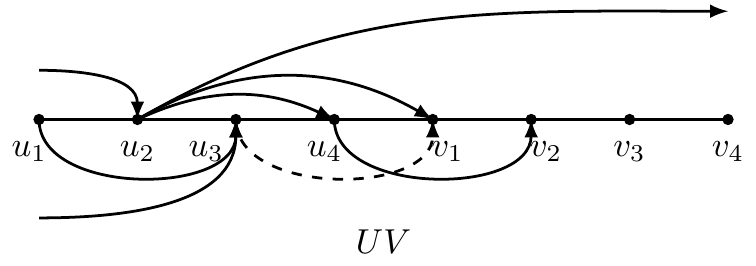}
 \caption{Left: $|I(U)| = 13$ and $|I(V)| = 10$.
 Right: $(u_4,v_i) \in E^-(UV)$; $i = 2$ here.}
 \label{fig:f9}
\end{figure}

\medskip
\emph{Case 1.2:} $(u_4,v_i)  \in  E^+(UV)$; see Fig.~\ref{fig:f12}\,(right).
Then by planarity $(u_2,v_1) \notin E(UV)$ and
$|U_{*2}| |V_{1*}| \geq 2 \times 3 = 6$ combinations are incompatible.
\begin{figure}[htbp]
 \centering
 \includegraphics[scale=0.83]{f1.pdf}
 \hspace{0.1cm}
 \includegraphics[scale=0.83]{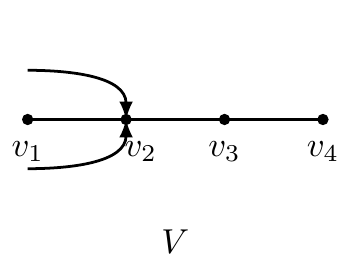}
 \hspace{0.1cm}
 \includegraphics[scale=0.83]{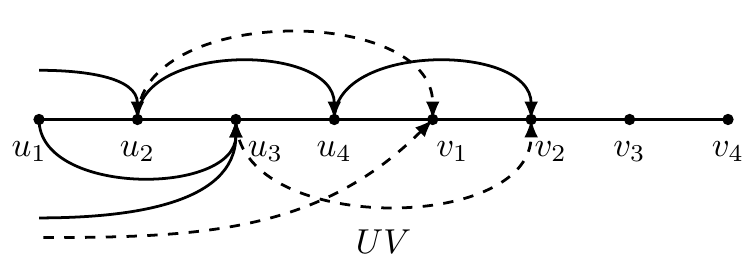}
 \caption{Left: $|I(U)| = 13$ and $|I(V)| = 10$.
   Right: $(u_4,v_i) \in E^+(UV)$; $i = 2$ here.}
 \label{fig:f12}
\end{figure}

An incoming edge into $v_i$ in $E(V)$ implies $|V_{i*}| \geq 1$. If $u_3 v_i \notin E(UV)$,
then $|U_{*3}| |V_{i*}| \geq 4 \times 1 = 4$ combinations are incompatible.
Hence there are at least $6 + 4 = 10$ incompatible patterns.
If $u_3 v_i \in E(UV)$, then by planarity an incoming edge into $v_1$ in $E(UV)$ cannot
exist and $|\{ \emptyset \}| |V_{1*}| \geq 1 \times 3 = 3$ combinations
are incompatible. If $u_2 v_i \in E(UV)$, then by planarity an outgoing
edge from $u_4$ cannot exist. So $|U_{*4}| |\{ \emptyset \}| \geq 6 \times 1 = 6$
combinations are incompatible. So there are at least $6 + 3 + 6 = 15$
incompatible patterns. If $u_2 v_i \notin E(UV)$, then
$|U_{*2}| |V_{1*}| \geq 2 \times 1 = 2$ combinations are incompatible.
Hence there are at least $6 + 3 + 2 = 11$ incompatible patterns.

\medskip
\emph{Case 2:} $(u_4,v_i) \notin E(UV)$.
By showing $|V_{i*}| \geq 2$ for all possible values of $i$ (\ie, $2$ and $3$),
we can conclude that at least $|U_{*4}| |V_{i*}| \geq 6 \times 2 = 12$
combinations are incompatible.

If $i = 2$, then $v_2 v_3 v_4$ $\in$ $V_{2*}$. By
Lemma~\ref{lem:group-with-at-least-10}\,(i), there is an outgoing edge
from $v_2$ or $v_3$ in $E(V)$ which implies $v_2 \in V_{2*}$ or
$v_2 v_3 \in V_{2*}$. Hence $|V_{2*}| \geq 2$.

If $i = 3$ and there is no outgoing edge from $v_3$ in $E(V)$, then by
Lemma~\ref{lem:group-with-at-least-10}\,(i), there is an outgoing edge
from $v_2$. In that case by planarity, there are only $7$ possible incidence
patterns $\{\emptyset, v_1 v_2 v_3 v_4, v_1 v_2 v_4,$ $v_1 v_3 v_4,
v_1 v_2, v_3 v_4, v_4 \}$ in $V$, which is a contradiction. So if
$i = 3$, then there is an outgoing edge from $v_3$ in $E(V)$ which implies
$\{v_3 v_4, v_3\} \subset V_{3*}$ therefore $|V_{3*}| \geq 2$.
\end{proof}

\begin{lemma}\label{lem:V-has-11}
Consider a group $UV$ consisting of two consecutive groups of $4$ vertices,
where $|I(U)|\geq 11$ and $|I(V)|=11$.
Then $UV$ allows at most $120$ incidence patterns.
\end{lemma}
\begin{proof}
We distinguish three cases depending on $|I(U)|$.

\medskip
\emph{Case 1:} $|I(U)| = 11$.
Since $|I(UV)| \leq |I(U)| \cdot |I(V)| = 11 \times 11 = 121$, it
suffices to show that at least one of these patterns is incompatible.
By Lemma~\ref{lem:group-with-at-least-11}, there is an outgoing edge
from $u_3$ in $E(U)$ and an incoming edge into $v_2$ in $E(V)$.
Therefore $u_1 u_2 u_3 \in U_{*3}$ and $v_2 v_3 v_4 \in V_{2*}$.
Refer to Fig.~\ref{fig:f16}\,(left). If $(u_3 v_2) \notin E(UV)$, then
$u_1 u_2 u_3 v_2 v_3 v_4$ is not in $I(UV)$. If $(u_3 v_2) \in E(UV)$,
then by planarity either an outgoing edge from $u_4$ \wrt $UV$, or
an incoming edge into $v_1$ \wrt $UV$, cannot be in $E(UV)$, implying that
either $u_1 u_2 u_3 u_4$ or $v_1 v_2 v_3 v_4$ is not in $I(UV)$.
\begin{figure}[htbp]
 \centering
 \includegraphics[scale=0.8]{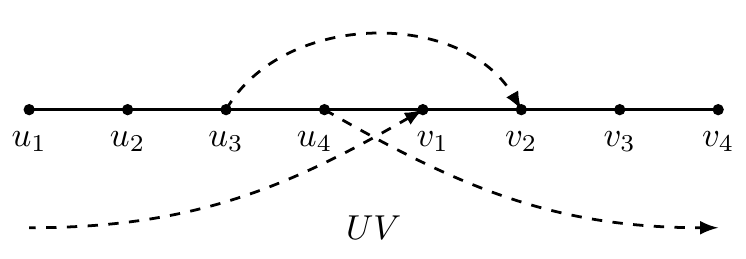}
 \hspace{0.2cm}
 \includegraphics[scale=0.8]{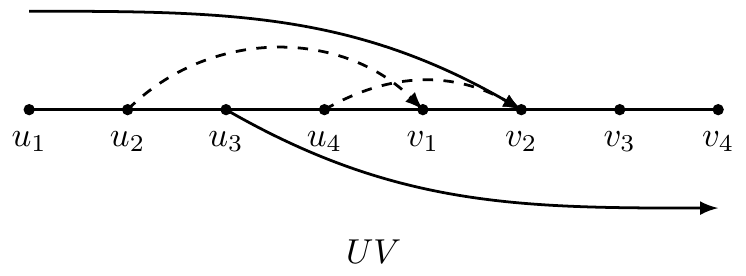}
 \caption{Left: $|I(U)| = |I(V)| = 11$.
 Right: outgoing edge from $u_3$ is in $E^-(UV)$ and outgoing edge
 from $u_2$ is in $E^+(UV)$.}
 \label{fig:f16}
\end{figure}

\medskip
\emph{Case 2:} $|I(U)| = 12$.
By Lemma~\ref{lem:group-with-at-least-12}\,(ii),
if $U$ has outgoing edges from exactly one vertex, then they are from $u_3$
and we have $|U_{*3}| = 4$, $|U_{*4}| = 7$, otherwise $|U_{*3}| \geq 3$
and $|U_{*4}| \geq 5$. By Lemma~\ref{lem:group-with-at-least-12}\,(i),
all the outgoing edges from $u_3$ in $E(U)$ are on one side of $U$.
For simplicity assume those are in $E^-(U)$.
Since $|I(UV)| \leq |I(U)| \cdot |I(V)| = 12 \times 11 = 132$,
it suffices to show that at least $132 - 120 = 12$ of these patterns
are incompatible.

\medskip
\emph{Case 2.1:} There is no incoming edge into $v_3$ in $E(V)$.
Then by Lemma~\ref{lem:group-with-exactly-11}\,(i), all the
incoming edges into $v_2$ in $E(V)$ are on one side of $V$ and we have
$|V_{1*}| \geq 5$ and $|V_{2*}| \geq 3$.

\medskip
\emph{Case 2.1.1:} The incoming edges into $v_2$ \wrt $V$ are in $E^+(V)$.
So by planarity $u_3 v_2 \notin E(UV)$ and at least $|U_{*3}||V_{2*}|$
patterns are incompatible. If $U$ has outgoing edges from exactly
one vertex, then $|U_{*3}||V_{2*}| \geq 4 \times 3 = 12$ and we are done.
Otherwise $U$ has outgoing edges from $u_2$, where $|U_{*2}| = 2$ and at least
$|U_{*3}||V_{2*}| \geq 3 \times 3 = 9$ patterns are incompatible. Also
by Lemma~\ref{lem:group-with-at-least-12}\,(i), all the outgoing edges
from $u_2$ in $E(U)$ are on one side of $U$.
If the outgoing edges from $u_2$ \wrt $U$ are in $E^+(U)$, see
Fig.~\ref{fig:f16}\,(right), then $u_2 v_1$ and $u_4 v_2$ can only be
in $E^+(UV)$; by planarity both edges cannot be in $E(UV)$ and thus at least
$\min(|U_{*2}| |V_{1*}|, |U_{*4}||V_{2*}|) \geq \min(2 \times 5, 5 \times 3) = 10$
patterns are incompatible.
If the outgoing edges from $u_2$ \wrt $U$ are in $E^-(U)$, then by planarity
$u_2 v_2 \notin E(UV)$ and thus at least $|U_{*2}||V_{2*}| \geq 2 \times 3 = 6$
patterns are incompatible.
Therefore irrespective of the relative position of the outgoing
edge from $u_2$ in $E(U)$, at least $9 + \min(10,6) = 15$ patterns
are incompatible and we are done.

\medskip
\emph{Case 2.1.2:} The incoming edges into $v_2$ \wrt $V$ are in $E^-(V)$.
Therefore $u_3 v_1$ and $u_4 v_2$ can only be in $E^-(UV)$. By planarity
both edges cannot be in $E(UV)$. Hence at least \linebreak
$\min(|U_{*3}||V_{1*}|, |U_{*4}||V_{2*}|) = \min(3 \times 5, 5 \times 3) = 15$
patterns are incompatible.

\medskip
\emph{Case 2.2:} There is an incoming edge into $v_3$ in $E(V)$.
By Lemma~\ref{lem:group-with-exactly-11}\,(ii), all the incoming
edges into $v_3$ in $E(V)$ are on one side of $V$, $|V_{1*}| \geq 4$,
$|V_{2*}| \geq 2$ and $|V_{3*}| = 2$.

\medskip
\emph{Case 2.2.1:} The incoming edges into $v_2$ in $E(V)$ are on both sides of $V$.

If the incoming edges into $v_3$ \wrt $V$ are in $E^+(V)$,
see Fig.~\ref{fig:f46}\,(left), then by planarity $u_3 v_3 \notin E(UV)$.
So at least $|U_{*3}||V_{3*}| \geq 3 \times 2 = 6$
patterns are incompatible. By planarity an outgoing edge from
$u_4$ \wrt $UV$, an incoming edge into $v_3$ \wrt $UV$ and $u_3 v_2$ cannot be
in $E(UV)$ together. Therefore at least
$\min (|\{ \emptyset \}||V_{3*}|, |U_{*4}||\{ \emptyset \}|, |U_{*3}||V_{2*}|)
\geq \min (1 \times 2, 5 \times 1, 3 \times 2) =  2$ patterns are
incompatible. By the same argument, the edges $u_3 v_2$, $u_4 v_2$ and an
incoming edge into $v_1$ \wrt $UV$ cannot be in $E(UV)$ together. Hence at least
$ \min (|U_{*3}||V_{2*}|, |U_{*4}||V_{2*}|, |\{ \emptyset \}||V_{1*}|)
= \min (3 \times 2, 5 \times 2, 1 \times 4) = 4$ patterns are
incompatible. Therefore at least $6 + 2 + 4 = 12$ patterns are incompatible.
\begin{figure}[htbp]
 \centering
 \includegraphics[scale=0.83]{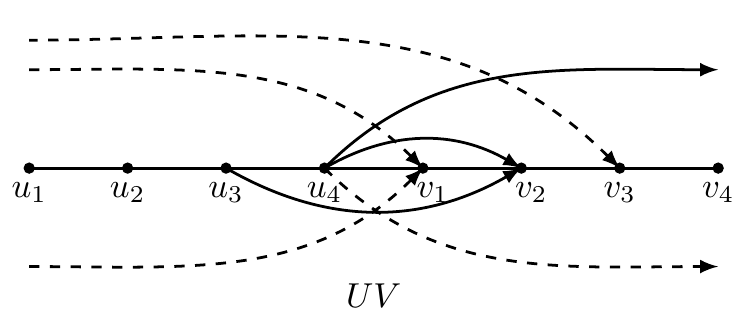}
 \hspace{0.2cm}
 \includegraphics[scale=0.83]{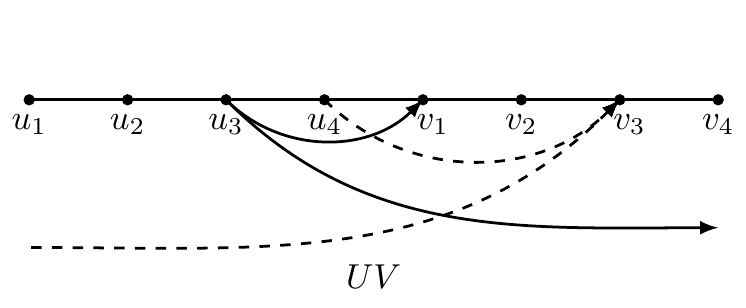}
 \caption{Left: incoming edges into $v_2$ are in both $E^+(V)$ and
 $E^-(V)$ and incoming edge into $v_3$ is in $E^+(V)$.
 Right: incoming edges into $v_2$ are in both $E^+(V)$ and
 $E^-(V)$ and incoming edge into $v_3$ is in $E^-(V)$.}
 \label{fig:f46}
\end{figure}

If incoming edges into $v_3$ \wrt $V$ are in $E^-(V)$, see Fig.~\ref{fig:f46}\,(right),
then either an outgoing edge from $u_3$ \wrt $UV$ or an incoming edge
into $v_3$ \wrt $UV$ cannot be in $E(UV)$. So at least
$ \min (|U_{*3}||\{ \emptyset \}|, |\{ \emptyset \}||V_{3*}|)
\geq \min(3 \times 1, 1 \times 2) = 2$ patterns are incompatible.
Also $u_3 v_1$ and $u_4 v_3$ can only be in $E^-(UV)$. By planarity both
edges cannot be in $E(UV)$. Hence at least $\min (|U_{*3}||V_{1*}|, |U_{*4}||V_{3*}|)
\geq \min (3 \times 4, 5 \times 2) = 10$ patterns are incompatible.
Therefore at least $2 + 10 = 12$ patterns are incompatible.

\medskip
\emph{Case 2.2.2:} All the incoming edges into $v_2$ in $E(V)$ are on one side
of $V$ and the incoming edges into $v_2$ and $v_3$ in $E(V)$ are on same side
of $\xi_0$.

If the incoming edges into $v_2$ and $v_3$ \wrt $V$ are in $E^+(V)$, see
Fig.~\ref{fig:f32}\,(left), then by planarity $u_3 v_2$ and $u_3 v_3$ are
not in $E(UV)$. So at least $|U_{*3}||V_{2*}| + |U_{*3}||V_{3*}| \geq 3 \times 2
+ 3 \times 2 = 12$ patterns are incompatible.
\begin{figure}[htbp]
 \centering
 \includegraphics[scale=0.83]{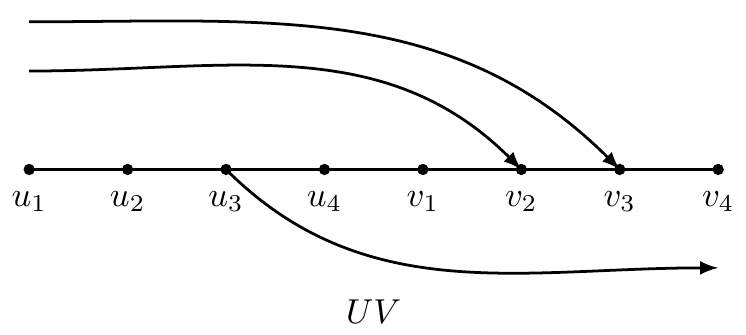}
 \hspace{0.2cm}
 \includegraphics[scale=0.83]{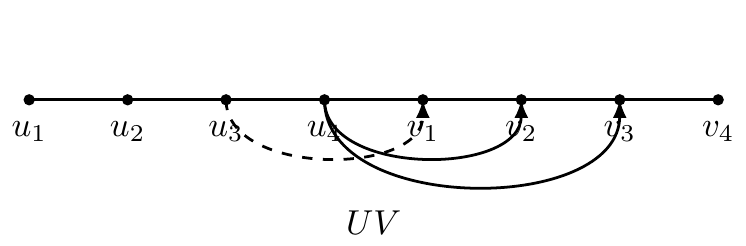}
 \caption{Left: both incoming edges into $v_2$ and $v_3$ are in $E^+(V)$.
 Right: both incoming edges into $v_2$ and $v_3$ are in $E^-(V)$.}
 \label{fig:f32}
\end{figure}

If the incoming edges into $v_2$ and $v_3$ \wrt $V$ are in $E^-(V)$, see
Fig.~\ref{fig:f32}\,(right), then $u_3 v_1$ and both $u_4 v_2$ and $u_4 v_3$
can only be in $E^-(UV)$. By planarity either $u_3 v_1$ or both $u_4 v_2$ and
$u_4 v_3$ are not in $E(UV)$. Consequently, at least
$\min( |U_{*3}||V_{1*}|, |U_{*4}||V_{2*}| + |U_{*4}||V_{3*}| )
= \min( 3 \times 4, 5 \times 2 + 5 \times 2 ) = 12$
patterns are incompatible.

\medskip
\emph{Case 2.2.3:} All the incoming edges into $v_2$ in $E(V)$ are on one
side of $V$ and all the incoming edges into $v_3$ in $E(V)$ are on the
opposite side of $V$.
\begin{figure}[htbp]
 \centering
 \includegraphics[scale=0.83]{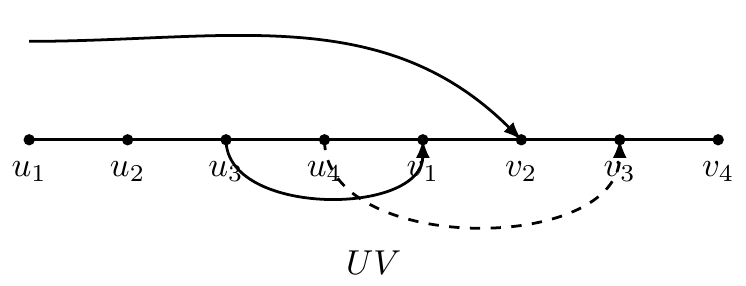}
 \hspace{0.2cm}
 \includegraphics[scale=0.83]{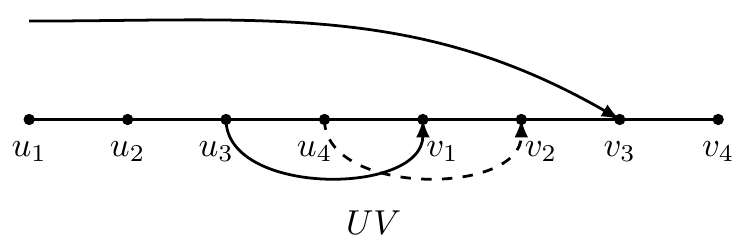}
 \caption{Left: incoming edge into $v_2$ is $E^+(V)$ and incoming
  edge into $v_3$ is in $E^-(V)$.
  Right: incoming edge into $v_2$ is in $E^-(V)$ and incoming
  edge into $v_3$ is in $E^+(V)$.}
 \label{fig:f34}
\end{figure}
Let the incoming edges \wrt $V$ in $E^+(V)$ are into $v_i$ and
the incoming edges \wrt $V$ in $E^-(V)$ are into $v_j$.
So either $i = 2$, $j = 3$ or $i = 3$, $j = 2$, see Fig.~\ref{fig:f34}.
Therefore $|V_{i*}|$, $|V_{j*}|$ are at least $\min(|V_{2*}|, |V_{3*}|) = 2$.
By planarity $u_3 v_i \notin E(UV)$. So at least
$|U_{*3}||V_{i*}| \geq 3 \times 2 = 6$ patterns are incompatible.
Also $u_3 v_1$ and $u_4 v_j$ can only be
in $E^-(UV)$. By planarity both edges cannot be in $E(UV)$. Hence at least
$\min(|U_{*3}||V_{1*}|, |U_{*4}||V_{j*}|) = \min(3 \times 4, 5 \times 2) = 10$
patterns are incompatible. Therefore at least $6 + 10 = 16$ patterns
are incompatible.

\medskip
\emph{Case 3:} $|I(U)| = 13$. By Lemma~\ref{lem:group-with-at-least-13}, $U$ is
either $A$ or $A^R$. We may assume, by reflecting $UV$ around the horizontal axis
if necessary, that $U$ is $A$. Therefore $|U_{*2}| = 2$,
$|U_{*3}| = 4$ and $|U_{*4}| = 6$, see Figure~\ref{fig:f1}.
Since $|I(UV)| \leq |I(U)| \cdot |I(V)| = 13 \times 11 = 143$,
it suffices to show that at least $143 - 120 = 23$ of these patterns
are incompatible.

\medskip
\emph{Case 3.1:} There is no incoming edge into $v_3$ in $E(V)$.
Then by Lemma~\ref{lem:group-with-exactly-11}\,(i), all the incoming
edges into $v_2$ in $E(V)$ are on one side of of $V$, $|V_{1*}| \geq 5$
and $|V_{2*}| \geq 3$.

If the incoming edges into $v_2$ \wrt $V$ are in $E^-(V)$,
see Fig.~\ref{fig:f20}\,(left), by planarity $u_2 v_2 \notin E(UV)$.
So at least $|U_{*2}| |V_{2*}| \geq 2 \times 3 = 6$ patterns are
incompatible. Also $u_4 v_2$ and $u_3 v_1$ can only be in $E^-(UV)$.
By planarity both edges cannot be in $E(UV)$.
Hence at least $\min (|U_{*4}| |V_{2*}|, |U_{*3}| |V_{1*}|)
= \min (6 \times 3, 4 \times 5)= 18$ patterns are incompatible.
Therefore at least $6 + 18 = 24$ patterns are incompatible.
\begin{figure}[htbp]
 \centering
 \includegraphics[scale=0.83]{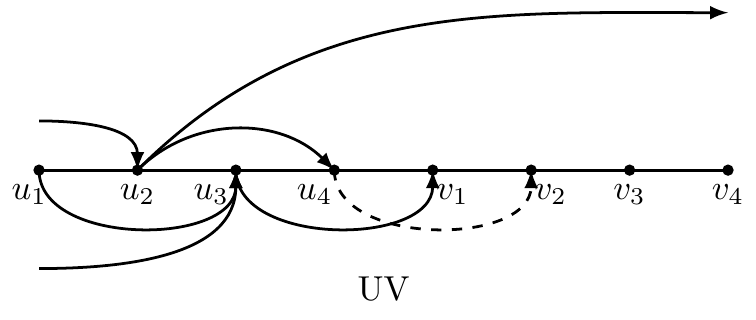}
 \hspace{0.2cm}
 \includegraphics[scale=0.83]{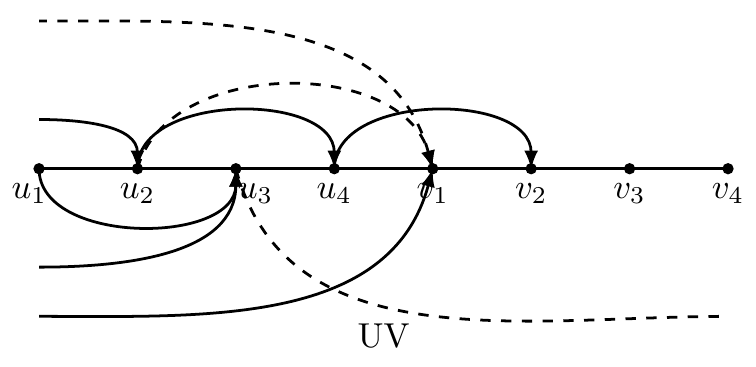}
 \caption{Left: incoming edge into $v_2$ is in $E^-(V)$.
   Right: incoming edge into $v_2$ is in $E^+(V)$.}
 \label{fig:f20}
\end{figure}

Similarly if the incoming edges into $v_2$ \wrt $V$ are in $E^+(V)$,
cf.~Fig.~\ref{fig:f20}\,(right), by planarity $u_3 v_2 \notin E(UV)$.
So at least $|U_{*3}| |V_{2*}| \geq 4 \times 3 = 12$
patterns are incompatible. Also $u_4 v_2$ and $u_2 v_1$ can
only be in $E^+(UV)$. By planarity both edges cannot be in $E(UV)$.
If $u_4 v_2 \notin E(UV)$, then at least $ |U_{*4}| |V_{2*}|
\geq 6 \times 3 = 18$ patterns are incompatible. Otherwise
$u_2 v_1 \notin E(UV)$ and either an incoming edge into $v_1$ \wrt $UV$
or an outgoing edge from $u_3$ \wrt $UV$ cannot be in $E(UV)$. Hence at least
$$|U_{*2}||V_{1*}| + \min(|\{\emptyset\}||V_{1*}|, |U_{*3}||\{\emptyset\}|)
\geq 2 \times 5 + \min(1 \times 5, 4 \times 1) = 14$$
patterns are incompatible. Therefore at least $ 12 + \min(18, 14) = 26$
patterns are incompatible.

\medskip
\emph{Case 3.2:} There is an incoming edge into $v_3$ in $E(V)$.
By Lemma~\ref{lem:group-with-exactly-11}\,(ii), all the
incoming edges into $v_3$ are on one side of $V$,
$|V_{1*}| \geq 4$, $|V_{2*}| \geq 2$ and $|V_{3*}| = 2$.

\medskip
\emph{Case 3.2.1:} The incoming edges into $v_2$ in $E(V)$ are on both
sides of $V$.

If the incoming edges into $v_3$ \wrt $V$ are in $E^+(V)$, see
Fig.~\ref{fig:f44}\,(left), then by planarity $u_3 v_3 \notin E(UV)$.
So at least $|U_{*3}||V_{3*}| \geq 4 \times 2 = 8$ patterns are
incompatible. Also $u_2 v_1$ and $u_4 v_3$ can only be in
$E^+(UV)$. By planarity both edges cannot be in $E(UV)$. Hence at least
$\min(|U_{*2}||V_{1*}|, |U_{*4}||V_{3*}|) = \min(2 \times 4, 6 \times 2) = 8$
patterns are incompatible.
By the same token an outgoing edge from $u_4$ \wrt $UV$ and the edges
$u_2 v_3$ and $u_3 v_2$ cannot exist together in $E(UV)$. Therefore at least
$$ \min (|U_{*4}||\{ \emptyset \}|, |U_{*2}||V_{3*}|, |U_{*3}||V_{2*}| )
= \min ( 6 \times 1, 2 \times 2, 4 \times 2 ) = 4 $$
patterns are incompatible. Similarly by planarity an incoming
edge into $v_1$ \wrt $UV$ and the edges $u_3 v_2$ and $u_4 v_3$ cannot
exist together in $E(UV)$. Therefore at least
$$ \min (|\{ \emptyset \}||V_{1*}|, |U_{*3}||V_{2*}|, |U_{*4}||V_{3*}|)
= \min (1 \times 4, 4 \times 2, 6 \times 2) = 4 $$ patterns are incompatible.
Hence at least $8+8+4+4 = 24$ patterns are incompatible.
\begin{figure}[htbp]
 \centering
 \includegraphics[scale=0.83]{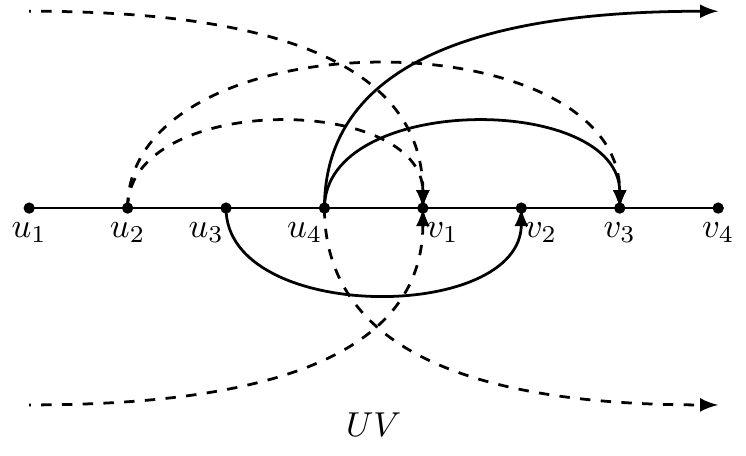}
 \hspace{0.2cm}
 \includegraphics[scale=0.83]{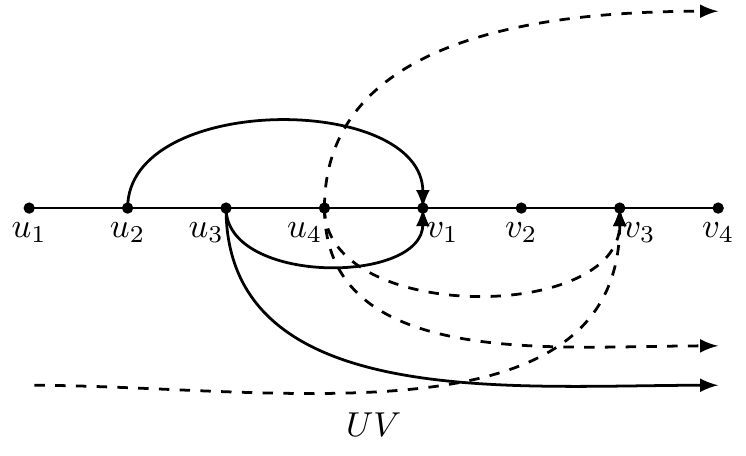}
 \caption{Left: incoming edges into $v_2$ are on both sides and
 incoming edges into $v_3$ are in $E^+(V)$.
 Right: incoming edges into $v_2$ are on both sides and
 incoming edges into $v_3$ are in $E^-(V)$.}
 \label{fig:f44}
\end{figure}

If the incoming edges into $v_3$ \wrt $V$ are in $E^-(V)$, see
Fig.~\ref{fig:f44}\,(right), then by planarity $u_2 v_3 \notin E(UV)$.
So at least $|U_{*2}||V_{3*}| \geq 2 \times 2 = 4$ patterns are
incompatible. Also $u_3 v_1$ and $u_4 v_3$ can only be in
$E^-(UV)$. By planarity both the edges cannot be in $E(UV)$. Hence at least
$ \min (|U_{*3}||V_{1*}|, |U_{*4}||V_{3*}|) =
\min (4 \times 4, 6 \times 2) = 12$ edges are incompatible.
By planarity an outgoing going edge from $u_3$ \wrt $UV$ and
an incoming edge into $v_3$ \wrt $UV$ cannot exist together in $E(UV)$.
Therefore at least
$ \min ( |U_{*3}||\{ \emptyset \}|, |\{ \emptyset \}||V_{3*}| )
= \min(4 \times 1, 1 \times 2) = 2$ patterns are incompatible.
By the same token an outgoing edge from $u_4$ \wrt $UV$ and the edges
$u_2 v_1$ and $u_3 v_1$ cannot be together in $E(UV)$. Hence at least
$ \min (|U_{*2}||V_{1*}|, |U_{*3}||V_{1*}|, |U_{*4}||\{ \emptyset \}|)
= \min (2 \times 4, 4 \times 4, 6 \times 1) = 6$ patterns are
incompatible. So at least $4 + 12 + 2 + 6 = 24$ patterns are
incompatible.

\medskip
\emph{Case 3.2.2:} All the incoming edges into $v_2$ in $E(V)$ are on one
side of $V$ and all the incoming edges into $v_2$ and $v_3$ in $E(V)$
are on the same side of $\xi_0$.

If the incoming edges into $v_2$ and $v_3$ \wrt $V$ are in $E^+(V)$,
see Fig.~\ref{fig:f24}\,(left), by planarity $u_3 v_2$ and
$u_3 v_3$ are not in $E(UV)$. So at least
$|U_{*3}||V_{2*}| + |U_{*3}||V_{3*}| =  4 \times 2 + 4 \times 2 = 16$
patterns are incompatible. Also $u_2 v_1$ and $u_4 v_2$ can
only be in $E^+(UV)$. By planarity both the edges cannot be in $E(UV)$.
Hence at least $ \min (|U_{*2}||V_{1*}|, |U_{*4}||V_{2*}|)
= \min (2 \times 4, 6 \times 2) = 8$ patterns are incompatible.
Therefore at least $16 + 8 = 24$ patterns are incompatible.
\begin{figure}[htbp]
 \centering
 \includegraphics[scale=0.83]{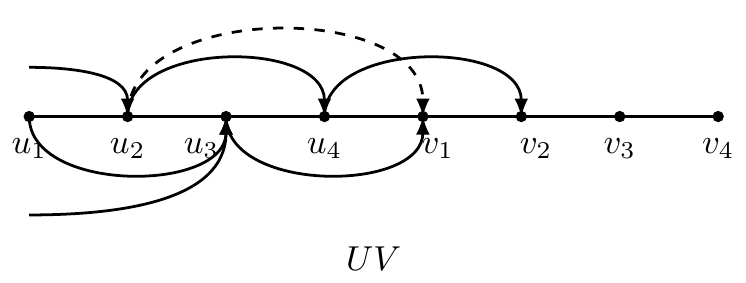}
 \hspace{0.2cm}
 \includegraphics[scale=0.83]{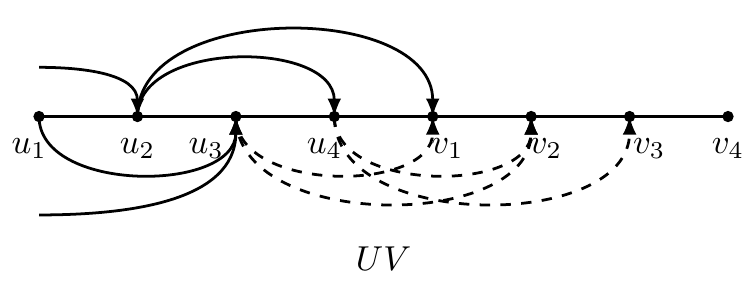}
 \caption{Left: incoming edges into $v_2$, $v_3$ are in $E^+(V)$.
 Right: incoming edges into $v_2$, $v_3$ are in $E^-(V)$.}
 \label{fig:f24}
\end{figure}

If the incoming edges into $v_2$ and $v_3$ \wrt $V$ are in $E^-(V)$, see
Fig.~\ref{fig:f24}\,(right), by planarity $u_2 v_3 \notin E(UV)$.
So at least $|U_{*2}||V_{3*}| \geq 2 \times 2 = 4$ patterns
are incompatible. Also $u_3 v_1$, $u_3 v_2$ $u_4 v_2$, $u_4 v_3$
can only be in $E^-(UV)$. By planarity either $u_3 v_1$ or $u_4 v_2$
and either $u_3 v_2$ or $u_4 v_3$ can be in $E(UV)$. Hence at least
\begin{align}
& \min( |U_{*3}||V_{1*}|, |U_{*4}||V_{2*}| )
+ \min ( |U_{*3}||V_{2*}|, |U_{*4}||V_{3*}| ) \nonumber\\
=& \min(4 \times 4, 6 \times 2) + \min(4 \times 2, 6 \times 2) = 20\nonumber
\end{align}
combinations are incompatible. Therefore at least
$4 + 20 = 24$ patterns are incompatible.

\medskip
\emph{Case 3.2.3:} All the incoming edges into $v_2$ are on one side
of $V$ and all the incoming edges into $v_3$ are on the opposite side of $V$.
\begin{figure}[htbp]
 \centering
 \includegraphics[scale=0.83]{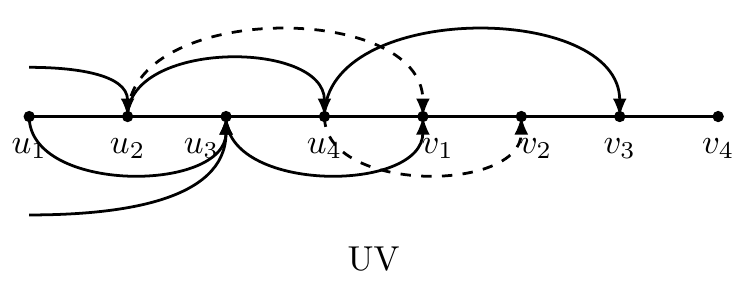}
 \hspace{0.2cm}
 \includegraphics[scale=0.83]{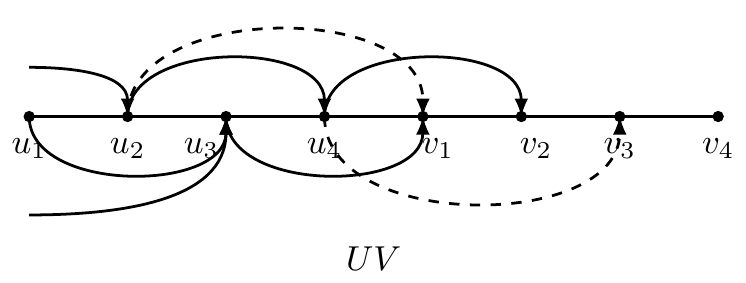}
 \caption{Left: incoming edge into $v_2$ is in $E^-(V)$
 and into $v_3$ is in $E^+(V)$.
 Right: Incoming edge into $v_2$ is in $E^+(V)$
 and into $v_3$ is in $E^-(V)$.}
 \label{fig:f22}
\end{figure}
Let the incoming edges \wrt $V$ in $E^+(V)$ are into $v_i$ and
the incoming edges \wrt $V$ in $E^-(V)$ are into $v_j$.
So either $i = 2$, $j = 3$ or $i = 3$, $j = 2$, see Fig.~\ref{fig:f22}.
Therefore $|V_{i*}|$, $|V_{j*}|$ are at least $\min(|V_{2*}|, |V_{3*}|) = 2$.
By planarity $u_3 v_i \notin E(UV)$.
So at least $|U_{*3}||V_{i*}| \geq 4 \times 2 = 8$ patterns are
incompatible. Also $u_2 v_1$ and $u_4 v_i$ can only be in $E^+(UV)$.
By planarity both edges cannot be in $E(UV)$. Hence at least
$ \min ( |U_{*2}||V_{1*}|, |U_{*4}||V_{i*}| )
= \min (2 \times 4, 6 \times 2) = 8$ patterns are incompatible.
Similarly both $u_3 v_1$ and $u_4 v_j$ can only be in $E^-(UV)$.
By planarity both edges cannot be in $E(UV)$. Hence at least
$ \min ( |U_{*3}||V_{1*}|, |U_{*4}||V_{j*}| )
= \min (4 \times 4, 6 \times 2) = 12$ patterns are incompatible.
Therefore at least $8 + 8 + 12 = 28$ patterns are incompatible.
\end{proof}

\begin{lemma}\label{lem:V-has-12}
Consider a group $UV$ consisting of two consecutive groups of $4$ vertices,
where $|I(U)|\geq 12$ and $|I(V)|=12$.
Then $UV$ allows at most $120$ incidence patterns.
\end{lemma}
\begin{proof}
We distinguish two cases depending on $|I(U)|$.

\medskip
\emph{Case 1:} $|I(U)| = 12$.
Then by Lemma~\ref{lem:group-with-at-least-12}\,(i) \& (ii), for
each vertex $u_i$, all the outgoing edges from $u_i$, if any,
are on one side of $U$.
Since $|I(UV)| \leq |I(U)| \cdot |I(V)| = 12 \times 12 = 144$, it
suffices to show that at least $144 - 120 = 24$ of these patterns
are incompatible. We distinguish three cases depending on which
vertex in $U$ have outgoing edges and which sides are containing
those outgoing edges.

\medskip
\emph{Case 1.1:} $U$ has outgoing edges from exactly one vertex.
By Lemma~\ref{lem:group-with-at-least-12}\,(ii), they are from
$u_3$ and we have $|U_{*3}| = 4$ and $|U_{*4}| = 7$.
For simplicity assume they are in $E^-(U)$.

\medskip
\emph{Case 1.1.1:} $V$ has incoming edges into exactly one vertex.
By Lemma~\ref{lem:group-with-at-least-12}\,(iv), they are into
$v_2$ and we have $|V_{2*}| = 4$ and $|V_{1*}| = 7$.

If the incoming edges into $v_2$ \wrt $V$ are in $E^-(V)$,
see Fig.~\ref{fig:f52}\,(left), then $u_3 v_1$ and $u_4 v_2$
can only be in $E^-(UV)$. By planarity both edges cannot be
in $E(UV)$. So at least
$ \min(|U_{*3}||V_{1*}|, |U_{*4}||V_{2*}|) = \min(4 \times 7, 7 \times 4) = 28$
patterns are incompatible.
\begin{figure}[htbp]
 \centering
 \includegraphics[scale=0.83]{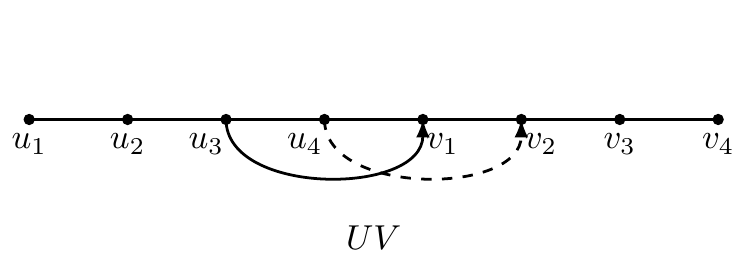}
 \hspace{0.2cm}
 \includegraphics[scale=0.83]{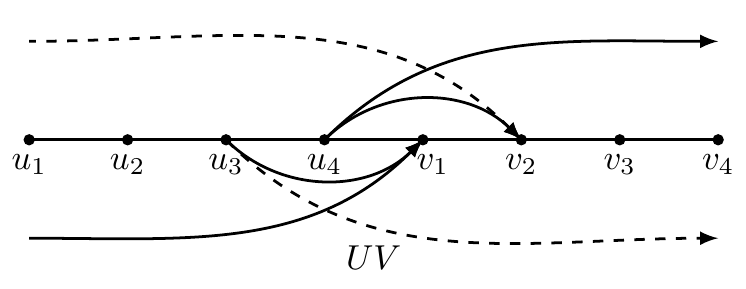}
 \caption{Left: the incoming edges into $v_2$ \wrt $V$ are in $E^-(V)$.
 Right: the incoming edges into $v_2$ \wrt $V$ are in $E^+(V)$.}
 \label{fig:f52}
\end{figure}

If the incoming edges into $v_2$ \wrt $V$ are in $E^+(V)$,
see Fig.~\ref{fig:f52}\,(right), then by planarity $u_3 v_2$
is not in $E(UV)$. So $|U_{*3}||V_{2*}| = 4 \times 4 = 16$
patterns are incompatible.
By planarity the edge $u_3 v_1$, an outgoing edge from $u_4$
\wrt $UV$ and an incoming edge into $v_2$ \wrt $UV$ cannot be
in $E(UV)$ together. Therefore at least
$$ \min(|U_{*3}||V_{1*}|, |U_{*4}||\{ \emptyset \}|, |\{ \emptyset \}||V_{2*}|)
= \min( 4 \times 7, 7 \times 1, 1 \times 4) = 4$$ patterns are incompatible.
By the same token the edge $u_4 v_2$, an incoming edge into $v_1$
\wrt $UV$ and an outgoing edge from $u_3$ \wrt $UV$ cannot be
in $E(UV)$ together. Therefore at least
$$ \min(|U_{*4}||V_{2*}|, |\{ \emptyset \}||V_{1*}|, |U_{*3}||\{ \emptyset \}|)
= \min( 7 \times 4, 1 \times 7, 4 \times 1) = 4$$ patterns are incompatible.
So $16 + 4 + 4 = 24$ patterns are incompatible.
Observe this group, $UV$, has $120$ patterns, which is the maximum
number of patterns.

\medskip
\emph{Case 1.1.2:} $V$ has incoming edges into more than one vertex.
Then by Lemma~\ref{lem:group-with-at-least-12}\,(iv), there are
incoming edges into $v_3$ and $v_2$ and we have $|V_{3*}| = 2$,
$|V_{2*}| \geq 3$ and $|V_{1*}| \geq 5$. We distinguish four scenarios
based on which sides of $V$ are containing the incoming edges into
$v_2$ and $v_3$.

If the incoming edges into $v_2$ and $v_3$ \wrt $V$ are in $E^-(V)$,
see Fig.~\ref{fig:f54}\,(left), then both $u_3 v_1$ and $u_4 v_2$
can only be in $E^-(UV)$. By planarity both the edges cannot be
in $E(UV)$. So at least
$\min(|U_{*3}||V_{1*}|, |U_{*4}||V_{2*}|) = \min(4 \times 5, 7 \times 3) = 20$
patterns are incompatible.
By the same token both $u_3 v_2$ and $u_4 v_3$ can only be in
$E^-(UV)$. By planarity both the edges cannot be in $E(UV)$. So at least
$\min(|U_{*3}||V_{2*}|, |U_{*4}||V_{3*}|) = \min (4 \times 3, 7 \times 2) = 12$
patterns are incompatible. So at least $20 + 12 = 32$ patterns are
incompatible.

\begin{figure}[htbp]
 \centering
 \includegraphics[scale=0.83]{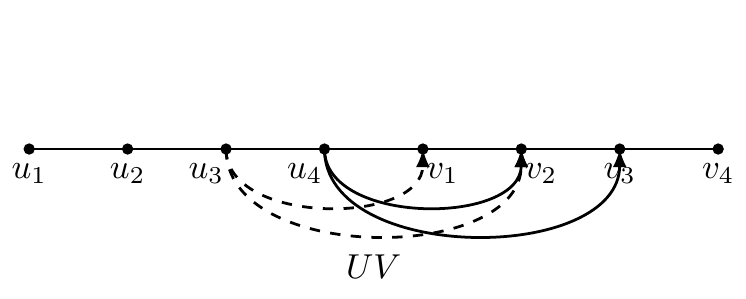}
 \hspace{0.2cm}
 \includegraphics[scale=0.83]{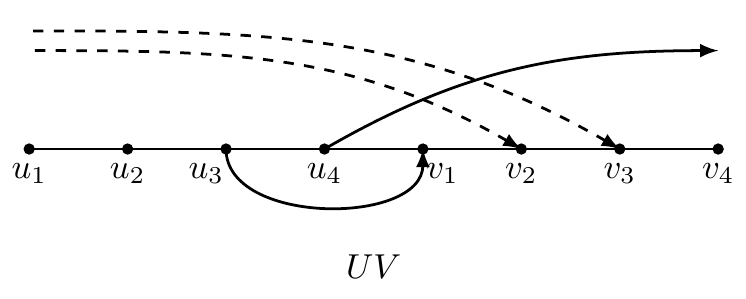}
 \caption{Left: the incoming edges into $v_2$ and $v_3$ \wrt $V$ are in $E^-(V)$.
 Right: the incoming edges into $v_2$ and $v_3$ \wrt $V$ are in $E^+(V)$.}
 \label{fig:f54}
\end{figure}

If the incoming edges into $v_2$ and $v_3$ \wrt $V$ are in $E^+(V)$,
see Fig.~\ref{fig:f54}\,(right), then by planarity both $u_3 v_2$
and $u_3 v_3$ are not in $E(UV)$. So at least
$|U_{*3}||V_{2*}| + |U_{*3}||V_{3*}| = 4 \times 3 + 4 \times 2 = 20$
patterns are incompatible.
By planarity incoming edges into $v_2$ and $v_3$ \wrt $UV$ cannot
exist together with an outgoing edges from $u_4$ \wrt $U$ and
the edge $u_3 v_1$. So at least
$$\min (|\{ \emptyset \}||V_{2*}| + |\{ \emptyset \}||V_{3*}|,
|U_{*4}||\{ \emptyset \}|, |U_{*3}||V_{1*}|)
= \min (1 \times 3 + 1 \times 2, 7 \times 1, 4 \times 5) = 5$$
patterns are incompatible. So at least $20 + 5 = 25$ patterns
are incompatible.

If the incoming edges into $v_2$ \wrt $V$ are in $E^-(V)$ and
the incoming edges into $v_3$ \wrt $V$ are in $E^+(V)$,
see Fig.~\ref{fig:f56}\,(left),
then by planarity $u_3 v_3$ is not in $E(UV)$. So at least
$|U_{*3}||V_{3*}| = 4 \times 2 = 8$ patterns are incompatible.
Also both $u_3 v_1$ and $u_4 v_2$ can only be in $E^-(UV)$. By
planarity both edges cannot be in $E(UV)$. So at least
$\min (|U_{*3}||V_{1*}|, |U_{*4}||V_{2*}|) = \min (4 \times 5, 7 \times 3) = 20$
patterns are incompatible. So at least $8 + 20 = 28$ patterns
are incompatible.

\begin{figure}[htbp]
 \centering
 \includegraphics[scale=0.83]{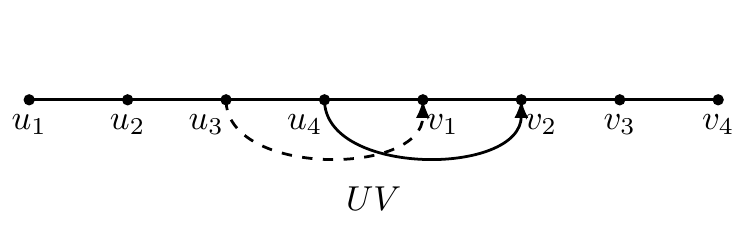}
 \hspace{0.2cm}
 \includegraphics[scale=0.83]{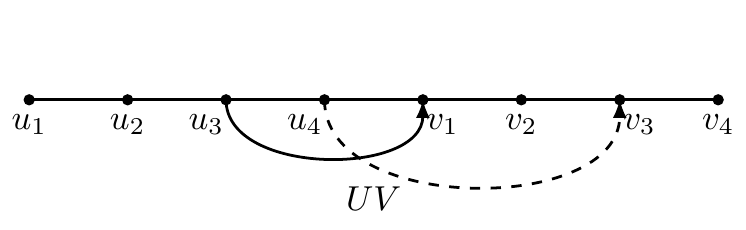}
 \caption{Left: the incoming edges into $v_2$ \wrt $V$ are in $E^-(V)$.
 Right: the incoming edges into $v_2$ \wrt $V$ are in $E^+(V)$.}
 \label{fig:f56}
\end{figure}

If the incoming edges into $v_2$ \wrt $V$ are in $E^+(V)$ and
the incoming edges into $v_3$ \wrt $V$ are in $E^-(V)$,
see Fig.~\ref{fig:f56}\,(right),
then by planarity $u_3 v_2$ is not in $E(UV)$. So at least
$|U_{*3}||V_{2*}| = 4 \times 3 = 12$ patterns are incompatible.
Also both $u_3 v_1$ and $u_4 v_3$ can only be in $E^-(UV)$. By
planarity both edges cannot be in $E(UV)$. So at least
$\min (|U_{*3}||V_{1*}|, |U_{*4}||V_{3*}|) = \min(4 \times 5, 7 \times 2) = 14$
patterns are incompatible. So at least $12 + 14 = 26$ patterns
are incompatible.

\medskip
\emph{Case 1.2:} $U$ has outgoing edges from $u_2$ and $u_3$ and both
are on the same side. By Lemma~\ref{lem:group-with-at-least-12}\,(ii),
$|U_{*2}| = 2$, $|U_{*3}| \geq 3$ and $|U_{*4}| \geq 5$. For simplicity
assume the outgoing edges are in $E^-(U)$.

\medskip
\emph{Case 1.2.1:} $V$ has incoming edges into exactly one vertex.
By Lemma~\ref{lem:group-with-at-least-12}\,(iv), these are into
$v_2$ and we have $|V_{2*}| = 4$ and $|V_{1*}| = 7$.

If the incoming edges into $v_2$ are in $E^-(V)$,
see Fig.~\ref{fig:f58}\,(left), then both $u_3 v_1$ and $u_4 v_2$
can only be in $E^-(UV)$. By planarity both edges cannot be
in $E(UV)$. So at least
$\min(|U_{*3}||V_{1*}|, |U_{*4}||V_{2*}|) = \min(3 \times 7, 5 \times 4) = 20$
patterns are incompatible.
By planarity an outgoing edge from $u_4$ \wrt $UV$, an incoming edges
into $v_1$ \wrt $UV$ and the edge $u_2 v_2$ cannot exist together
in $E(UV)$. So at least
$$\min (|U_{*4}||\{ \emptyset \}|, |\{ \emptyset \}||V_{1*}|, |U_{*2}||V_{2*}|)
= \min (5 \times 1, 1 \times 7, 2 \times 4) = 5$$
patterns are incompatible. So in total at least $20 + 5 = 25$
incidence patterns are incompatible.

\begin{figure}[htbp]
 \centering
 \includegraphics[scale=0.83]{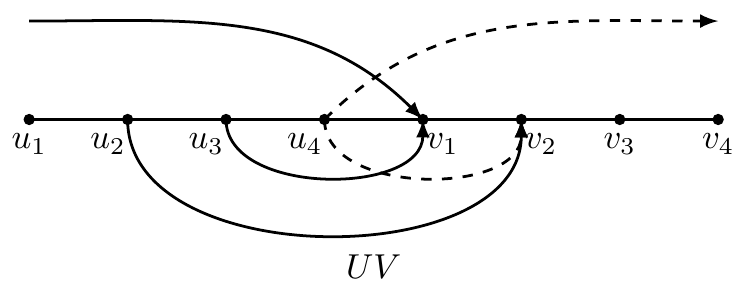}
 \hspace{0.2cm}
 \includegraphics[scale=0.83]{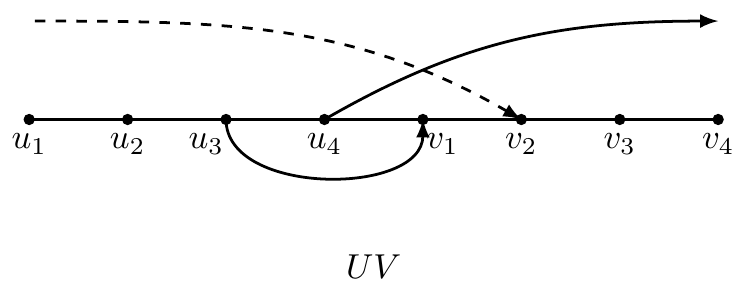}
 \caption{Left: the incoming edges into $v_2$ \wrt $V$ are in $E^-(V)$.
 Right: the incoming edges into $v_2$ \wrt $V$ are in $E^+(V)$.}
 \label{fig:f58}
\end{figure}

If the incoming edges into $v_2$ are in $E^+(V)$,
see Fig.~\ref{fig:f58}\,(right), then by planarity $u_2 v_2$
and $u_3 v_2$ are not in $E(UV)$. So at least
$|U_{*2}||V_{2*}| + |U_{*3}||V_{2*}| = 2 \times 4 + 3 \times 4 = 20$
patterns are incompatible.
By planarity an outgoing edge from $u_4$ \wrt $UV$, an incoming edges
into $v_2$ \wrt $UV$ and the edge $u_3 v_1$ cannot exist together
in $E(UV)$. So at least
$$\min (|U_{*4}||\{ \emptyset \}|, |\{ \emptyset \}||V_{2*}|, |U_{*3}||V_{1*}|)
= \min (5 \times 1, 1 \times 4, 3 \times 7) = 4$$
patterns are incompatible. So in total at least $20 + 4 = 24$
incidence patterns are incompatible.

\medskip
\emph{Case 1.2.2:} $V$ has incoming edges into more than one vertex.
Then by Lemma~\ref{lem:group-with-at-least-12}\,(iv), there are
incoming edges into $v_3$ and $v_2$ and we have $|V_{3*}| = 2$,
$|V_{2*}| \geq 3$ and $|V_{1*}| \geq 5$. We distinguish four scenarios
based on which sides of $V$ are containing the incoming edges into
$v_2$ and $v_3$.

If the incoming edges into $v_2$ and $v_3$ \wrt $V$ are in $E^-(V)$,
see Fig.~\ref{fig:f60}\,(left), then both $u_3 v_1$ and $u_4 v_2$
can only be in $E^-(UV)$. By planarity both the edges cannot be
in $E(UV)$. So at least
$\min(|U_{*3}||V_{1*}|, |U_{*4}||V_{2*}|) = \min(3 \times 5, 5 \times 3) = 15$
patterns are incompatible.
By the same token both $u_2 v_1$ and $u_4 v_3$ can only be in
$E^-(UV)$. By planarity both the edges cannot be in $E(UV)$. So at least
$\min(|U_{*2}||V_{1*}|, |U_{*4}||V_{3*}|) = \min (2 \times 5, 5 \times 2) = 10$
patterns are incompatible. So at least $15 + 10 = 25$ patterns are
incompatible.

\begin{figure}[htbp]
 \centering
 \includegraphics[scale=0.83]{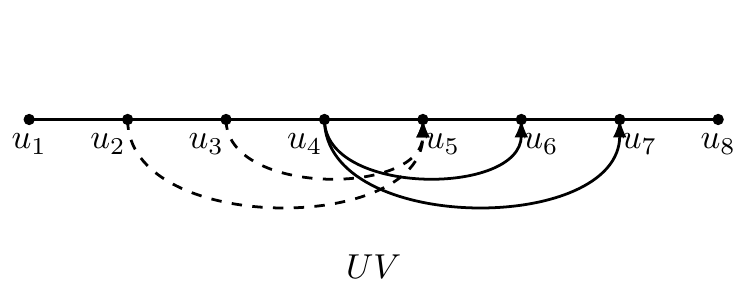}
 \hspace{0.2cm}
 \includegraphics[scale=0.83]{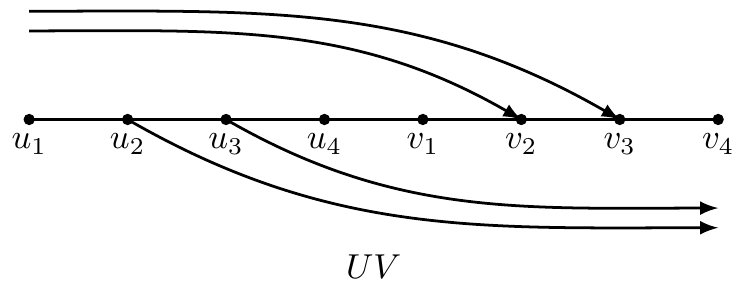}
 \caption{Left: the incoming edges into $v_2$ and $v_3$ \wrt $V$ are in $E^-(V)$.
 Right: the incoming edges into $v_2$ and $v_3$ \wrt $V$ are in $E^+(V)$.}
 \label{fig:f60}
\end{figure}

If the incoming edges into $v_2$ and $v_3$ \wrt $V$ are in $E^+(V)$,
see Fig.~\ref{fig:f60}\,(right), then by planarity $u_2 v_2$,
$u_2 v_3$, $u_3 v_2$ and $u_3 v_3$ are not in $E(UV)$. So at least
$$|U_{*2}||V_{2*}| + |U_{*2}||V_{3*}| + |U_{*3}||V_{2*}| + |U_{*3}||V_{3*}|
= 2 \times 3 + 2 \times 2 + 3 \times 3 + 3 \times 2 = 25$$
patterns are incompatible.

If the incoming edges into $v_2$ \wrt $V$ are in $E^-(V)$ and
the incoming edges into $v_3$ \wrt $V$ are in $E^+(V)$,
see Fig.~\ref{fig:f62}\,(left), then by planarity $u_2 v_3$
and $u_3 v_3$ are not in $E(UV)$. So at least
$|U_{*2}||V_{3*}| + |U_{*3}||V_{3*}| = 2 \times 2 + 3 \times 2 = 10$
patterns are incompatible.
Both $u_3 v_1$ and $u_4 v_2$ can only be in $E^-(UV)$. By planarity
both edges cannot be in $E(UV)$. So at least
$\min(|U_{*3}||V_{1*}|, |U_{*4}||V_{2*}|) = \min(3 \times 5, 5 \times 3) = 15$
patterns are incompatible. So at least $10 + 15 = 25$ patterns
are incompatible.

\begin{figure}[htbp]
 \centering
 \includegraphics[scale=0.83]{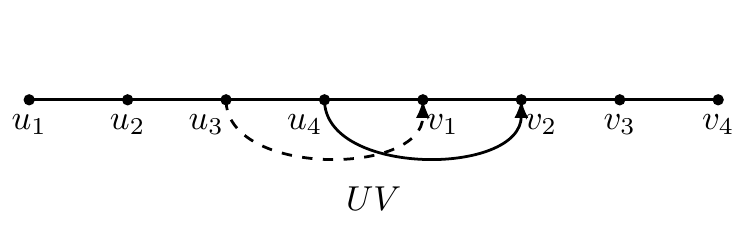}
 \hspace{0.2cm}
 \includegraphics[scale=0.83]{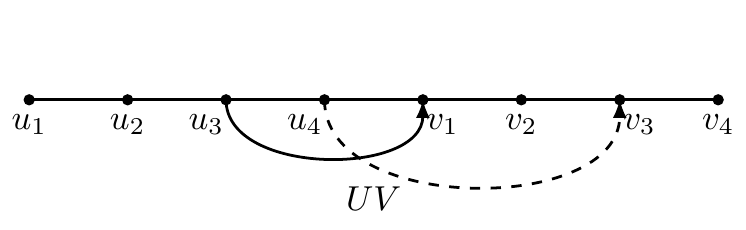}
 \caption{Left: the incoming edges into $v_2$ \wrt $V$ are in $E^-(V)$.
 Right: the incoming edges into $v_2$ \wrt $V$ are in $E^+(V)$.}
 \label{fig:f62}
\end{figure}

If the incoming edges into $v_2$ \wrt $V$ are in $E^+(V)$ and
the incoming edges into $v_3$ \wrt $V$ are in $E^-(V)$,
see Fig.~\ref{fig:f62}\,(right), then by planarity $u_2 v_2$
and $u_3 v_2$ are not in $E(UV)$. So at least
$|U_{*2}||V_{2*}| + |U_{*3}||V_{2*}| = 2 \times 3 + 3 \times 3 = 15$
patterns are incompatible.
Both $u_3 v_1$ and $u_4 v_3$ can only be in $E^-(UV)$. By planarity
both the edges cannot be in $E(UV)$. So at least
$\min(|U_{*3}||V_{1*}|, |U_{*4}||V_{3*}|) = \min(3 \times 5, 5 \times 2) = 10$
patterns are incompatible. So at least $15 + 10 = 25$ patterns
are incompatible.

\medskip
\emph{Case 1.3:} $U$ has outgoing edges from $u_2$ and $u_3$ and both
are on opposite sides. By Lemma~\ref{lem:group-with-at-least-12}\,(ii),
$|U_{*2}| = 2$, $|U_{*3}| \geq 3$ and $|U_{*4}| \geq 5$. For simplicity
assume the outgoing edges from $u_3$ \wrt $U$ are in $E^-(U)$.

\medskip
\emph{Case 1.3.1:} $V$ has incoming edges into exactly one vertex.
By Lemma~\ref{lem:group-with-at-least-12}\,(iv), they are into
$v_2$ and we have $|V_{2*}| = 4$ and $|V_{1*}| = 7$.

If the incoming edges into $v_2$ \wrt $V$ are in $E^-(V)$,
see Fig.~\ref{fig:f64}\,(left), then by planarity $u_2 v_2$
is not in $E(UV)$. So at least
$|U_{*2}||V_{2*}| = 2 \times 4 = 8$ patterns are incompatible.
Also both $u_3 v_1$ and $u_4 v_2$ can only be in $E^-(UV)$. By
planarity both edges cannot be in $E(UV)$. So at least
$\min(|U_{*3}||V_{1*}|, |U_{*4}||V_{2*}|) = \min(3 \times 7, 5 \times 4) = 20$
patterns are incompatible. So at least $8 + 20 = 28$
patterns are incompatible.

\begin{figure}[htbp]
 \centering
 \includegraphics[scale=0.83]{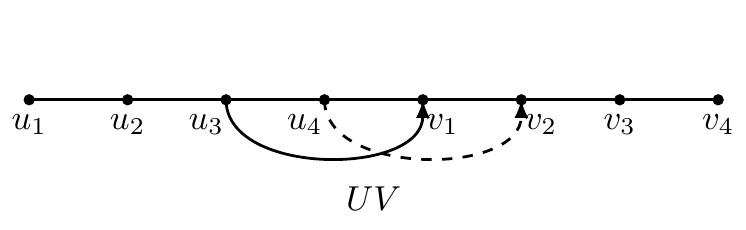}
 \hspace{0.2cm}
 \includegraphics[scale=0.83]{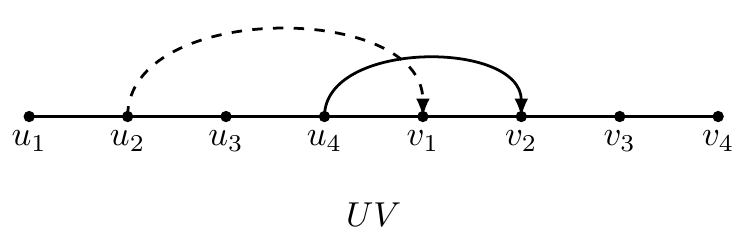}
 \caption{Left: the incoming edges into $v_2$ \wrt $V$ are in $E^-(V)$.
 Right: the incoming edges into $v_2$ \wrt $V$ are in $E^+(V)$.}
 \label{fig:f64}
\end{figure}

If the incoming edges into $v_2$ \wrt $V$ are in $E^+(V)$,
see Fig.~\ref{fig:f64}\,(right), then by planarity $u_3 v_2$
is not in $E(UV)$. So at least $|U_{*3}||V_{2*}| = 3 \times 4 = 12$
patterns are incompatible. Also both $u_2 v_1$ and $u_4 v_2$
can only be in $E^+(UV)$. By planarity both edges cannot
be in $E(UV)$. So at least
$\min(|U_{*2}||V_{1*}|, |U_{*4}||V_{2*}|) = \min(2 \times 7, 5 \times 4) = 14$
patterns are incompatible. So at least $12 + 14 = 26$ patterns are incompatible.

\medskip
\emph{Case 1.3.2:} $V$ has incoming edges into more than one vertex.
Then by Lemma~\ref{lem:group-with-at-least-12}\,(iv), there are
incoming edges into $v_3$ and $v_2$ and we have $|V_{3*}| = 2$,
$|V_{2*}| \geq 3$ and $|V_{1*}| \geq 5$. We distinguish four scenarios
based on which sides of $V$ are containing the incoming edges into
$v_2$ and $v_3$.

If the incoming edges into $v_2$ and $v_3$ \wrt $V$ are in $E^-(V)$,
see Fig.~\ref{fig:f66}\,(left),
then by planarity $u_2 v_2$ and $u_2 v_3$ are not in
$E(UV)$. So at least
$|U_{*2}||V_{2*}| + |U_{*2}||V_{3*}| = 2 \times 3 + 2 \times 2 = 10$
patterns are incompatible.
Also both $u_3 v_1$ and $u_4 v_2$ can only be in $E^-(UV)$. By
planarity both edges cannot be in $E(UV)$. So at least
$\min(|U_{*3}||V_{1*}|, |U_{*4}||V_{2*}|) = \min(3 \times 5, 5 \times 3) = 15$
patterns are incompatible. So at least $10 + 15 = 25$ patterns
are incompatible.
\begin{figure}[htbp]
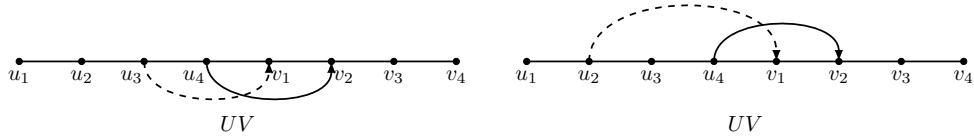

 \centering
 \includegraphics[scale=0.83]{f55.pdf}
 \hspace{0.2cm}
 \includegraphics[scale=0.83]{f62.pdf}
 \caption{Left: the incoming edges into $v_2$ and $v_3$ \wrt $V$ are in $E^-(V)$.
 Right: the incoming edges into $v_2$ and $v_3$ \wrt $V$ are in $E^+(V)$.}
 \label{fig:f66}
\end{figure}

If the incoming edges into $v_2$ and $v_3$ \wrt $V$ are in $E^+(V)$,
see Fig.~\ref{fig:f66}\,(right),
then by planarity $u_3 v_2$ and $u_3 v_3$ are not in $E(UV)$. So at least
$|U_{*3}||V_{2*}| + |U_{*3}||V_{3*}| = 3 \times 3 + 3 \times 2 = 15$
patterns are incompatible. Also both $u_2 v_1$ and $u_4 v_2$ can only
be in $E^+(UV)$. By planarity both edges cannot be in $E(UV)$. So at least
$\min(|U_{*2}||V_{1*}|, |U_{*4}||V_{2*}|) = \min(2 \times 5, 5 \times 3) = 10$
patterns are incompatible. So at least $15 + 10 = 25$ patterns
are incompatible.

If the incoming edges into $v_2$ \wrt $V$ is in $E^-(V)$ and
the incoming edges into $v_3$ \wrt $V$ are in $E^+(V)$,
see Fig.~\ref{fig:f68}\,(left), then both $u_2 v_1$ and $u_4 v_3$
can only be in $E^+(UV)$. By planarity both edges cannot be in $E(UV)$. So at least
$\min(|U_{*2}||V_{1*}|, |U_{*4}||V_{3*}|) = \min(2 \times 5, 5 \times 2) = 10$
patterns are incompatible.
By similar token both $u_3 v_1$ and $u_4 v_2$ can only be in
$E^-(UV)$. By planarity both edges cannot be in $E(UV)$. So at least
$\min(|U_{*3}||V_{1*}|, |U_{*4}||V_{2*}|) = \min(3 \times 5, 5 \times 3) = 15$
patterns are incompatible. So at least $10 + 15 = 25$ patterns
are incompatible.
\begin{figure}[htbp]
 \centering
 \includegraphics[scale=0.83]{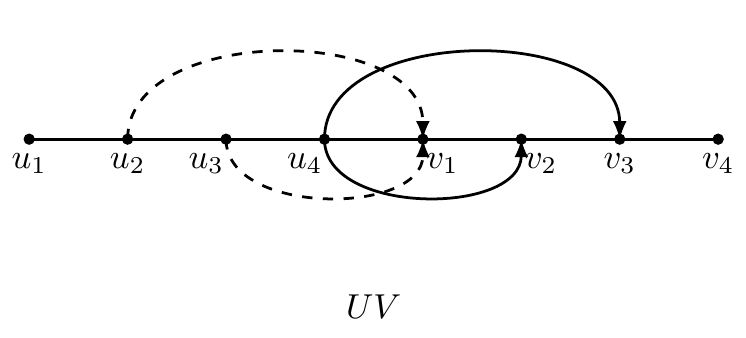}
 \hspace{0.2cm}
 \includegraphics[scale=0.83]{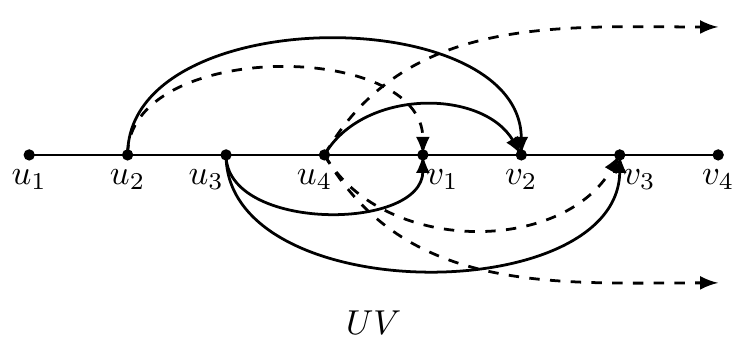}
 \caption{Left: the incoming edges into $v_2$ \wrt $V$ are in $E^-(V)$.
 Right: the incoming edges into $v_2$ \wrt $V$ are in $E^+(V)$.}
 \label{fig:f68}
\end{figure}

If the incoming edges into $v_2$ \wrt $V$ is in $E^+(V)$ and
the incoming edges into $v_3$ \wrt $V$ are in $E^-(V)$,
see Fig.~\ref{fig:f68}\,(right), then both $u_2 v_1$ and
$u_4 v_2$ can only be in $E^+(UV)$. By planarity both edges
cannot be in $E(UV)$. So at least
$\min(|U_{*2}||V_{1*}|, |U_{*4}||V_{2*}|) = \min(2 \times 5, 5 \times 3) = 10$
patterns are incompatible.
Also both $u_3 v_1$ and $u_4 v_3$ can only be in
$E^-(UV)$. By planarity both edges cannot be in $E(UV)$. So at least
$\min(|U_{*3}||V_{1*}|, |U_{*4}||V_{3*}|) = \min(3 \times 5, 5 \times 2) = 10$
patterns are incompatible.
By planarity an outgoing edge from $u_4$ \wrt $UV$, the edge $u_2 v_2$
and the edge $u_3 v_3$ cannot exist together in $E(UV)$. So at least
$$\min (|U_{*4}||\{ \emptyset \}|, |U_{*2}||V_{2*}|, |U_{*3}||V_{3*}|)
= \min (5 \times 1, 2 \times 3, 3 \times 2) = 5$$
patterns are incompatible. So in total at least $10 + 10 + 5 = 25$
incidence patterns are incompatible.

\medskip
\emph{Case 2:} $|I(U)| = 13$.
By Lemma~\ref{lem:group-with-at-least-13}, $U$ is either $A$ or $A^R$.
If $U$ is $A^R$, then after reflecting $UV$ around the horizontal
axis, $U$ is $A$. We analyze the cases based on this fact.
Since $|I(UV)| \leq |I(U)| \cdot |I(V)| = 13 \times 12 = 156$,
it suffices to show that at least $156 - 120 = 36$ of these patterns
are incompatible.

\medskip
\emph{Case 2.1:} $V$ has incoming edges into exactly one vertex.
Then by Lemma~\ref{lem:group-with-at-least-12}\,(iv), they are into
$v_2$ and we have $|V_{2*}| = 4$ and $|V_{1*}| = 7$.

If the incoming edges into $v_2$ \wrt $V$ are in $E^-(V)$,
see Fig.~\ref{fig:f70}\,(left),
then by planarity $u_2 v_2$ is not in $E(UV)$. So at least
$|U_{*2}||V_{2*}| = 2 \times 4 = 8$ patterns are incompatible.
Both $u_3 v_1$ and $u_4 v_2$ can only be in
$E^-(UV)$. By planarity both edges cannot be in $E(UV)$. So at least
$\min(|U_{*3}||V_{1*}|, |U_{*4}||V_{2*}|) = \min(4 \times 7, 6 \times 4) = 24$
patterns are incompatible.
By planarity an outgoing edge from $u_4$ \wrt $UV$, the edge $u_2 v_1$
and the edge $u_3 v_2$ cannot exist together in $E(UV)$. So at least
$$\min (|U_{*4}||\{ \emptyset \}|, |U_{*2}||V_{1*}|, |U_{*3}||V_{2*}|)
= \min (6 \times 1, 2 \times 7, 4 \times 4) = 6$$
patterns are incompatible. So in total at least $8 + 24 + 6 = 38$
incidence patterns are incompatible.

\begin{figure}[htbp]
 \centering
 \includegraphics[scale=0.83]{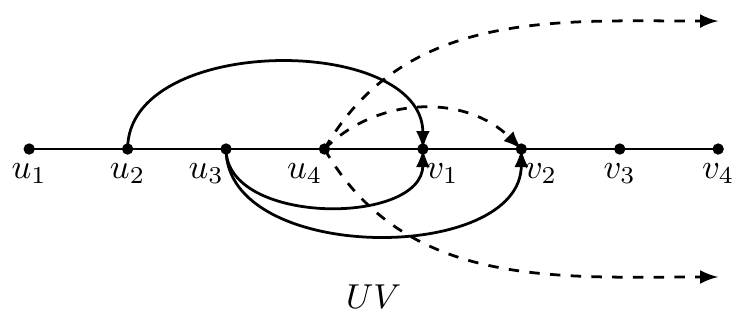}
 \hspace{0.2cm}
 \includegraphics[scale=0.83]{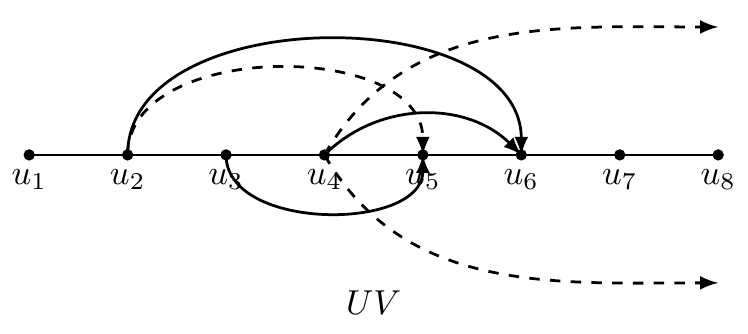}
 \caption{Left: the incoming edges into $v_2$ \wrt $V$ are in $E^-(V)$.
 Right: the incoming edges into $v_2$ \wrt $V$ are in $E^+(V)$.}
 \label{fig:f70}
\end{figure}

If the incoming edges into $v_2$ \wrt $V$ are in $E^+(V)$,
see Fig.~\ref{fig:f70}\,(right),
then by planarity $u_3 v_2$ is not allowed. So at least
$|U_{*3}||V_{2*}| = 4 \times 4 = 16$ patterns are incompatible.
Both $u_2 v_1$ and $u_4 v_2$ can only be in
$E^+(UV)$. By planarity both edges cannot be in $E(UV)$. So at least
$\min(|U_{*2}||V_{1*}|, |U_{*4}||V_{2*}|) = \min(2 \times 7, 6 \times 4) = 14$
patterns are incompatible.
By planarity an outgoing edge from $u_4$ \wrt $UV$, the edge $u_2 v_2$
and the edge $u_3 v_1$ cannot exist together in $E(UV)$. So at least
$$\min (|U_{*4}||\{ \emptyset \}|, |U_{*2}||V_{2*}|, |U_{*3}||V_{1*}|)
= \min (6 \times 1, 2 \times 4, 4 \times 7) = 6$$
patterns are incompatible. So at least $16 + 14 + 6 = 36$ patterns
are incompatible.

\emph{Case 2.2:} $V$ has incoming edges into more than one vertex.
By Lemma~\ref{lem:group-with-at-least-12}\,(iv), there are
incoming edges into $v_3$ and $v_2$ and we have $|V_{3*}| = 2$,
$|V_{2*}| \geq 3$ and $|V_{1*}| \geq 5$. We distinguish four scenarios
based on which sides of $V$ are containing the incoming edges into
$v_2$ and $v_3$.

If the incoming edges into $v_2$ and $v_3$ \wrt $V$ are in $E^-(V)$,
see Fig.~\ref{fig:f72}\,(left),
then by planarity $u_2 v_2$ is not in $E(UV)$. So at least
$|U_{*2}||V_{2*}| = 2 \times 3 = 6$ patterns are incompatible.
Both $u_3 v_1$ and $u_4 v_2$ can only be in
$E^-(UV)$. By planarity both edges cannot be in $E(UV)$. So at least
$\min(|U_{*3}||V_{1*}|, |U_{*4}||V_{2*}|) = \min(4 \times 5, 6 \times 3) = 18$
patterns are incompatible.
Similarly both $u_3 v_2$ and $u_4 v_3$ can only be in
$E^-(UV)$. By planarity both edges cannot be in $E(UV)$. So at least
$\min(|U_{*3}||V_{2*}|, |U_{*4}||V_{3*}|) = \min(4 \times 3, 6 \times 2) = 12$
patterns are incompatible. So at least $6 + 18 + 12 = 36$ patterns
are incompatible.

\begin{figure}[htbp]
 \centering
 \includegraphics[scale=0.83]{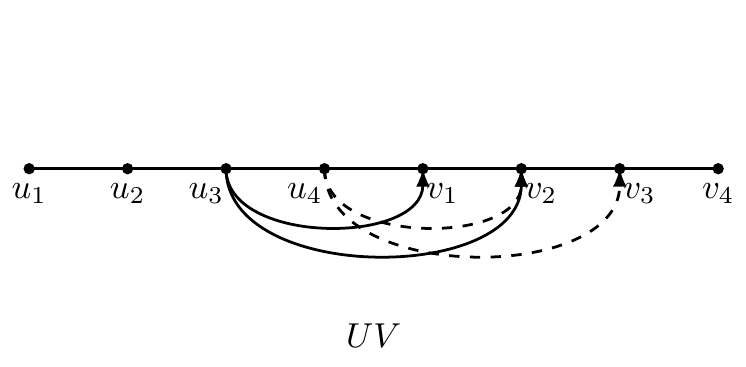}
 \hspace{0.2cm}
 \includegraphics[scale=0.83]{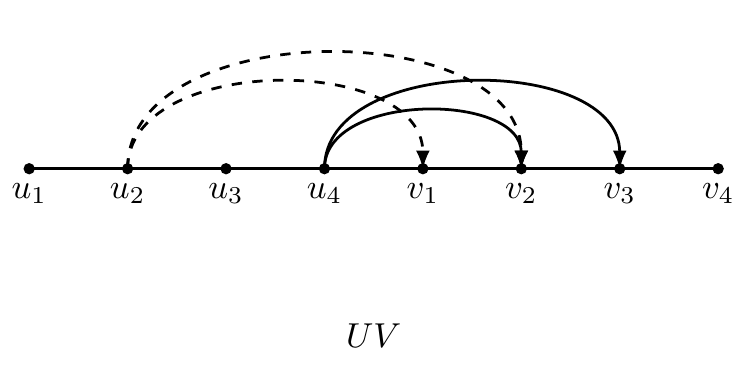}
 \caption{Left: the incoming edges into $v_2$ and $v_3$ \wrt $V$ are in $E^-(V)$.
 Right: the incoming edges into $v_2$ and $v_3$ \wrt $V$ are in $E^+(V)$.}
 \label{fig:f72}
\end{figure}

If the incoming edges into $v_2$ and $v_3$ \wrt $V$ are in $E^+(V)$,
see Fig.~\ref{fig:f72}\,(right),
then by planarity both $u_3 v_2$ and $u_3 v_3$ are not in $E(UV)$.
So at least $|U_{*3}||V_{2*}| + |U_{*3}||V_{3*}| = 4 \times 3 + 4 \times 2 = 20$
patterns are incompatible.
Also both $u_2 v_1$ and $u_4 v_2$ can only be in
$E^+(UV)$. By planarity both edges cannot be in $E(UV)$. So at least
$\min(|U_{*2}||V_{1*}|, |U_{*4}||V_{2*}|) = \min(2 \times 5, 6 \times 3) = 10$
patterns are incompatible.
Similarly both $u_2 v_2$ and $u_4 v_3$ can only be in
$E^+(UV)$. By planarity both edges cannot be in $E(UV)$. So at least
$\min(|U_{*2}||V_{2*}|, |U_{*4}||V_{3*}|) = \min(2 \times 3, 6 \times 2) = 6$
patterns are incompatible. So at least $20 + 10 + 6 = 36$ patterns
are incompatible.

If the incoming edges into $v_2$ \wrt $V$ is in $E^-(V)$ and
the incoming edges into $v_3$ \wrt $V$ are in $E^+(V)$,
see Fig.~\ref{fig:f74}\,(left),
then by planarity $u_3 v_3$ is not in $E(UV)$. So at least
$|U_{*3}||V_{3*}| = 4 \times 2 = 8$ patterns are incompatible.
Both $u_2 v_1$ and $u_4 v_3$ can only be in
$E^+(UV)$. By planarity both edges cannot be in $E(UV)$. So at least
$\min(|U_{*2}||V_{1*}|, |U_{*4}||V_{3*}|) = \min(2 \times 5, 6 \times 2) = 10$
patterns are incompatible.
Similarly both $u_3 v_1$ and $u_4 v_2$ can only be in
$E^-(UV)$. By planarity both edges cannot be in $E(UV)$. So at least
$\min(|U_{*3}||V_{1*}|, |U_{*4}||V_{2*}|) = \min(4 \times 5, 6 \times 3) = 18$
patterns are incompatible. So at least $8 + 10 + 18 = 36$ patterns
are incompatible.

\begin{figure}[htbp]
 \centering
 \includegraphics[scale=0.83]{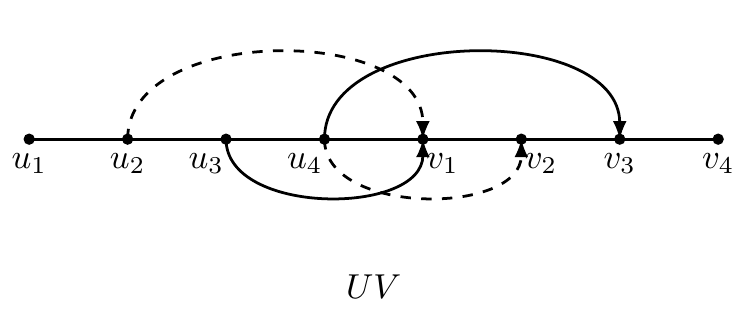}
 \hspace{0.2cm}
 \includegraphics[scale=0.83]{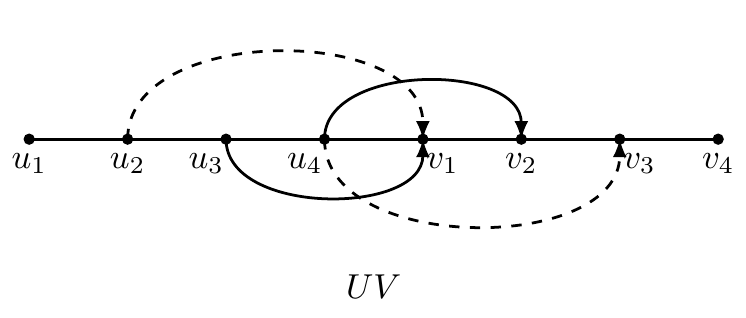}
 \caption{Left: the incoming edges into $v_2$ \wrt $V$ are in $E^-(V)$.
 Right: the incoming edges into $v_2$ \wrt $V$ are in $E^+(V)$.}
 \label{fig:f74}
\end{figure}

If the incoming edges into $v_2$ \wrt $V$ is in $E^+(V)$ and
the incoming edges into $v_3$ \wrt $V$ are in $E^-(V)$,
see Fig.~\ref{fig:f74}\,(right),
then by planarity both $u_2 v_3$ and $u_3 v_2$ are not in $E(UV)$.
So at least $|U_{*2}||V_{3*}| + |U_{*3}||V_{2*}| = 2 \times 2 + 4 \times 3 = 16$
patterns are incompatible.
Both $u_2 v_1$ and $u_4 v_2$ can only be in
$E^+(UV)$. By planarity both edges cannot be in $E(UV)$. So at least
$\min(|U_{*2}||V_{1*}|, |U_{*4}||V_{2*}|) = \min(2 \times 5, 6 \times 3) = 10$
patterns are incompatible.
Similarly both $u_3 v_1$ and $u_4 v_3$ can only be in
$E^-(UV)$. By planarity both edges cannot be in $E(UV)$. So at least
$\min(|U_{*3}||V_{1*}|, |U_{*4}||V_{3*}|) = \min(4 \times 5, 6 \times 2) = 12$
patterns are incompatible. So at least $16 + 10 + 12 = 38$ patterns
are incompatible.
\end{proof}

\begin{lemma}\label{lem:V-has-13}
Consider a group $UV$ consisting of two consecutive groups of $4$ vertices,
where $|I(U)|=|I(V)|=13$. Then $UV$ allows at most $120$ incidence patterns.
\end{lemma}
\begin{proof}
By Lemma~\ref{lem:group-with-at-least-13}, $U$ is either $A$ or $A^R$.
We may assume, after reflecting $UV$ around a horizontal axis, that $U$ is $A$.
Therefore $|U_{*2}| = 2$, $|U_{*3}| = 4$ and $|U_{*4}| = 6$, see Figure~\ref{fig:f1}.
Similarly, Lemma~\ref{lem:group-with-at-least-13} implies that $V$ is either $A$ or $A^R$.
We distinguish two cases depending on whether $V$ is $A$ or $A^R$.
The cross product of $I(U)$ and $I(V)$ yields $13 \times 13 = 169$ possible patterns.
It suffices to show that at least $169 - 120 = 49$ of these patterns are incompatible.

\medskip
\emph{Case 1:} $V$ is $A$, see Figure~\ref{fig:f48}\,(left).
By planarity $u_2 v_3$ and $u_3 v_2$ are not in $E(UV)$. So at least
$|U_{*2}||V_{3*}| + |U_{*3}||V_{2*}| = 2 \times 2 + 4 \times 4 = 20$
patterns are incompatible. Further, $u_2 v_1$ and $u_4 v_2$ can only be in
$E^+(UV)$. By planarity both edges cannot be in $E(UV)$. Hence at least
$\min (|U_{*2}||V_{1*}|, |U_{*4}||V_{2*}|) = \min (2 \times 6, 6 \times 4) = 12$
patterns are incompatible. Similarly $u_3 v_1$ and $u_4 v_3$ can only be
in $E^-(UV)$. By planarity both edges cannot be in $E(UV)$. Therefore at least
$\min (|U_{*3}||V_{1*}|, |U_{*4}||V_{3*}|) = \min (4 \times 6, 6 \times 2) = 12$
patterns are incompatible. By planarity an outgoing edge from $u_4$ \wrt
$UV$ and the edges $u_2 v_2$ and $u_3 v_1$ cannot exist together in $E(UV)$.
So at least
$$ \min (|U_{*4}| |\emptyset|, |U_{*2}| |V_{2*}|, |U_{*3}| |V_{3*}| )
= \min (6 \times 1, 2 \times 4, 4 \times 2) = 6$$
patterns are incompatible.
Overall, at least $20 + 12 + 12 + 6 = 50$ patterns are incompatible.
\begin{figure}[htbp]
 \centering
 \includegraphics[scale=0.83]{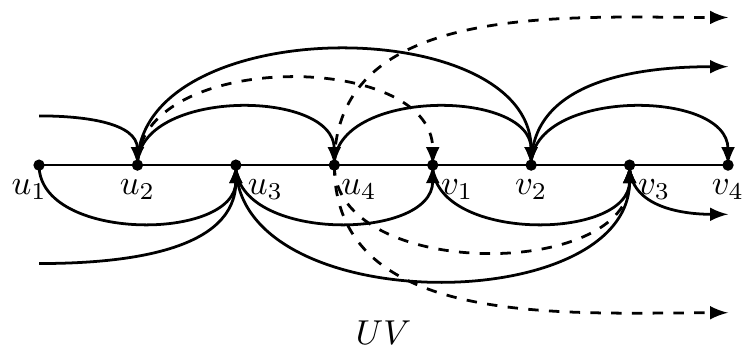}
 \hspace{0.2cm}
 \includegraphics[scale=0.83]{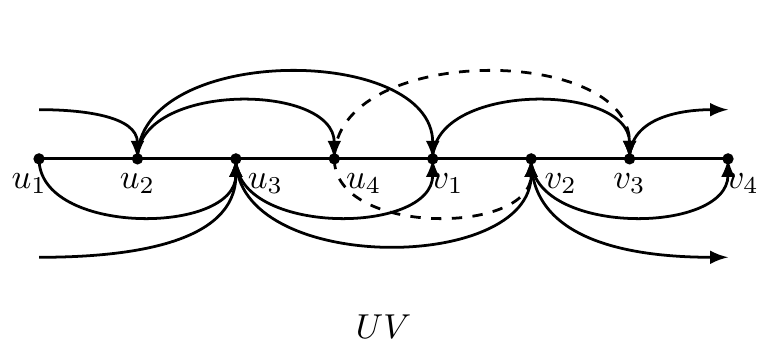}
 \caption{Left: $V$ is $A$. Right: $V$ is $A^R$.}
 \label{fig:f48}
\end{figure}

\medskip
\emph{Case 2:} $V$ is $A^R$, see Figure~\ref{fig:f48}\,(right).
By planarity $u_2 v_2$ and $u_3 v_3$ are not in $E(UV)$. So at least
$|U_{*2}||V_{2*}| + |U_{*3}||V_{3*}| = 2 \times 4 + 4 \times 2 = 16$
patterns are incompatible. Also $u_2 v_1$ and $u_4 v_3$ can only be in
$E^+(UV)$. By planarity both edges cannot be in $E(UV)$. Hence at least
$\min (|U_{*2}||V_{1*}|, |U_{*4}||V_{3*}|) = \min (2 \times 6, 6 \times 2) = 12$
patterns are incompatible. Similarly $u_3 v_1$ and $u_4 v_2$ can only be in
$E^-(UV)$. By planarity both edges cannot be in $E(UV)$. Hence at least
$\min (|U_{*3}||V_{1*}|, |U_{*4}||V_{2*}|) = \min (4 \times 6, 6 \times 4) = 24$
patterns are incompatible. Overall, at least $16 + 12 + 24 = 52$
patterns are incompatible.
\end{proof}

\begin{figure}[htbp]
\begin{minipage}{0.5\textwidth}
\centering
\includegraphics[scale=0.83]{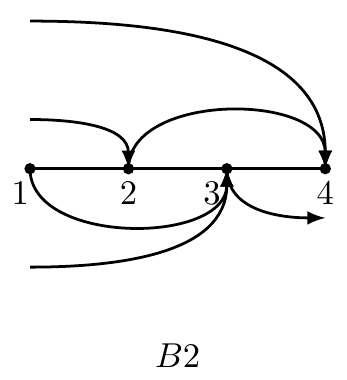}
\includegraphics[scale=0.83]{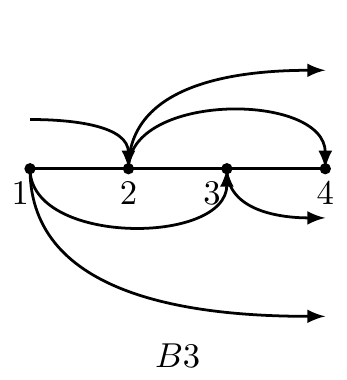}
\end{minipage}
\begin{minipage}{0.4\textwidth}
\centering
\includegraphics[scale=0.83]{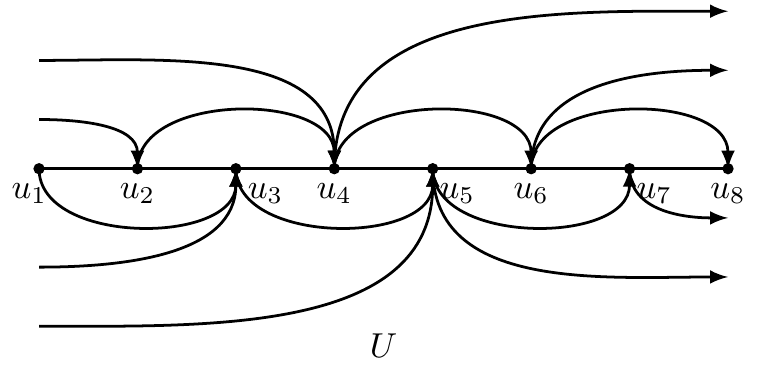}
\end{minipage}
\caption{$U$ has $120$ patterns. The $24$ missing patterns are
$123678$, $12367$, $12368$, $1236$, $123$, $13678$, $1367$, $1368$, $136$, $13$,
$23678$, $2367$, $2368$, $236$, $23$,
$3678$, $367$, $368$, $36$, $3$,
$678$, $67$, $68$, $6$.}
\label{fig:f49}
\end{figure}
\begin{lemma}\label{lem:summary-group-of-8}
Every group on 8 vertices has at most 120 incidence patterns,
and this bound is the best possible. Consequently, $p_8=120$.
\end{lemma}
\begin{proof}
A group of $8$, denoted by $UV$, where $U$ and $V$ are the groups
induced by the first and last four vertices of $UV$, respectively.
If $|I(U)| \leq 9$ or $|I(V)| \leq 9$, then
$|I(UV)| \leq |I(U)| \cdot |I(V)| \leq 9 \times 13 = 117$ by
Lemma~\ref{lem:group-with-at-least-13}.
Otherwise, Lemmas~\ref{lem:V-has-10}--\ref{lem:V-has-13} show that $|I(UV)| \leq 120$.

Consider the group $(U,E^-(U),E^+(U))$ of $8$ vertices depicted in Fig.~\ref{fig:f49}\,(right).
The first and second half of $U$ are the groups $B2$ and $B3$
in Fig.~\ref{fig:f49}\,(left), each with $12$ patterns.
Observe that exactly $24$ patterns are incompatible, thus $U$ has exactly
$|I(B2)| \cdot |I(B3)| - 24 = 12 \times 12 - 24 = 120$ patterns.
Aside from reflections, the extremal group of $8$ vertices in Fig.~\ref{fig:f49}\,(right)
is unique.
\end{proof}

\section{Groups of $9$, $10$ and $11$ vertices via computer search} \label{sec:groups-of-9-10}

The application of the same fingerprinting technique to groups of $9$, $10$  and $11$ vertices
via a computer program (included in Section~\ref{sec:source}) shows that
\begin{itemize}
\item A group of $9$ allows at most $201$ incidence patterns
(listed in Section~\ref{sec:extreme} of the Appendix);
the extremal configuration appears in Fig.~\ref{fig:f80}.
This yields the upper bound of $O(n^3 201^{n/9}) = O(1.8027^n)$ for the number
of monotone paths in an $n$-vertex triangulation.
Aside from reflections, the extremal group of $9$ vertices in Fig.~\ref{fig:f80}\,(left)
is unique.
\item A group of $10$ allows at most $346$ incidence patterns
(listed in Section~\ref{sec:extreme} of the Appendix);
the extremal configuration appears in Fig.~\ref{fig:f86}.
This yields the upper bound of $O(n^3 346^{n/10}) = O(1.7944^n)$ for the number
of monotone paths in an $n$-vertex triangulation, as given in Theorem~\ref{thm:upper}.
Aside from reflections, the extremal group of $10$ vertices in Fig.~\ref{fig:f86}
is unique.
\item A group of $11$ allows at most $591$ incidence patterns
(listed in Section~\ref{sec:extreme} of the Appendix);
the extremal configuration appears in Fig.~\ref{fig:f87}.
This yields the upper bound of $O(n^3 591^{n/11}) = O(1.7864^n)$ for the number
of monotone paths in an $n$-vertex triangulation.
Aside from reflections, the extremal group of $11$ vertices in Fig.~\ref{fig:f87}\,(left)
is unique.
\end{itemize}

\begin{figure}[htbp]
 \centering
 \includegraphics[scale=0.73]{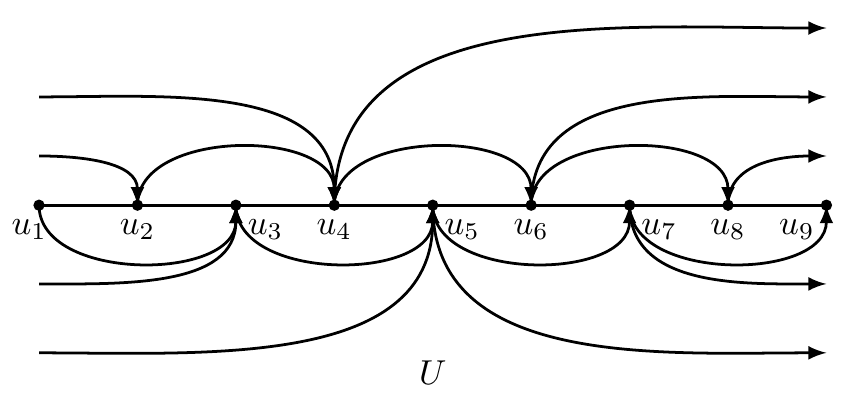}
 \hspace{2mm}
 \includegraphics[scale=0.73]{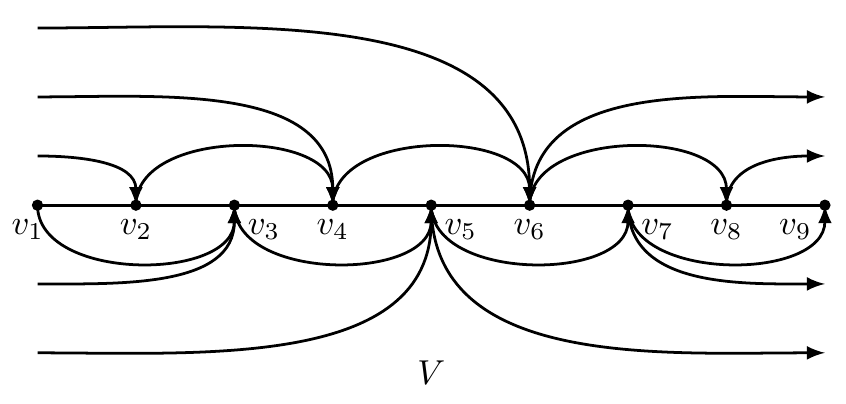}
 \caption{Groups $U$ and $V$ (hence also $U^R$ and $V^R$)
 are the only groups of $9$ vertices with $201$ incidence patterns.
 Observe that $V$ is the reflection of $U$ in the $y$-axis.}
\label{fig:f80}
\end{figure}
\begin{figure}[htbp]
 \centering
 \includegraphics[scale=0.77]{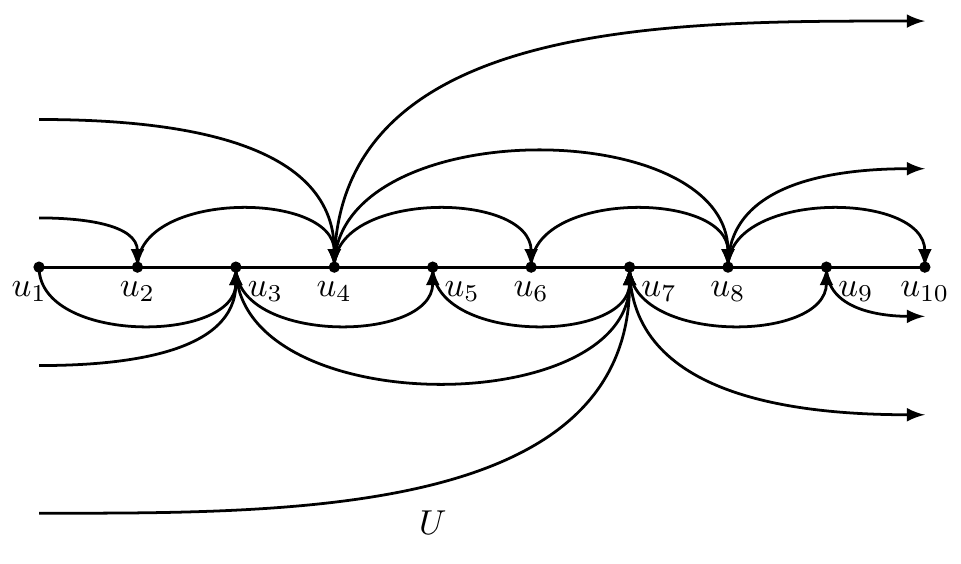}
 \caption{Group $U$ (hence also $U^R$) is the only group of $10$ vertices
 with $346$ incidence patterns.
Observe that the reflection of $U$ in the $y$-axis is $U^R$.}
\label{fig:f86}
\end{figure}
\begin{figure}[htbp]
 \centering
 \includegraphics[scale=0.73]{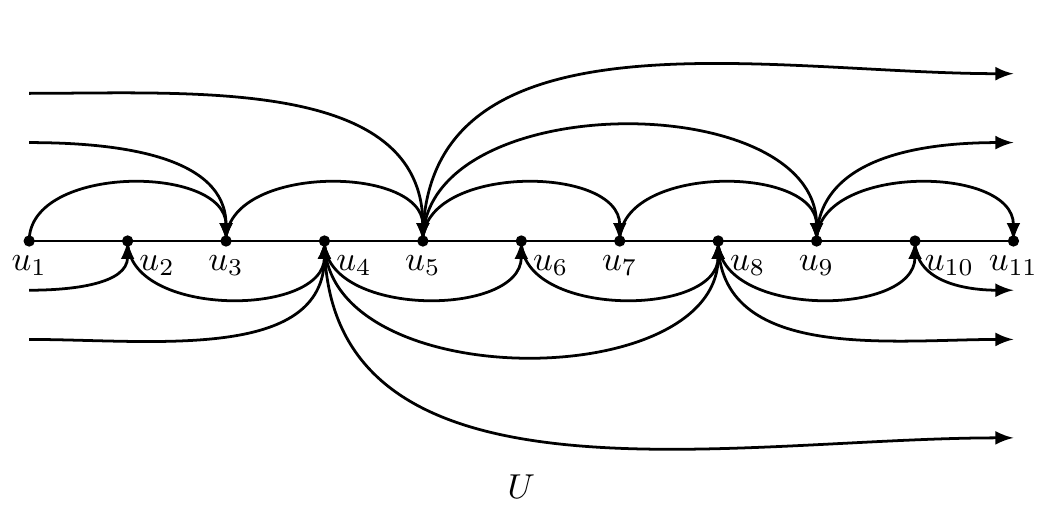}
 \hspace{2mm}
 \includegraphics[scale=0.73]{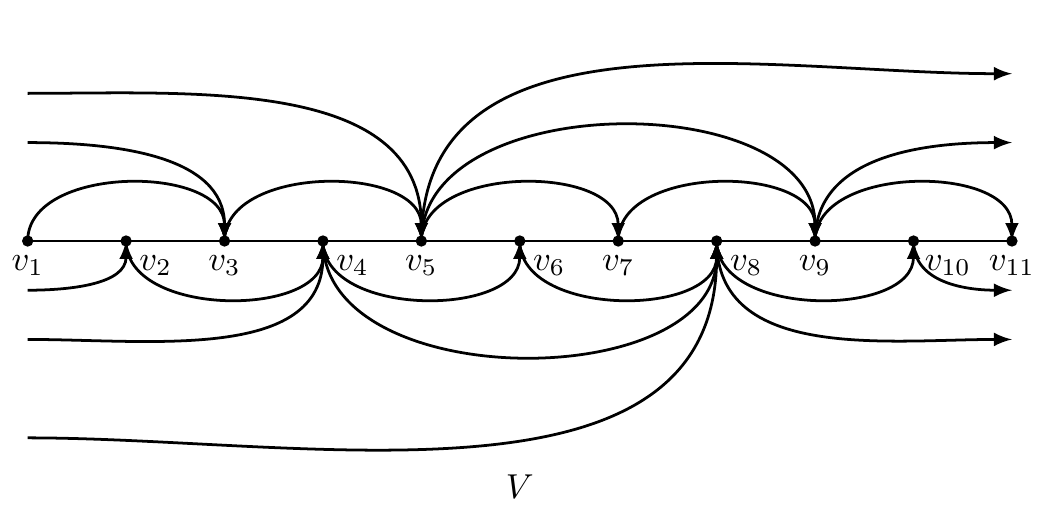}
 \caption{Groups $U$ and $V$ (hence also $U^R$ and $V^R$)
 are the only groups of $11$ vertices with $591$ incidence patterns.
 Observe that $V$ is the reflection of $U$ in the $y$-axis.}
\label{fig:f87}
\end{figure}

To generate all groups of $k$ vertices, the program first generates all
possible \emph{sides} of $k$ vertices, essentially by brute force.
A side of $k$ vertices $V=\{ v_1, \ldots v_k \}$
is represented by a directed planar graph with $k+2$ vertices, where the
edges $v_0 v_i$ and $v_i v_{k+1}$, for $1\leq i\leq k$, denote an incoming
edge into $v_i$ and an outgoing edge from $v_i$, respectively.
The edge $v_0 v_{k+1}$ represents the $\emptyset$ pattern.
Note that $\xi_0 \cup v_0 v_{k+1}$ forms a plane cycle on $k+2$
vertices in the underlying undirected graph.
Therefore, $E^+(V)$ and $E^-(V)$ can each have at most $(k+2)-3 = k-1$ edges.
After all possible sides are generated, the program combines all pairs
of sides with no common inner edge to generate a group $(V,E^-(V),E^+(V))$.
For each generated group, the program calculates the corresponding number of patterns
and in the end returns the group with the maximum number of patterns.

\paragraph{Remark.} It is interesting to observe how the structure of the unique
extremal groups of $9$, $10$ and $11$ vertices (depicted in Figures~\ref{fig:f80},~\ref{fig:f86} and~\ref{fig:f87})
matches that of the current best lower bound construction illustrated in
Fig.~\ref{fig:lb-constr}\,(right).

\section{Algorithm for counting and enumeration of monotone paths} \label{sec:compute}

\paragraph{Counting and enumeration of $x$-monotone paths.}
Let $G=(V,E)$ be a plane geometric graph with $n$ vertices.
We first observe that the number of $x$-monotone paths in $G$
can be computed by a sweep-line algorithm.
For every vertex $v\in V$, denote by $m(v)$ the number
of (directed) nonempty $x$-monotone paths that end at $v$.
Sweep a vertical line $\ell$ from left to right, and whenever $\ell$ reaches
a vertex $v$, we compute $m(v)$ according to the relation
$$ m(v)=\sum_{q \in L(v)} [m(q)+1], $$
where $L(v)$ denotes the set of neighbors of vertex $v$ in $G$ that lie to the left of $v$.
The total number of $x$-monotone paths in $G$ is $\sum_{v \in V} m(v)$.
For every $v \in V$, the computation of $m(v)$ takes $O(\deg(v))$ time,
thus computing $m(v)$ for all $v \in V$ takes $\sum_{v\in V} O(\deg(v)) =O(n)$ time.
The algorithm for computing the number of $x$-monotone paths in $G$ takes
$O(n)$ time if the vertices are in sorted order, or $O(n \log{n})$ time otherwise.

The sweep-line algorithm can be adapted to enumerate the $x$-monotone paths in $G$
in $O(n\log n+K)$ time, where $K$ is the sum of the lengths of all $x$-monotone paths.
For every vertex $v\in V$, denote by $M(v)$ the set of (directed)
nonempty $x$-monotone paths that end at $v$. When the vertical sweep-line $\ell$
reaches a vertex $v$, we compute $M(v)$ according to the relation
$$ M(v)=\bigcup_{q \in L(v)} \{(q,v)\} \cup \{ p\oplus(q,v): p\in M(q)\}, $$
where the concatenation of two paths, $p_1$ and $p_2$, is denoted by $p_1 \oplus p_2$.
The set of all $x$-monotone paths is $\bigcup_{v\in V} M(v)$.

The number (resp., set) of $\mathbf{u}$-monotone paths in $G$ for
any direction $\mathbf{u} \in \mathbb{R}^2\setminus \{\mathbf{0}\}$
can be computed in a similar way, using the
counts $m_{\mathbf{u}}(v)$ (resp., sets $M_{\mathbf{u}}(v)$)
and the neighbor sets $L_{\mathbf{u}}(v)$, instead;
the overall time per direction is $O(n)$ (resp., $O(n +K_{\mathbf{u}})$)
if the vertices are in sorted order, or $O(n \log{n})$
(resp., $O(n\log n +K_{\mathbf{u}})$) otherwise.

\paragraph{Computing all monotone paths.}
For computing the total number of monotone paths over all directions
$\mathbf{u} \in \mathbb{R}^2\setminus \{\mathbf{0}\}$, some care is required\footnote{The
algorithm has been revised, as some of the ideas were implemented incorrectly in the
earlier conference version.}.
As shown subsequently, it suffices to consider monotone paths in at most $2|E|$ directions,
one direction between any two consecutive directions orthogonal to the edges in $E$
(each edge yields two opposite directions). It is worth noting however, that it does \emph{not}
suffice to consider only directions parallel or orthogonal to the edges of $G$, \ie, the union of
monotone paths over all such directions may not contain all monotone paths.
Second, observe that once a sufficient set of directions is established,
we cannot simply sum up the number of monotone paths over all such directions,
since a monotone path in $G$ may be monotone in several of these directions.

We first compute a set of directions that is \emph{sufficient} for our purpose, \ie,
every monotone path is monotone with respect to at least one of these directions.
We then consider these directions sorted by angle, and for each new direction,
we compute the number of \emph{new} paths.
In fact, we count \emph{directed} monotone paths, that is, each path will be counted
twice, as traversed in two opposite directions. We now proceed with the details.

Partition the edge set $E$ of $G$ into subsets of parallel edges;
note that each subset yields two opposite directions.
Since $|E| \leq 3n-6$, the edges are partitioned into at most $|E| \leq 3n-6$ subsets.
Let $\D$ be a set of direction vectors $\overrightarrow{ab}$
of the edges $(a,b)\in E$, one from each subset of parallel edges.
Let $\D^\perp$ be a set of vectors obtained by rotating each vector  in $\D$
counterclockwise by $\pi/2$ and by $-\pi/2$; note that $|\D^\perp|=2|\D| \leq 2|E|\leq 6n-12$.
Sort the vectors in $\D^\perp$ by their arguments\footnote{The
\emph{argument} of a vector $\mathbf{u}\in \mathbb{R}^2\setminus \{\mathbf{0}\}$
is the angle measure in $[0, 2\pi)$ of the minimum counterclockwise
rotation that carries the positive $x$-axis to the ray spanned by $\mathbf{u}$.}
in cyclic order, and let $\U$ be a set of vector sums of all pairs of consecutive vectors
in $\D^\perp$; clearly $|\U| \leq |\D^\perp| \leq 6n-12$.
We show that $\U$ is a sufficient set of directions. Indeed, consider a path $\xi$
monotone with respect to some direction $\mathbf{u} \in \mathbb{R}^2\setminus \{\mathbf{0}\}$.
Observe that the edges in $\xi$ cannot be orthogonal to $\mathbf{u}$; it follows that $\xi$
is still monotone with respect to at least one of the two adjacent directions in $\U$
(closest to $\mathbf{u}$).

We next present the algorithm. Sort the vectors in $\U$ by their arguments, in $O(n \log n)$ time.
Let $\mathbf{u}_0\in \U$ be an vector of minimum argument in $\U$.
We first compute the number of $\mathbf{u}_0$-monotone directed paths in $G$,
$\sum_{v \in V} m_{\mathbf{u_0}}(v)$,
by the sweep-line algorithm described above in $O(n \log n)$ time.

Consider the directions $\mathbf{u}\in \U \setminus \{\mathbf{u}_0\}$,
sorted by increasing arguments. For each $\mathbf{u}$, we maintain the
number of directed paths in $G$ that are monotone in some direction
between $\mathbf{u}_0$ and $\mathbf{u}$
(implicitly, by the sum of parameters $\gamma$ defined below).
For each new direction $\mathbf{u}$, exactly one subset of parallel edges, denoted $E_{\mathbf{u}}$,
becomes $\mathbf{u}$-monotone (\ie, it consists of $\mathbf{u}$-monotone edges).
Therefore, it is enough to count the number of $\mathbf{u}$-monotone paths that
traverse some edge in $E_{\mathbf{u}}$.

\paragraph{Counting the monotone paths in $G$.}
These paths can be counted by sweeping $G$ with a line $\ell$ orthogonal to $\mathbf{u}$:
Sort the vertices in direction $\mathbf{u}$ (ties are broken arbitrarily).
Compute two parameters for every vertex $v \in V$:
\begin{itemize} \itemsep 0pt
\item the number $m_{\mathbf{u}}(v)$ of $\mathbf{u}$-monotone paths that end at $v$,
\item the number $\gamma_{\mathbf{u}}(v)$ of $\mathbf{u}$-monotone paths
that end at $v$ \emph{and} contain some edge from $E_{\mathbf{u}}$.
\end{itemize}
When reaching vertex $v$, the sweep-line algorithm computes the first parameter
$m_{\mathbf{u}}(v)$ according to the relation:
\begin{equation}\label{eq:m}
m_{\mathbf{u}}(v) = \sum_{q \in L_{\mathbf{u}}(v)} [m_{\mathbf{u}}(q)+1] .
\end{equation}
The second parameter $\gamma_{\mathbf{u}}(v)$ is computed as follows:
\begin{equation}\label{eq:gamma}
\gamma_{\mathbf{u}}(v)=
  \begin{cases}
 1 + m_{\mathbf{u}}(a) + \sum_{q \in L_{\mathbf{u}}(v)\setminus \{a\}} \gamma_{\mathbf{u}}(q) &
 \text{ if } \exists (a,v) \in E_{\mathbf{u}} \text{ with }
 \langle \overrightarrow{av},\mathbf{u}\rangle>0,\\
\sum_{q \in L_{\mathbf{u}}(v)} \gamma_{\mathbf{u}}(q) & \text{ otherwise. }
  \end{cases}
\end{equation}

The total number of monotone paths, returned by the algorithm in the end is
$$ \sum_{v \in V} \left( m_{\mathbf{u_0}}(v) +
\sum_{\mathbf{u} \in \U \setminus \{\mathbf{u_0}\}} \gamma_{\mathbf{u}}(v)\right). $$

We next show that the sorted order of vertices in the $O(n)$ directions in
$\U$ can be computed in $O(n^2)$ time. Consider the duality transform,
where every point $v=(a,b)$ is mapped to a dual line $v^*:y=ax-b$,
and every line $\ell:y=ax-b$ is mapped to a dual point $\ell^*=(a,b)$.
It is known that the duality preserves the above-below relationship between
points and lines~\cite[Ch.~8]{BCKO08}. Note that the lines of slope $a$ are
mapped to dual points on the vertical line $x=a$. Consequently, when we sweep
$V$ by a line of slope $a$ in direction $\mathbf{u}=(1,-a)$, we encounter
the points in $V$ in the order determined by $y$-coordinates of the intersections
of the vertical line $x=a$ with the dual lines in $V^*=\{v^*: v\in V\}$.

Let $\A$ be the arrangement of the $n$ dual lines in $V^*$ plus the
$O(n)$ vertical lines corresponding to the slopes of the vectors in $\U$.
The arrangement $\A$ of these $O(n)$ lines has $O(n^2)$ vertices,
and can be computed in $O(n^2)$ time~\cite[Ch.~8]{BCKO08}. By tracing the
vertical line corresponding to each vector $\mathbf{u}\in \U$ in $\A$,
we find its intersection points with the dual lines in $V^*$, sorted by $y$-coordinates,
in $O(n)$ time. Since $|\U| =O(n)$, the total running time of the algorithm is $O(n^2)$,
as claimed.

\paragraph{Enumerating the monotone paths in $G$.}
To this end we adapt the formulae~\eqref{eq:m} and \eqref{eq:gamma} to sets of monotone
paths in a straightforward manner. For every vertex $v \in V$, we compute two sets:
\begin{itemize} \itemsep 0pt
\item the set $M_{\mathbf{u}}(v)$ of $\mathbf{u}$-monotone paths that end at $v$,
\item the set $\Gamma_{\mathbf{u}}(v)$ of $\mathbf{u}$-monotone paths
that end at $v$ \emph{and} contain some edge from $E_{\mathbf{u}}$.
\end{itemize}
When reaching vertex $v$, the sweep-line algorithm computes $M_{\mathbf{u}}(v)$
according to the relation:
\begin{equation}\label{eq:M}
M_{\mathbf{u}}(v) = \bigcup_{q \in L_{\mathbf{u}}(v)} \{(q,v)\} \cup \{p\oplus (q,v): p\in M(q)\} .
\end{equation}
The set $\Gamma_{\mathbf{u}}(v)$ is computed as follows:
\begin{equation}\label{eq:Gamma}
\Gamma_{\mathbf{u}}(v)=
\begin{cases}
 \{(a,v)\} \cup \{p\oplus (a,v): p\in M_\mathbf{u}(a)\} \, \cup & \\
 ~~~~~~~~~~\bigcup_{q\in L_{\mathbf{u}}(v)\setminus \{a\}} \{p\oplus (q,v): p\in \Gamma_\mathbf{u}(q)\}
 &  \text{ if } \exists (a,v) \in E_{\mathbf{u}} \text{ with }
 \langle \overrightarrow{av},\mathbf{u}\rangle>0,\\
 \bigcup_{q\in L_{\mathbf{u}}(v)} \{p\oplus (q,v): p\in \Gamma_\mathbf{u}(q)\} & \text{ otherwise. }
  \end{cases}
\end{equation}

Now the set of directed monotone paths in $G$ is
$$ \bigcup_{v\in V} \left(M_{\mathbf{u}_0(v)}\cup
\bigcup_{\mathbf{u} \in \U\setminus \{\mathbf{u}_0\}} \Gamma_{\mathbf{u}}(v) \right), $$
where every undirected monotone path appears twice: once in each direction.
This completes the proof of Theorem~\ref{thm:compute}.
\qed

\section{Concluding remarks} \label{sec:remarks}

A path is \emph{simple} if it has no repeated vertices;
obviously every monotone path is simple.
A directed polygonal path $\xi=(v_1,v_2,\ldots, v_t)$ in $\mathbb{R}^d$ is \emph{weakly monotone} if
there exists a nonzero vector $\mathbf{u}\in \mathbb{R}^d$ that
has a nonnegative inner product with every directed edge of $\xi$, that is,
$\langle \overrightarrow{v_iv_{i+1}},\mathbf{u}\rangle \geq 0$ for $i=1,\ldots, t-1$.
In many applications such as local search, a weakly monotone path may be as good as
a monotone one, since both guarantee that the objective function is nondecreasing.

It therefore appears as a natural problem to find a tight asymptotic bound on the
maximum number of weakly monotone simple paths over all plane geometric graphs with $n$ vertices.
As for monotone paths, it is easy to see that triangulations maximize the number of such paths.
Specifically, let $\beta_n$ denote the maximum number of weakly monotone simple paths
in an $n$-vertex triangulation.

We clearly have $\beta_n \geq \lambda_n$, and so $\beta_n =\Omega(1.7003^n)$.
However, $\beta_n$ could in principle grow faster than $\lambda_n$.
For instance, we have $\beta_4 \geq 18$: label the vertices of the unit
square $[0,1]^2$ by $1,2,3,4$, counterclockwise starting at the origin
and add the diagonal $13$ to obtain a triangulation. Observe that the following
$18$ (undirected) paths are weakly monotone:
$12$, $23$, $34$, $41$,
$123$, $234$, $341$, $412$, $132$, $134$, $312$, $314$,
$1234$, $2341$, $3412$, $4123$, $2134$, $2314$, hence $\beta_4 \geq 18$.
However, a detailed analysis (omitted here) shows that $\lambda_4 \leq 16$.
We conclude with the following open problems.
\begin{enumerate} \itemsep 0pt
 \item What upper and lower bounds can be derived for $\beta_n$?
 Is $\beta_n =\omega(\lambda_n)$ ?
 \item What can be said about counting and enumeration of weakly monotone paths
 in a given plane geometric graph?
\end{enumerate}

\appendix

\section{Extremal configurations}\label{sec:extreme}

\paragraph{The groups of 4 vertices with 12 and 11 patterns.}
There are exactly $4$ groups with exactly $12$ incidence patterns; see Fig.~\ref{fig:f2}
(reflections about the $x$-axis are not shown).
\begin{figure}[hbtp]
\centering
  \includegraphics[scale=0.83]{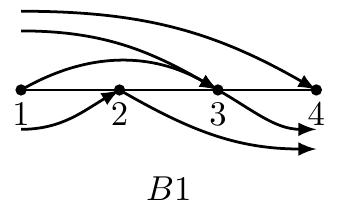}
 \hspace{0.1cm}
  \includegraphics[scale=0.83]{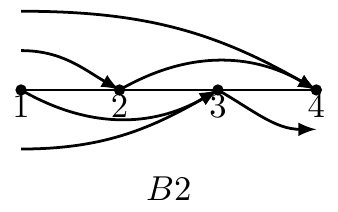}
  \hspace{0.1cm}
  \includegraphics[scale=0.83]{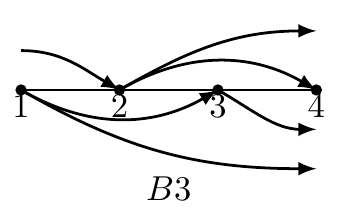}
 \hspace{0.1cm}
  \includegraphics[scale=0.83]{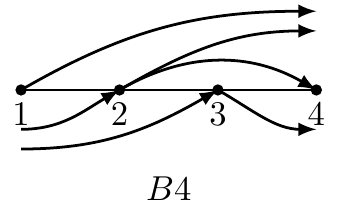}
 \caption{$B1-B4$ are the only four groups with $12$ incidence patterns.}
 \label{fig:f2}
\end{figure}

\begin{enumerate} \itemsep 0pt
\item[] $I(B1)$: $\emptyset,1234,123,12,134,13,234,23,2,34,3,4$.
\item[] $I(B2)$: $\emptyset,1234,123,124,134,13,234,23,24,34,3,4$.
\item[] $I(B3)$: $\emptyset,1234,123,124,12,134,13,1,23,234,24,2$.
\item[] $I(B4)$: $\emptyset,1234,123,124,12,1,23,234,24,2,34,3$.
\end{enumerate}

There are exactly $20$ groups with exactly $11$ incidence patterns; see Fig.~\ref{fig:f6}
(reflections around the $x$-axis are not shown).

\begin{figure}[H]
\centering
  \includegraphics[scale=0.8]{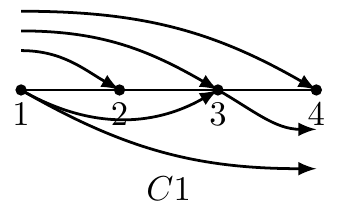}
 \hspace{0.1cm}
  \includegraphics[scale=0.8]{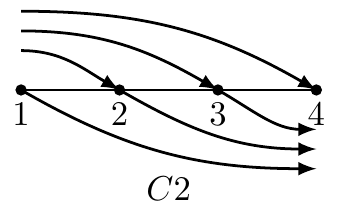}
  \hspace{0.1cm}
  \includegraphics[scale=0.8]{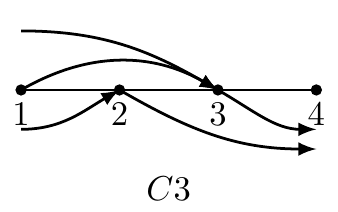}
 \hspace{0.1cm}
  \includegraphics[scale=0.8]{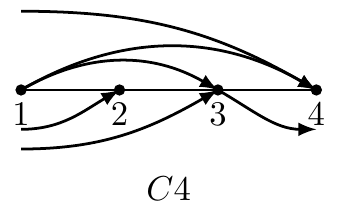}
\end{figure}

\begin{figure}[H]
\centering
  \includegraphics[scale=0.8]{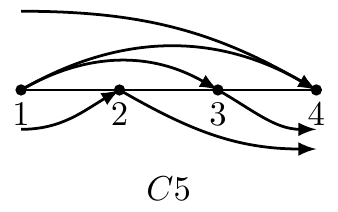}
 \hspace{0.1cm}
  \includegraphics[scale=0.8]{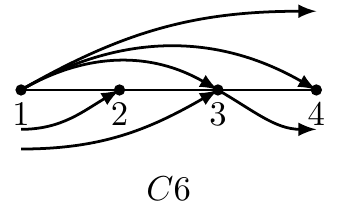}
  \hspace{0.1cm}
  \includegraphics[scale=0.8]{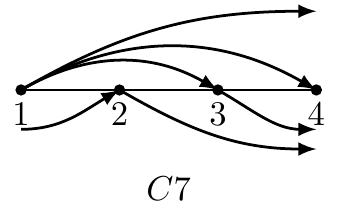}
 \hspace{0.1cm}
  \includegraphics[scale=0.8]{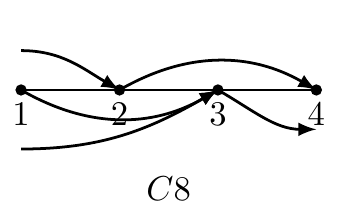}
\end{figure}

\begin{figure}[H]
\centering
  \includegraphics[scale=0.8]{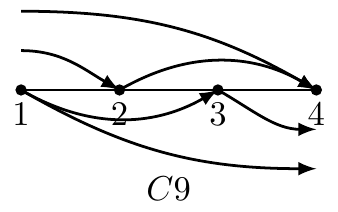}
 \hspace{0.1cm}
  \includegraphics[scale=0.8]{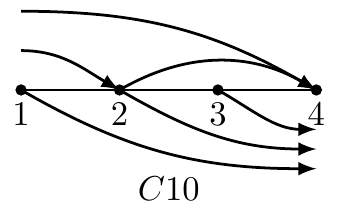}
  \hspace{0.1cm}
  \includegraphics[scale=0.8]{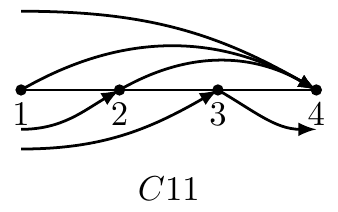}
 \hspace{0.1cm}
  \includegraphics[scale=0.8]{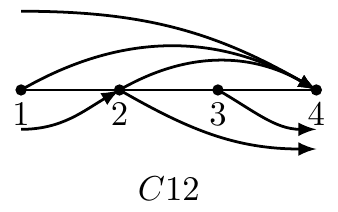}
\end{figure}

\begin{figure}[H]
\centering
  \includegraphics[scale=0.8]{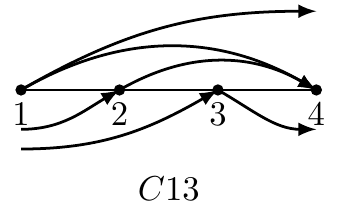}
 \hspace{0.1cm}
  \includegraphics[scale=0.8]{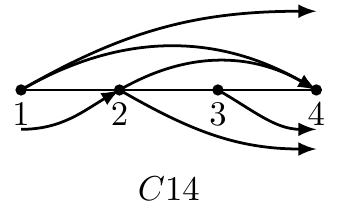}
  \hspace{0.1cm}
  \includegraphics[scale=0.8]{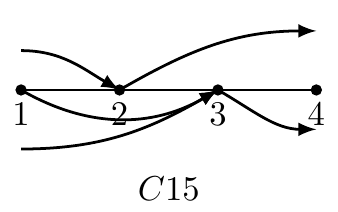}
 \hspace{0.1cm}
  \includegraphics[scale=0.8]{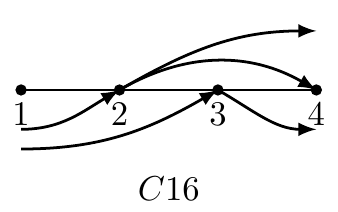}
\end{figure}

\begin{figure}[H]
\centering
  \includegraphics[scale=0.8]{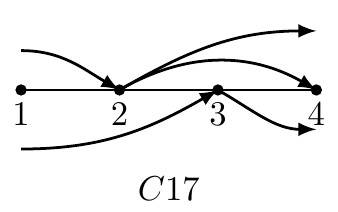}
 \hspace{0.1cm}
  \includegraphics[scale=0.8]{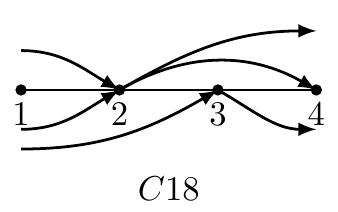}
  \hspace{0.1cm}
  \includegraphics[scale=0.8]{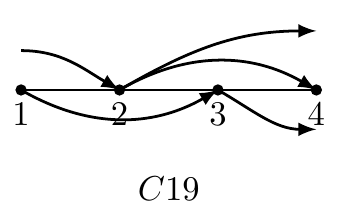}
 \hspace{0.1cm}
  \includegraphics[scale=0.8]{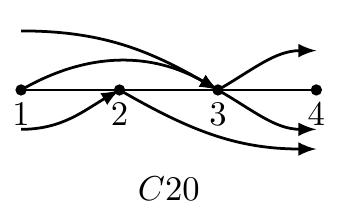}
 \caption{$C1-C20$ are the only $20$ groups with $11$ incidence patterns.}
 \label{fig:f6}
\end{figure}

\noindent
$I(C1): \emptyset,1234,123,134,13,1,234,23,34,3,4.$\\
$I(C2): \emptyset,1234,123,12,1,234,23,2,34,3,4.$\\
$I(C3): \emptyset,1234,123,12,134,13,234,23,2,34,3.$\\
$I(C4): \emptyset,1234,123,134,13,14,234,23,34,3,4.$\\
$I(C5): \emptyset,1234,123,12,134,13,14,234,23,2,4.$\\
$I(C6): \emptyset,1234,123,134,13,14,1,234,23,34,3.$\\
$I(C7): \emptyset,1234,123,12,134,13,14,1,234,23,2.$\\
$I(C8): \emptyset,1234,123,124,134,13,234,23,24,34,3.$\\
$I(C9): \emptyset,1234,123,124,134,13,1,234,23,24,4.$\\
$I(C10): \emptyset,1234,123,124,12,1,234,23,24,2,4.$\\
$I(C11): \emptyset,1234,123,124,14,234,23,24,34,3,4.$\\
$I(C12): \emptyset,1234,123,124,12,14,234,23,24,2,4.$\\
$I(C13): \emptyset,1234,123,124,14,1,234,23,24,34,3.$\\
$I(C14): \emptyset,1234,123,124,12,14,1,234,23,24,2.$\\
$I(C15): \emptyset,1234,123,12,134,13,234,23,2,34,3.$\\
$I(C16): \emptyset,1234,123,124,12,234,23,24,2,34,3.$\\
$I(C17): \emptyset,1234,123,124,12,234,23,24,2,34,3.$\\
$I(C18): \emptyset,1234,123,124,12,234,23,24,2,34,3.$\\
$I(C19): \emptyset,1234,123,124,12,134,13,234,23,24,2.$\\
$I(C20): \emptyset,1234,123,12,134,13,234,23,2,34,3.$

\paragraph{The $201$ incidence patterns of the extremal group
$U$ of $9$ vertices depicted in Fig.~\ref{fig:f80}:}
$\emptyset$, $1 2 3 4 5 6 7 8 9 $, $1 2 3 4 5 6 7 8 $, $1 2 3 4 5 6 7 9 $,
$1 2 3 4 5 6 7 $, $1 2 3 4 5 6 8 9 $, $1 2 3 4 5 6 8 $, $1 2 3 4 5 6 $,
$1 2 3 4 5 7 8 9 $, $1 2 3 4 5 7 8 $, $1 2 3 4 5 7 9 $, $1 2 3 4 5 7 $,
$1 2 3 4 5 $, $1 2 3 4 6 7 8 9 $, $1 2 3 4 6 7 8 $, $1 2 3 4 6 7 9 $,
$1 2 3 4 6 7 $, $1 2 3 4 6 8 9 $, $1 2 3 4 6 8 $, $1 2 3 4 6 $,
$1 2 3 4 $, $1 2 3 5 6 7 8 9 $, $1 2 3 5 6 7 8 $, $1 2 3 5 6 7 9 $,
$1 2 3 5 6 7 $, $1 2 3 5 6 8 9 $, $1 2 3 5 6 8 $, $1 2 3 5 6 $, $1 2 3 5 7 8 9 $,
$1 2 3 5 7 8 $, $1 2 3 5 7 9 $, $1 2 3 5 7 $, $1 2 3 5 $,
$1 2 4 5 6 7 8 9 $, $1 2 4 5 6 7 8 $, $1 2 4 5 6 7 9 $, $1 2 4 5 6 7 $,
$1 2 4 5 6 8 9 $, $1 2 4 5 6 8 $, $1 2 4 5 6 $, $1 2 4 5 7 8 9 $,
$1 2 4 5 7 8 $, $1 2 4 5 7 9 $, $1 2 4 5 7 $, $1 2 4 5 $, $1 2 4 6 7 8 9 $,
$1 2 4 6 7 8 $, $1 2 4 6 7 9 $, $1 2 4 6 7 $, $1 2 4 6 8 9 $, $1 2 4 6 8 $,
$1 2 4 6 $, $1 2 4 $, $1 3 4 5 6 7 8 9 $, $1 3 4 5 6 7 8 $, $1 3 4 5 6 7 9 $,
$1 3 4 5 6 7 $, $1 3 4 5 6 8 9 $, $1 3 4 5 6 8 $, $1 3 4 5 6 $,
$1 3 4 5 7 8 9 $, $1 3 4 5 7 8 $, $1 3 4 5 7 9 $, $1 3 4 5 7 $,
$1 3 4 5 $, $1 3 4 6 7 8 9 $, $1 3 4 6 7 8 $, $1 3 4 6 7 9 $, $1 3 4 6 7 $,
$1 3 4 6 8 9 $, $1 3 4 6 8 $, $1 3 4 6 $, $1 3 4 $, $1 3 5 6 7 8 9 $,
$1 3 5 6 7 8 $, $1 3 5 6 7 9 $, $1 3 5 6 7 $, $1 3 5 6 8 9 $,
$1 3 5 6 8 $, $1 3 5 6 $, $1 3 5 7 8 9 $, $1 3 5 7 8 $, $1 3 5 7 9 $,
$1 3 5 7 $, $1 3 5 $, $2 3 4 5 6 7 8 9 $, $2 3 4 5 6 7 8 $, $2 3 4 5 6 7 9 $,
$2 3 4 5 6 7 $, $2 3 4 5 6 8 9 $, $2 3 4 5 6 8 $, $2 3 4 5 6 $,
$2 3 4 5 7 8 9 $, $2 3 4 5 7 8 $, $2 3 4 5 7 9 $, $2 3 4 5 7 $,
$2 3 4 5 $, $2 3 4 6 7 8 9 $, $2 3 4 6 7 8 $, $2 3 4 6 7 9 $, $2 3 4 6 7 $,
$2 3 4 6 8 9 $, $2 3 4 6 8 $, $2 3 4 6 $, $2 3 4 $, $2 3 5 6 7 8 9 $,
$2 3 5 6 7 8 $, $2 3 5 6 7 9 $, $2 3 5 6 7 $, $2 3 5 6 8 9 $, $2 3 5 6 8 $,
$2 3 5 6 $, $2 3 5 7 8 9 $, $2 3 5 7 8 $, $2 3 5 7 9 $, $2 3 5 7 $,
$2 3 5 $, $2 4 5 6 7 8 9 $, $2 4 5 6 7 8 $, $2 4 5 6 7 9 $,
$2 4 5 6 7 $, $2 4 5 6 8 9 $, $2 4 5 6 8 $, $2 4 5 6 $, $2 4 5 7 8 9 $,
$2 4 5 7 8 $, $2 4 5 7 9 $, $2 4 5 7 $, $2 4 5 $, $2 4 6 7 8 9 $, $2 4 6 7 8 $,
$2 4 6 7 9 $, $2 4 6 7 $, $2 4 6 8 9 $, $2 4 6 8 $, $2 4 6 $, $2 4 $,
$3 4 5 6 7 8 9 $, $3 4 5 6 7 8 $, $3 4 5 6 7 9 $, $3 4 5 6 7 $,
$3 4 5 6 8 9 $, $3 4 5 6 8 $, $3 4 5 6 $, $3 4 5 7 8 9 $, $3 4 5 7 8 $,
$3 4 5 7 9 $, $3 4 5 7 $, $3 4 5 $, $3 4 6 7 8 9 $, $3 4 6 7 8 $,
$3 4 6 7 9 $, $3 4 6 7 $, $3 4 6 8 9 $, $3 4 6 8 $, $3 4 6 $, $3 4 $,
$3 5 6 7 8 9 $, $3 5 6 7 8 $, $3 5 6 7 9 $, $3 5 6 7 $, $3 5 6 8 9 $,
$3 5 6 8 $, $3 5 6 $, $3 5 7 8 9 $, $3 5 7 8 $, $3 5 7 9 $, $3 5 7 $,
$3 5 $, $4 5 6 7 8 9 $, $4 5 6 7 8 $, $4 5 6 7 9 $, $4 5 6 7 $,
$4 5 6 8 9 $, $4 5 6 8 $, $4 5 6 $, $4 5 7 8 9 $, $4 5 7 8 $, $4 5 7 9 $,
$4 5 7 $, $4 5 $, $4 6 7 8 9 $, $4 6 7 8 $, $4 6 7 9 $, $4 6 7 $,
$4 6 8 9 $, $4 6 8 $, $4 6 $, $4 $, $5 6 7 8 9 $, $5 6 7 8 $, $5 6 7 9 $,
$5 6 7 $, $5 6 8 9 $, $5 6 8 $, $5 6 $, $5 7 8 9 $, $5 7 8 $, $5 7 9 $,
$5 7 $, $5 $.

\paragraph{The $346$ incidence patterns of the extremal group
$U$ of $10$ vertices depicted in Fig.~\ref{fig:f86}:}
\begin{sloppypar}
$\emptyset$, $1 2 3 4 5 6 7 8 9 10 $, $1 2 3 4 5 6 7 8 9 $, $1 2 3 4 5 6 7 8 10 $, 
$1 2 3 4 5 6 7 8 $, $1 2 3 4 5 6 7 9 10 $, $1 2 3 4 5 6 7 9 $, $1 2 3 4 5 6 7 $, 
$1 2 3 4 5 6 8 9 10 $, $1 2 3 4 5 6 8 9 $, $1 2 3 4 5 6 8 10 $, $1 2 3 4 5 6 8 $, 
$1 2 3 4 5 7 8 9 10 $, $1 2 3 4 5 7 8 9 $, $1 2 3 4 5 7 8 10 $, $1 2 3 4 5 7 8 $, 
$1 2 3 4 5 7 9 10 $, $1 2 3 4 5 7 9 $, $1 2 3 4 5 7 $, $1 2 3 4 6 7 8 9 10 $, 
$1 2 3 4 6 7 8 9 $, $1 2 3 4 6 7 8 10 $, $1 2 3 4 6 7 8 $, $1 2 3 4 6 7 9 10 $, 
$1 2 3 4 6 7 9 $, $1 2 3 4 6 7 $, $1 2 3 4 6 8 9 10 $, 
$1 2 3 4 6 8 9 $, $1 2 3 4 6 8 10 $, $1 2 3 4 6 8 $, $1 2 3 4 8 9 10 $, 
$1 2 3 4 8 9 $, $1 2 3 4 8 10 $, $1 2 3 4 8$, $1 2 3 4$, $1 2 3 5 6 7 8 9 10$, $1 2 3 5 6 7 8 9$, 
$1 2 3 5 6 7 8 10$, $1 2 3 5 6 7 8$, $1 2 3 5 6 7 9 10$, $1 2 3 5 6 7 9$, 
$1 2 3 5 6 7$, $1 2 3 5 6 8 9 10$, $1 2 3 5 6 8 9$, $1 2 3 5 6 8 10$, 
$1 2 3 5 6 8$, $1 2 3 5 7 8 9 10$, $1 2 3 5 7 8 9$, $1 2 3 5 7 8 10$, 
$1 2 3 5 7 8$, $1 2 3 5 7 9 10$, $1 2 3 5 7 9$, $1 2 3 5 7$, $1 2 3 7 8 9 10$, 
$1 2 3 7 8 9$, $1 2 3 7 8 10$, $1 2 3 7 8$, $1 2 3 7 9 10$, $1 2 3 7 9$, 
$1 2 3 7$, $1 2 4 5 6 7 8 9 10$, $1 2 4 5 6 7 8 9$, $1 2 4 5 6 7 8 10$, 
$1 2 4 5 6 7 8$, $1 2 4 5 6 7 9 10$, $1 2 4 5 6 7 9$, $1 2 4 5 6 7$, 
$1 2 4 5 6 8 9 10$, $1 2 4 5 6 8 9$, $1 2 4 5 6 8 10$, $1 2 4 5 6 8$, 
$1 2 4 5 7 8 9 10$, $1 2 4 5 7 8 9$, $1 2 4 5 7 8 10$, $1 2 4 5 7 8$, 
$1 2 4 5 7 9 10$, $1 2 4 5 7 9$, $1 2 4 5 7$, $1 2 4 6 7 8 9 10$, $1 2 4 6 7 8 9$, 
$1 2 4 6 7 8 10$, $1 2 4 6 7 8$, $1 2 4 6 7 9 10$, $1 2 4 6 7 9$, $1 2 4 6 7$, 
$1 2 4 6 8 9 10$, $1 2 4 6 8 9$, $1 2 4 6 8 10$, $1 2 4 6 8$, $1 2 4 8 9 10$, 
$1 2 4 8 9$, $1 2 4 8 10$, $1 2 4 8$, $1 2 4$, $1 3 4 5 6 7 8 9 10$, 
$1 3 4 5 6 7 8 9$, $1 3 4 5 6 7 8 10$, $1 3 4 5 6 7 8$, $1 3 4 5 6 7 9 10$, 
$1 3 4 5 6 7 9$, $1 3 4 5 6 7$, $1 3 4 5 6 8 9 10$, $1 3 4 5 6 8 9$, 
$1 3 4 5 6 8 10$, $1 3 4 5 6 8$, $1 3 4 5 7 8 9 10$, $1 3 4 5 7 8 9$, 
$1 3 4 5 7 8 10$, $1 3 4 5 7 8$, $1 3 4 5 7 9 10$, $1 3 4 5 7 9$, $1 3 4 5 7$, 
$1 3 4 6 7 8 9 10$, $1 3 4 6 7 8 9$, $1 3 4 6 7 8 10$, $1 3 4 6 7 8$, 
$1 3 4 6 7 9 10$, $1 3 4 6 7 9$, $1 3 4 6 7$, $1 3 4 6 8 9 10$, $1 3 4 6 8 9$, 
$1 3 4 6 8 10$, $1 3 4 6 8$, $1 3 4 8 9 10$, $1 3 4 8 9$, $1 3 4 8 10$, 
$1 3 4 8$, $1 3 4$, $1 3 5 6 7 8 9 10$, $1 3 5 6 7 8 9$, $1 3 5 6 7 8 10$, 
$1 3 5 6 7 8$, $1 3 5 6 7 9 10$, $1 3 5 6 7 9$, $1 3 5 6 7$, $1 3 5 6 8 9 10$, 
$1 3 5 6 8 9$, $1 3 5 6 8 10$, $1 3 5 6 8$, $1 3 5 7 8 9 10$, $1 3 5 7 8 9$, 
$1 3 5 7 8 10$, $1 3 5 7 8$, $1 3 5 7 9 10$, $1 3 5 7 9$, $1 3 5 7$, 
$1 3 7 8 9 10$, $1 3 7 8 9$, $1 3 7 8 10$, $1 3 7 8$, $1 3 7 9 10$, $1 3 7 9$, 
$1 3 7$, $2 3 4 5 6 7 8 9 10$, $2 3 4 5 6 7 8 9$, $2 3 4 5 6 7 8 10$, 
$2 3 4 5 6 7 8$, $2 3 4 5 6 7 9 10$, $2 3 4 5 6 7 9$, $2 3 4 5 6 7$, 
$2 3 4 5 6 8 9 10$, $2 3 4 5 6 8 9$, $2 3 4 5 6 8 10$, $2 3 4 5 6 8$, 
$2 3 4 5 7 8 9 10$, $2 3 4 5 7 8 9$, $2 3 4 5 7 8 10$, $2 3 4 5 7 8$, 
$2 3 4 5 7 9 10$, $2 3 4 5 7 9$, $2 3 4 5 7$, $2 3 4 6 7 8 9 10$, 
$2 3 4 6 7 8 9$, $2 3 4 6 7 8 10$, $2 3 4 6 7 8$, $2 3 4 6 7 9 10$, 
$2 3 4 6 7 9$, $2 3 4 6 7$, $2 3 4 6 8 9 10$, $2 3 4 6 8 9$, $2 3 4 6 8 10$, 
$2 3 4 6 8$, $2 3 4 8 9 10$, $2 3 4 8 9$, $2 3 4 8 10$, $2 3 4 8$, $2 3 4$, 
$2 3 5 6 7 8 9 10$, $2 3 5 6 7 8 9$, $2 3 5 6 7 8 10$, $2 3 5 6 7 8$, 
$2 3 5 6 7 9 1$, $2 3 5 6 7 9$, $2 3 5 6 7$, $2 3 5 6 8 9 10$, $2 3 5 6 8 9$, 
$2 3 5 6 8 10$, $2 3 5 6 8$, $2 3 5 7 8 9 10$, $2 3 5 7 8 9$, $2 3 5 7 8 10$, 
$2 3 5 7 8$, $2 3 5 7 9 10$, $2 3 5 7 9$, $2 3 5 7$, $2 3 7 8 9 10$, 
$2 3 7 8 9$, $2 3 7 8 10$, $2 3 7 8$, $2 3 7 9 10$, $2 3 7 9$, $2 3 7$, 
$2 4 5 6 7 8 9 10$, $2 4 5 6 7 8 9$, $2 4 5 6 7 8 10$, $2 4 5 6 7 8$, 
$2 4 5 6 7 9 10$, $2 4 5 6 7 9$, $2 4 5 6 7$, $2 4 5 6 8 9 10$, $2 4 5 6 8 9$, 
$2 4 5 6 8 10$, $2 4 5 6 8$, $2 4 5 7 8 9 10$, $2 4 5 7 8 9$, $2 4 5 7 8 10$, 
$2 4 5 7 8$, $2 4 5 7 9 10$, $2 4 5 7 9$, $2 4 5 7$, $2 4 6 7 8 9 10$, 
$2 4 6 7 8 9$, $2 4 6 7 8 10$, $2 4 6 7 8$, $2 4 6 7 9 10$, $2 4 6 7 9$, 
$2 4 6 7$, $2 4 6 8 9 10$, $2 4 6 8 9$, $2 4 6 8 10$, $2 4 6 8$, $2 4 8 9 10$, 
$2 4 8 9$, $2 4 8 10$, $2 4 8$, $2 4$, $3 4 5 6 7 8 9 10$, $3 4 5 6 7 8 9$, 
$3 4 5 6 7 8 10$, $3 4 5 6 7 8$, $3 4 5 6 7 9 10$, $3 4 5 6 7 9$, $3 4 5 6 7$, 
$3 4 5 6 8 9 10$, $3 4 5 6 8 9$, $3 4 5 6 8 10$, $3 4 5 6 8$, $3 4 5 7 8 9 10$, 
$3 4 5 7 8 9$, $3 4 5 7 8 10$, $3 4 5 7 8$, $3 4 5 7 9 10$, $3 4 5 7 9$, 
$3 4 5 7$, $3 4 6 7 8 9 10$, $3 4 6 7 8 9$, $3 4 6 7 8 10$, $3 4 6 7 8$, 
$3 4 6 7 9 10$, $3 4 6 7 9$, $3 4 6 7$, $3 4 6 8 9 10$, $3 4 6 8 9$, 
$3 4 6 8 10$, $3 4 6 8$, $3 4 8 9 10$, $3 4 8 9$, $3 4 8 10$, $3 4 8$, 
$3 4$, $3 5 6 7 8 9 10$, $3 5 6 7 8 9$, $3 5 6 7 8 10$, $3 5 6 7 8$, 
$3 5 6 7 9 10$, $3 5 6 7 9$, $3 5 6 7$, $3 5 6 8 9 10$, $3 5 6 8 9$, 
$3 5 6 8 10$, $3 5 6 8$, $3 5 7 8 9 10$, $3 5 7 8 9$, $3 5 7 8 10$, $3 5 7 8$, 
$3 5 7 9 10$, $3 5 7 9$, $3 5 7$, $3 7 8 9 10$, $3 7 8 9$, $3 7 8 10$, $3 7 8$, 
$3 7 9 10$, $3 7 9$, $3 7$, $4 5 6 7 8 9 10$, $4 5 6 7 8 9$, $4 5 6 7 8 10$, 
$4 5 6 7 8$, $4 5 6 7 9 10$, $4 5 6 7 9$, $4 5 6 7$, $4 5 6 8 9 10$, 
$4 5 6 8 9$, $4 5 6 8 10$, $4 5 6 8$, $4 5 7 8 9 10$, $4 5 7 8 9$, 
$4 5 7 8 10$, $4 5 7 8$, $4 5 7 9 10$, $4 5 7 9$, $4 5 7$, $4 6 7 8 9 10$, 
$4 6 7 8 9$, $4 6 7 8 10$, $4 6 7 8$, $4 6 7 9 10$, $4 6 7 9$, $4 6 7$, 
$4 6 8 9 10$, $4 6 8 9$, $4 6 8 10$, $4 6 8$, $4 8 9 10$, $4 8 9$, $4 8 10$, 
$4 8$, $4$, $7 8 9 10$, $7 8 9$, $7 8 10$, $7 8$, $7 9 10$, $7 9$, $7$.
\end{sloppypar}

\paragraph{The $591$ incidence patterns of the extremal group
$U$ of $11$ vertices depicted in Fig.~\ref{fig:f87}:}
\begin{sloppypar}
$\emptyset$, $1 2 3 4 5 6 7 8 9 10 11 $, $1 2 3 4 5 6 7 8 9 10 $, 
$1 2 3 4 5 6 7 8 9 11 $, $1 2 3 4 5 6 7 8 9 $, $1 2 3 4 5 6 7 8 10 11 $, 
$1 2 3 4 5 6 7 8 10 $, $1 2 3 4 5 6 7 8 $, $1 2 3 4 5 6 7 9 10 11 $, 
$1 2 3 4 5 6 7 9 10 $, $1 2 3 4 5 6 7 9 11 $, $1 2 3 4 5 6 7 9 $, 
$1 2 3 4 5 6 8 9 10 11 $, $1 2 3 4 5 6 8 9 10 $, $1 2 3 4 5 6 8 9 11 $, 
$1 2 3 4 5 6 8 9 $, $1 2 3 4 5 6 8 10 11 $, $1 2 3 4 5 6 8 10 $, 
$1 2 3 4 5 6 8 $, $1 2 3 4 5 7 8 9 10 11 $, $1 2 3 4 5 7 8 9 10 $, 
$1 2 3 4 5 7 8 9 11 $, $1 2 3 4 5 7 8 9 $, $1 2 3 4 5 7 8 10 11 $, 
$1 2 3 4 5 7 8 10 $, $1 2 3 4 5 7 8 $, $1 2 3 4 5 7 9 10 11 $, 
$1 2 3 4 5 7 9 10 $, $1 2 3 4 5 7 9 11 $, $1 2 3 4 5 7 9 $, 
$1 2 3 4 5 9 10 11 $, $1 2 3 4 5 9 10 $, $1 2 3 4 5 9 11 $, $1 2 3 4 5 9 $, 
$1 2 3 4 5 $, $1 2 3 4 6 7 8 9 10 11 $, $1 2 3 4 6 7 8 9 10 $, 
$1 2 3 4 6 7 8 9 11 $, $1 2 3 4 6 7 8 9 $, $1 2 3 4 6 7 8 10 11 $, 
$1 2 3 4 6 7 8 10 $, $1 2 3 4 6 7 8 $, $1 2 3 4 6 7 9 10 11 $, 
$1 2 3 4 6 7 9 10 $, $1 2 3 4 6 7 9 11 $, $1 2 3 4 6 7 9 $, 
$1 2 3 4 6 8 9 10 11 $, $1 2 3 4 6 8 9 10 $, $1 2 3 4 6 8 9 11 $, 
$1 2 3 4 6 8 9 $, $1 2 3 4 6 8 10 11 $, $1 2 3 4 6 8 10 $, $1 2 3 4 6 8 $, 
$1 2 3 4 8 9 10 11 $, $1 2 3 4 8 9 10 $, $1 2 3 4 8 9 11 $, $1 2 3 4 8 9 $, 
$1 2 3 4 8 10 11 $, $1 2 3 4 8 10 $, $1 2 3 4 8 $, $1 2 3 4 $, 
$1 2 3 5 6 7 8 9 10 11 $, $1 2 3 5 6 7 8 9 10 $, $1 2 3 5 6 7 8 9 11 $, 
$1 2 3 5 6 7 8 9 $, $1 2 3 5 6 7 8 10 11 $, $1 2 3 5 6 7 8 10 $, 
$1 2 3 5 6 7 8 $, $1 2 3 5 6 7 9 10 11 $, $1 2 3 5 6 7 9 10 $, 
$1 2 3 5 6 7 9 11 $, $1 2 3 5 6 7 9 $, $1 2 3 5 6 8 9 10 11 $, 
$1 2 3 5 6 8 9 10 $, $1 2 3 5 6 8 9 11 $, $1 2 3 5 6 8 9 $, 
$1 2 3 5 6 8 10 11 $, $1 2 3 5 6 8 10 $, $1 2 3 5 6 8 $, 
$1 2 3 5 7 8 9 10 11 $, $1 2 3 5 7 8 9 10 $, $1 2 3 5 7 8 9 11 $, 
$1 2 3 5 7 8 9 $, $1 2 3 5 7 8 10 11 $, $1 2 3 5 7 8 10 $, $1 2 3 5 7 8 $, 
$1 2 3 5 7 9 10 11 $, $1 2 3 5 7 9 10 $, $1 2 3 5 7 9 11 $, $1 2 3 5 7 9 $, 
$1 2 3 5 9 10 11 $, $1 2 3 5 9 10 $, $1 2 3 5 9 11 $, $1 2 3 5 9 $, 
$1 2 3 5 $, $1 2 4 5 6 7 8 9 10 11 $, $1 2 4 5 6 7 8 9 10 $, 
$1 2 4 5 6 7 8 9 11 $, $1 2 4 5 6 7 8 9 $, $1 2 4 5 6 7 8 10 11 $, 
$1 2 4 5 6 7 8 10 $, $1 2 4 5 6 7 8 $, $1 2 4 5 6 7 9 10 11 $, 
$1 2 4 5 6 7 9 10 $, $1 2 4 5 6 7 9 11 $, $1 2 4 5 6 7 9 $, 
$1 2 4 5 6 8 9 10 11 $, $1 2 4 5 6 8 9 10 $, $1 2 4 5 6 8 9 11 $, 
$1 2 4 5 6 8 9 $, $1 2 4 5 6 8 10 11 $, $1 2 4 5 6 8 10 $, $1 2 4 5 6 8 $, 
$1 2 4 5 7 8 9 10 11 $, $1 2 4 5 7 8 9 10 $, $1 2 4 5 7 8 9 11 $, 
$1 2 4 5 7 8 9 $, $1 2 4 5 7 8 10 11 $, $1 2 4 5 7 8 10 $, $1 2 4 5 7 8 $, 
$1 2 4 5 7 9 10 11 $, $1 2 4 5 7 9 10 $, $1 2 4 5 7 9 11 $, $1 2 4 5 7 9 $, 
$1 2 4 5 9 10 11 $, $1 2 4 5 9 10 $, $1 2 4 5 9 11 $, $1 2 4 5 9 $, 
$1 2 4 5 $, $1 2 4 6 7 8 9 10 11 $, $1 2 4 6 7 8 9 10 $, $1 2 4 6 7 8 9 11 $, 
$1 2 4 6 7 8 9 $, $1 2 4 6 7 8 10 11 $, $1 2 4 6 7 8 10 $, $1 2 4 6 7 8 $, 
$1 2 4 6 7 9 10 11 $, $1 2 4 6 7 9 10 $, $1 2 4 6 7 9 11 $, $1 2 4 6 7 9 $, 
$1 2 4 6 8 9 10 11 $, $1 2 4 6 8 9 10 $, $1 2 4 6 8 9 11 $, $1 2 4 6 8 9 $, 
$1 2 4 6 8 10 11 $, $1 2 4 6 8 10 $, $1 2 4 6 8 $, $1 2 4 8 9 10 11 $, 
$1 2 4 8 9 10 $, $1 2 4 8 9 11 $, $1 2 4 8 9 $, $1 2 4 8 10 11 $, 
$1 2 4 8 10 $, $1 2 4 8 $, $1 2 4 $, $1 3 4 5 6 7 8 9 10 11 $, 
$1 3 4 5 6 7 8 9 10 $, $1 3 4 5 6 7 8 9 11 $, $1 3 4 5 6 7 8 9 $, 
$1 3 4 5 6 7 8 10 11 $, $1 3 4 5 6 7 8 10 $, $1 3 4 5 6 7 8 $, 
$1 3 4 5 6 7 9 10 11 $, $1 3 4 5 6 7 9 10 $, $1 3 4 5 6 7 9 11 $, 
$1 3 4 5 6 7 9 $, $1 3 4 5 6 8 9 10 11 $, $1 3 4 5 6 8 9 10 $, 
$1 3 4 5 6 8 9 11 $, $1 3 4 5 6 8 9 $, $1 3 4 5 6 8 10 11 $, 
$1 3 4 5 6 8 10 $, $1 3 4 5 6 8 $, $1 3 4 5 7 8 9 10 11 $, 
$1 3 4 5 7 8 9 10 $, $1 3 4 5 7 8 9 11 $, $1 3 4 5 7 8 9 $, 
$1 3 4 5 7 8 10 11 $, $1 3 4 5 7 8 10 $, $1 3 4 5 7 8 $, 
$1 3 4 5 7 9 10 11 $, $1 3 4 5 7 9 10 $, $1 3 4 5 7 9 11 $, 
$1 3 4 5 7 9 $, $1 3 4 5 9 10 11 $, $1 3 4 5 9 10 $, $1 3 4 5 9 11 $, 
$1 3 4 5 9 $, $1 3 4 5 $, $1 3 4 6 7 8 9 10 11 $, $1 3 4 6 7 8 9 10 $, 
$1 3 4 6 7 8 9 11 $, $1 3 4 6 7 8 9 $, $1 3 4 6 7 8 10 11 $, 
$1 3 4 6 7 8 10 $, $1 3 4 6 7 8 $, $1 3 4 6 7 9 10 11 $, $1 3 4 6 7 9 10 $, 
$1 3 4 6 7 9 11 $, $1 3 4 6 7 9 $, $1 3 4 6 8 9 10 11 $, $1 3 4 6 8 9 10 $, 
$1 3 4 6 8 9 11 $, $1 3 4 6 8 9 $, $1 3 4 6 8 10 11 $, $1 3 4 6 8 10 $, 
$1 3 4 6 8 $, $1 3 4 8 9 10 11 $, $1 3 4 8 9 10 $, $1 3 4 8 9 11 $, 
$1 3 4 8 9 $, $1 3 4 8 10 11 $, $1 3 4 8 10 $, $1 3 4 8 $, $1 3 4 $, 
$1 3 5 6 7 8 9 10 11 $, $1 3 5 6 7 8 9 10 $, $1 3 5 6 7 8 9 11 $, 
$1 3 5 6 7 8 9 $, $1 3 5 6 7 8 10 11 $, $1 3 5 6 7 8 10 $, $1 3 5 6 7 8 $, 
$1 3 5 6 7 9 10 11 $, $1 3 5 6 7 9 10 $, $1 3 5 6 7 9 11 $, $1 3 5 6 7 9 $, 
$1 3 5 6 8 9 10 11 $, $1 3 5 6 8 9 10 $, $1 3 5 6 8 9 11 $, $1 3 5 6 8 9 $, 
$1 3 5 6 8 10 11 $, $1 3 5 6 8 10 $, $1 3 5 6 8 $, $1 3 5 7 8 9 10 11 $, 
$1 3 5 7 8 9 10 $, $1 3 5 7 8 9 11 $, $1 3 5 7 8 9 $, $1 3 5 7 8 10 11 $, 
$1 3 5 7 8 10 $, $1 3 5 7 8 $, $1 3 5 7 9 10 11 $, $1 3 5 7 9 10 $, 
$1 3 5 7 9 11 $, $1 3 5 7 9 $, $1 3 5 9 10 11 $, $1 3 5 9 10 $, 
$1 3 5 9 11 $, $1 3 5 9 $, $1 3 5 $, $2 3 4 5 6 7 8 9 10 11 $, 
$2 3 4 5 6 7 8 9 10 $, $2 3 4 5 6 7 8 9 11 $, $2 3 4 5 6 7 8 9 $, 
$2 3 4 5 6 7 8 10 11 $, $2 3 4 5 6 7 8 10 $, $2 3 4 5 6 7 8 $, 
$2 3 4 5 6 7 9 10 11 $, $2 3 4 5 6 7 9 10 $, $2 3 4 5 6 7 9 11 $, 
$2 3 4 5 6 7 9 $, $2 3 4 5 6 8 9 10 11 $, $2 3 4 5 6 8 9 10 $, 
$2 3 4 5 6 8 9 11 $, $2 3 4 5 6 8 9 $, $2 3 4 5 6 8 10 11 $, 
$2 3 4 5 6 8 10 $, $2 3 4 5 6 8 $, $2 3 4 5 7 8 9 10 11 $, 
$2 3 4 5 7 8 9 10 $, $2 3 4 5 7 8 9 11 $, $2 3 4 5 7 8 9 $, 
$2 3 4 5 7 8 10 11 $, $2 3 4 5 7 8 10 $, $2 3 4 5 7 8 $, 
$2 3 4 5 7 9 10 11 $, $2 3 4 5 7 9 10 $, $2 3 4 5 7 9 11 $, 
$2 3 4 5 7 9 $, $2 3 4 5 9 10 11 $, $2 3 4 5 9 10 $, $2 3 4 5 9 11 $, 
$2 3 4 5 9 $, $2 3 4 5 $, $2 3 4 6 7 8 9 10 11 $, $2 3 4 6 7 8 9 10 $, 
$2 3 4 6 7 8 9 11 $, $2 3 4 6 7 8 9 $, $2 3 4 6 7 8 10 11 $, 
$2 3 4 6 7 8 10 $, $2 3 4 6 7 8 $, $2 3 4 6 7 9 10 11 $, $2 3 4 6 7 9 10 $, 
$2 3 4 6 7 9 11 $, $2 3 4 6 7 9 $, $2 3 4 6 8 9 10 11 $, $2 3 4 6 8 9 10 $, 
$2 3 4 6 8 9 11 $, $2 3 4 6 8 9 $, $2 3 4 6 8 10 11 $, $2 3 4 6 8 10 $, 
$2 3 4 6 8 $, $2 3 4 8 9 10 11 $, $2 3 4 8 9 10 $, $2 3 4 8 9 11 $, 
$2 3 4 8 9 $, $2 3 4 8 10 11 $, $2 3 4 8 10 $, $2 3 4 8 $, $2 3 4 $, 
$2 3 5 6 7 8 9 10 11 $, $2 3 5 6 7 8 9 10 $, $2 3 5 6 7 8 9 11 $, 
$2 3 5 6 7 8 9 $, $2 3 5 6 7 8 10 11 $, $2 3 5 6 7 8 10 $, $2 3 5 6 7 8 $, 
$2 3 5 6 7 9 10 11 $, $2 3 5 6 7 9 10 $, $2 3 5 6 7 9 11 $, $2 3 5 6 7 9 $, 
$2 3 5 6 8 9 10 11 $, $2 3 5 6 8 9 10 $, $2 3 5 6 8 9 11 $, $2 3 5 6 8 9 $, 
$2 3 5 6 8 10 11 $, $2 3 5 6 8 10 $, $2 3 5 6 8 $, $2 3 5 7 8 9 10 11 $, 
$2 3 5 7 8 9 10 $, $2 3 5 7 8 9 11 $, $2 3 5 7 8 9 $, $2 3 5 7 8 10 11 $, 
$2 3 5 7 8 10 $, $2 3 5 7 8 $, $2 3 5 7 9 10 11 $, $2 3 5 7 9 10 $, 
$2 3 5 7 9 11 $, $2 3 5 7 9 $, $2 3 5 9 10 11 $, $2 3 5 9 10 $, 
$2 3 5 9 11 $, $2 3 5 9 $, $2 3 5 $, $2 4 5 6 7 8 9 10 11 $, 
$2 4 5 6 7 8 9 10 $, $2 4 5 6 7 8 9 11 $, $2 4 5 6 7 8 9 $, 
$2 4 5 6 7 8 10 11 $, $2 4 5 6 7 8 10 $, $2 4 5 6 7 8 $, 
$2 4 5 6 7 9 10 11 $, $2 4 5 6 7 9 10 $, $2 4 5 6 7 9 11 $, 
$2 4 5 6 7 9 $, $2 4 5 6 8 9 10 11 $, $2 4 5 6 8 9 10 $, $2 4 5 6 8 9 11 $, 
$2 4 5 6 8 9 $, $2 4 5 6 8 10 11 $, $2 4 5 6 8 10 $, $2 4 5 6 8 $, 
$2 4 5 7 8 9 10 11 $, $2 4 5 7 8 9 10 $, $2 4 5 7 8 9 11 $, 
$2 4 5 7 8 9 $, $2 4 5 7 8 10 11 $, $2 4 5 7 8 10 $, $2 4 5 7 8 $, 
$2 4 5 7 9 10 11 $, $2 4 5 7 9 10 $, $2 4 5 7 9 11 $, $2 4 5 7 9 $, 
$2 4 5 9 10 11 $, $2 4 5 9 10 $, $2 4 5 9 11 $, $2 4 5 9 $, $2 4 5 $, 
$2 4 6 7 8 9 10 11 $, $2 4 6 7 8 9 10 $, $2 4 6 7 8 9 11 $, $2 4 6 7 8 9 $, 
$2 4 6 7 8 10 11 $, $2 4 6 7 8 10 $, $2 4 6 7 8 $, $2 4 6 7 9 10 11 $, 
$2 4 6 7 9 10 $, $2 4 6 7 9 11 $, $2 4 6 7 9 $, $2 4 6 8 9 10 11 $, 
$2 4 6 8 9 10 $, $2 4 6 8 9 11 $, $2 4 6 8 9 $, $2 4 6 8 10 11 $, 
$2 4 6 8 10 $, $2 4 6 8 $, $2 4 8 9 10 11 $, $2 4 8 9 10 $, $2 4 8 9 11 $, 
$2 4 8 9 $, $2 4 8 10 11 $, $2 4 8 10 $, $2 4 8 $, $2 4 $, 
$3 4 5 6 7 8 9 10 11 $, $3 4 5 6 7 8 9 10 $, $3 4 5 6 7 8 9 11 $, 
$3 4 5 6 7 8 9 $, $3 4 5 6 7 8 10 11 $, $3 4 5 6 7 8 10 $, $3 4 5 6 7 8 $, 
$3 4 5 6 7 9 10 11 $, $3 4 5 6 7 9 10 $, $3 4 5 6 7 9 11 $, $3 4 5 6 7 9 $, 
$3 4 5 6 8 9 10 11 $, $3 4 5 6 8 9 10 $, $3 4 5 6 8 9 11 $, $3 4 5 6 8 9 $, 
$3 4 5 6 8 10 11 $, $3 4 5 6 8 10 $, $3 4 5 6 8 $, $3 4 5 7 8 9 10 11 $, 
$3 4 5 7 8 9 10 $, $3 4 5 7 8 9 11 $, $3 4 5 7 8 9 $, $3 4 5 7 8 10 11 $, 
$3 4 5 7 8 10 $, $3 4 5 7 8 $, $3 4 5 7 9 10 11 $, $3 4 5 7 9 10 $, 
$3 4 5 7 9 11 $, $3 4 5 7 9 $, $3 4 5 9 10 11 $, $3 4 5 9 10 $, 
$3 4 5 9 11 $, $3 4 5 9 $, $3 4 5 $, $3 4 6 7 8 9 10 11 $, $3 4 6 7 8 9 10 $, 
$3 4 6 7 8 9 11 $, $3 4 6 7 8 9 $, $3 4 6 7 8 10 11 $, $3 4 6 7 8 10 $, 
$3 4 6 7 8 $, $3 4 6 7 9 10 11 $, $3 4 6 7 9 10 $, $3 4 6 7 9 11 $, 
$3 4 6 7 9 $, $3 4 6 8 9 10 11 $, $3 4 6 8 9 10 $, $3 4 6 8 9 11 $, 
$3 4 6 8 9 $, $3 4 6 8 10 11 $, $3 4 6 8 10 $, $3 4 6 8 $, $3 4 8 9 10 11 $, 
$3 4 8 9 10 $, $3 4 8 9 11 $, $3 4 8 9 $, $3 4 8 10 11 $, $3 4 8 10 $, 
$3 4 8 $, $3 4 $, $3 5 6 7 8 9 10 11 $, $3 5 6 7 8 9 10 $, $3 5 6 7 8 9 11 $, 
$3 5 6 7 8 9 $, $3 5 6 7 8 10 11 $, $3 5 6 7 8 10 $, $3 5 6 7 8 $, 
$3 5 6 7 9 10 11 $, $3 5 6 7 9 10 $, $3 5 6 7 9 11 $, $3 5 6 7 9 $, 
$3 5 6 8 9 10 11 $, $3 5 6 8 9 10 $, $3 5 6 8 9 11 $, $3 5 6 8 9 $, 
$3 5 6 8 10 11 $, $3 5 6 8 10 $, $3 5 6 8 $, $3 5 7 8 9 10 11 $, 
$3 5 7 8 9 10 $, $3 5 7 8 9 11 $, $3 5 7 8 9 $, $3 5 7 8 10 11 $, 
$3 5 7 8 10 $, $3 5 7 8 $, $3 5 7 9 10 11 $, $3 5 7 9 10 $, $3 5 7 9 11 $, 
$3 5 7 9 $, $3 5 9 10 11 $, $3 5 9 10 $, $3 5 9 11 $, $3 5 9 $, $3 5 $, 
$4 5 6 7 8 9 10 11 $, $4 5 6 7 8 9 10 $, $4 5 6 7 8 9 11 $, $4 5 6 7 8 9 $, 
$4 5 6 7 8 10 11 $, $4 5 6 7 8 10 $, $4 5 6 7 8 $, $4 5 6 7 9 10 11 $, 
$4 5 6 7 9 10 $, $4 5 6 7 9 11 $, $4 5 6 7 9 $, $4 5 6 8 9 10 11 $, 
$4 5 6 8 9 10 $, $4 5 6 8 9 11 $, $4 5 6 8 9 $, $4 5 6 8 10 11 $, 
$4 5 6 8 10 $, $4 5 6 8 $, $4 5 7 8 9 10 11 $, $4 5 7 8 9 10 $, 
$4 5 7 8 9 11 $, $4 5 7 8 9 $, $4 5 7 8 10 11 $, $4 5 7 8 10 $, 
$4 5 7 8 $, $4 5 7 9 10 11 $, $4 5 7 9 10 $, $4 5 7 9 11 $, $4 5 7 9 $, 
$4 5 9 10 11 $, $4 5 9 10 $, $4 5 9 11 $, $4 5 9 $, $4 5 $, 
$4 6 7 8 9 10 11 $, $4 6 7 8 9 10 $, $4 6 7 8 9 11 $, $4 6 7 8 9 $, 
$4 6 7 8 10 11 $, $4 6 7 8 10 $, $4 6 7 8 $, $4 6 7 9 10 11 $, $4 6 7 9 10 $, 
$4 6 7 9 11 $, $4 6 7 9 $, $4 6 8 9 10 11 $, $4 6 8 9 10 $, $4 6 8 9 11 $, 
$4 6 8 9 $, $4 6 8 10 11 $, $4 6 8 10 $, $4 6 8 $, $4 8 9 10 11 $, 
$4 8 9 10 $, $4 8 9 11 $, $4 8 9 $, $4 8 10 11 $, $4 8 10 $, $4 8 $, $4 $, 
$5 6 7 8 9 10 11 $, $5 6 7 8 9 10 $, $5 6 7 8 9 11 $, $5 6 7 8 9 $, 
$5 6 7 8 10 11 $, $5 6 7 8 10 $, $5 6 7 8 $, $5 6 7 9 10 11 $, $5 6 7 9 10 $, 
$5 6 7 9 11 $, $5 6 7 9 $, $5 6 8 9 10 11 $, $5 6 8 9 10 $, $5 6 8 9 11 $, 
$5 6 8 9 $, $5 6 8 10 11 $, $5 6 8 10 $, $5 6 8 $, $5 7 8 9 10 11 $, 
$5 7 8 9 10 $, $5 7 8 9 11 $, $5 7 8 9 $, $5 7 8 10 11 $, $5 7 8 10 $, 
$5 7 8 $, $5 7 9 10 11 $, $5 7 9 10 $, $5 7 9 11 $, $5 7 9 $, 
$5 9 10 11 $, $5 9 10 $, $5 9 11 $, $5 9 $, $5 $.
\end{sloppypar}

\section{Source code} \label{sec:source}

The C code was compiled with \texttt{gcc 5.4.0} in Linux 4.4.0 on an Intel i7 processor.
The program compiles with the command:
\begin{center}
\begin{tcolorbox}[width=10cm,height=10mm]
	\begin{center}
		\texttt{ { gcc path.c -lm}} \end{center}
\end{tcolorbox}
\end{center}

\medskip
The output of the program is the following.
\begin{center}
\begin{tcolorbox}[width=15cm,height=35mm]
\texttt{Number of vertices is 11. \\
Number of sides of size 11 is 1073523. \\
Number of groups of size 11 is 655273284083. \\
Maximum number of patterns in a group of size 11 is 591. \\
Number of groups of size 11 with maximum number of patterns is 2. \\
The program takes 1097678 seconds to execute. \\
}
\end{tcolorbox}
\end{center}

\begin{lstlisting}
/**
  * Input: n (type: Integer) number of vertices in the group
  * Output: Maximum number of patterns in groups of size n
  * This program
  * 1. generates all sides of sizes from 0 to n-1 using divide and conquer and stores them in files
  * 2. combine all these pairs of sides to generate all the groups of size n
  * 3. computes number of patterns each group has and output the maximum number of patterns
  * @authr: Ritankar Mandal
  * @date: Sep 30, 2016
  */


#define MAXNSQUARE 400
#define MAXN 20

#include<stdio.h>
#include<math.h>
#include<stdlib.h>
#include<string.h>
#include<time.h>


/// Variables
long long int numSidesOfSize[MAXN];

/// Structures
/// contains a graph in an adjacency matrix
typedef struct _group
{
    long long int index;
    int adjMat[MAXN][MAXN];
    int numV;
} group;

group groupWithMaxNumPaths[500];
int groupWithMaxNumPathsCounter;
int maxNumPathsOfGroupsOfSizen;


/// Functions
void introduction();
long long int generateSidesOfSize(int, int);
long long int generateSidesOfSizeBruteForce(int);
long long int nextPermutaion(long long int);
int checkPlanarity(group);
long long int generateSidesOfSizeDandC(int, int);
long long int combineSidesOfSize(long long int, FILE*, int, int, int);
long long int generateGroupsOfSize(int);
int checkCompatibility(group, group);
group createTempGroup(group, group);
int computeNumPaths(group);
group reverseSideWithYAxis(group);
void drawMaxGroups(int);
void drawGraph(FILE *fp, group);
void printGraph(FILE *fp, group);
void writePaths(FILE *fp, group);
long long int binCoeff(int, int);
int min(int, int);


int main()
{
    char timebuffStart[100], timebuffFinish[100];
    double timeDiff;
    /// Log the starting time of the experiment
    time_t start = time (0);
    strftime (timebuffStart, 100, "%H:%M:%S on %m-%d-%Y", localtime (&start));
    printf("Experiment started at %s. \n", timebuffStart);

    int n, i;
    long long int numGroupsOfSize;

    printf("Enter the number of vertices in a group: ");
    scanf("%d", &n);

    introduction();
    for(i=2 ; i<=n-1 ;i++)
        numSidesOfSize[i] = generateSidesOfSize(i, 0);
    numSidesOfSize[n] = generateSidesOfSize(n, 1);
    printf("Number of sides of size %d is %lld.\n", n, numSidesOfSize[n]);

    numGroupsOfSize = generateGroupsOfSize(n);
    printf("Number of groups of size %d is %lld.\n", n, numGroupsOfSize);
    printf("Maximum number of patterns in a group of size %d is %d.\n", n, maxNumPathsOfGroupsOfSizen);
    printf("Number of groups of size %d with maximum number of patterns is %d.\n", n, groupWithMaxNumPathsCounter);


    /// Log the finishing time of the experiment
    time_t finish = time (0);
    strftime (timebuffFinish, 100, "%H:%M:%S on %m-%d-%Y", localtime (&finish));
    printf("Experiment started at %s. \n", timebuffStart);
    printf("Experiment finished at %s. \n", timebuffFinish);
    timeDiff = difftime(finish, start);
    printf("The experiment takes %f seconds to execute.\n", timeDiff);

    drawMaxGroups(n);

    return 0;
}


void introduction()
{
    int i, j;
    group tempSide;
    FILE *fpSides;

    if((fpSides = fopen("./OutputFiles/sides_0.txt", "w")) == NULL)
    {
        printf("Can't open the file ./OutputFiles/sides_0.txt\n");
        exit(1);
    }
    tempSide.index = 0;
    tempSide.numV = 0;
    for(i = 0 ; i <= tempSide.numV+1 ; i++)
        for(j = 0 ; j <= tempSide.numV+1 ; j++)
            tempSide.adjMat[i][j] = 0;
    fwrite(&tempSide, sizeof(tempSide), 1, fpSides);
    fclose(fpSides);
    numSidesOfSize[0] = 1;

    if((fpSides = fopen("./OutputFiles/sides_1.txt", "w")) == NULL)
    {
        printf("Can't open the file ./OutputFiles/sides_1.txt\n");
        exit(1);
    }
    tempSide.index = 0;
    tempSide.numV = 1;
    for(i = 0 ; i <= tempSide.numV+1 ; i++)
        for(j = 0 ; j <= tempSide.numV+1 ; j++)
            tempSide.adjMat[i][j] = 0;
    fwrite(&tempSide, sizeof(tempSide), 1, fpSides);
    fclose(fpSides);
    numSidesOfSize[1] = 1;
}


long long int generateSidesOfSize(int sideSize, int last)
{
    long long int numPlanarSides = 0;
    if(sideSize <= 8)
        numPlanarSides = generateSidesOfSizeBruteForce(sideSize);
    else
        numPlanarSides = generateSidesOfSizeDandC(sideSize, last);

    return numPlanarSides;
}


/**
  * Input: sideSize (type: int).
  * Output: Number of sides of size sideSize (type: long long int).
  * This function generates all sides of sideSize using brute force.
  */
long long int generateSidesOfSizeBruteForce(int sideSize)
{
    printf("sideSize is %d.\n", sideSize);

    int j, k, l, m, numChords, numPossibleEdges = 0;
    int possibleEdgesSource[MAXNSQUARE], possibleEdgesDestination[MAXNSQUARE];
    group tempSide;
    int sideStructSize = sizeof(tempSide);
    long long int i, v, numSides, numPlanarSides = 0;

    for(l = 0 ; l <= sideSize-1 ; l++)
        for(m = l+2 ; m <= sideSize+1 ; m++)
            if( ( ( (m-l) >= 2 ) || (l==0) || ( m==(sideSize+1) ) ) && !((l==0) && ( m==(sideSize+1) )) )
            {
                possibleEdgesSource[numPossibleEdges] = l;
                possibleEdgesDestination[numPossibleEdges] = m;
                numPossibleEdges++;
            }

    /// numPossibleEdges = (sideSize^2 + sideSize - 2)/2
    printf("For sideSize %d, numPossibleEdges is %d.\n", sideSize, numPossibleEdges);

    /// Opening the file "./OutputFiles/sides_<sideSize>.txt" to write the sides
    char sideFileName[100];
    FILE *fpSides;
    sprintf(sideFileName, "./OutputFiles/sides_%g.txt", (double) sideSize);
    if((fpSides = fopen(sideFileName, "w")) == NULL)
    {
        printf("Can't open the file %s\n", sideFileName);
        exit(1);
    }

    /// Generating all planar sides

    /// Generate the empty side
    tempSide.index = numPlanarSides;
    tempSide.numV = sideSize;
    for(l = 0 ; l <= sideSize+1 ; l++)
        for(m = 0 ; m <= sideSize+1 ; m++)
            tempSide.adjMat[l][m] = 0;
    fwrite(&tempSide, sideStructSize, 1, fpSides);

    if(numPlanarSides%100 == 0)
        printf("%lld th planar side of size %d has been generated.\n", numPlanarSides, sideSize);
    numPlanarSides++;

    /// Generate the non-empty sides
    for(numChords = 1; numChords <= sideSize-1; numChords++)
    {
        numSides = binCoeff(numPossibleEdges, numChords);
        v = pow(2, numChords) - 1;

        /// Generating all planar sides with numChords chords
        for(i = 0 ; i < numSides ; i++)
        {
            tempSide.numV = sideSize;
            for(l = 0 ; l <= sideSize+1 ; l++)
                for(m = 0 ; m <= sideSize+1 ; m++)
                    tempSide.adjMat[l][m] = 0;

            for(j = 0 ; j < numPossibleEdges ; j++)
            {
                k = v >> j;
                tempSide.adjMat[possibleEdgesSource[j]][possibleEdgesDestination[j]] = (k & 1);
            }

            /// Checks planarity of tempSide
            if(checkPlanarity(tempSide) == 1)
            {
                tempSide.index = numPlanarSides;
                fwrite(&tempSide, sideStructSize, 1, fpSides);
                if(numPlanarSides%500 == 0)
                    printf("%lld th planar side of size %d has been generated.\n", numPlanarSides, sideSize);
                numPlanarSides++;
            }

            /// Generate the next subset (bit permutation)
            v = nextPermutaion(v);
        }
    }
    printf("Total number of planar sides of size %d is %lld.\n", sideSize, numPlanarSides);
    fclose(fpSides);
    return numPlanarSides;
}


/**
  * Input 1: parameter1 (type: int).
  * Input 2: parameter2 (type: int).
  * Output: Binomial Coefficient C(parameter1, parameter2) (type: long long int)
  */
long long int binCoeff(int parameter1, int parameter2)
{
    long long int C[100][100];
    int i, j;

    /// Calculates the value of Binomial Coefficient in bottom up manner
    for (i = 0; i <= parameter1; i++)
    {
        for (j = 0; j <= min(i, parameter2); j++)
        {
            /// Base Cases
            if (j == 0 || j == i)
                C[i][j] = 1;

            /// Calculate value using previously stored values
            else
                C[i][j] = C[i-1][j-1] + C[i-1][j];
        }
    }

    return C[parameter1][parameter2];
}


/**
  * Input: v (type: long long int)
  * Output: next pattern of bits after v (type: long long int).
  * It has same number of 1's in binary as the input.
  * It is taken from the website
    http://graphics.stanford.edu/~seander/bithacks.html#NextBitPermutation
  */
long long int nextPermutaion(long long int v)
{
    //long long int v; // current permutation of bits
    long long int w; // next permutation of bits

    long long int t = (v | (v - 1)) + 1;
    w = t | ((((t & -t) / (v & -v)) >> 1) - 1);

    return w;
}


int min(int parameter1, int parameter2)
{
    return (parameter1 < parameter2) ? parameter1: parameter2;
}


/**
  * Input: tempGraph (type: group).
  * Output: if the input is planar then returns 1, otherwise returns 0.
  */
int checkPlanarity(group tempSide)
{
    int j, k, l, m, status = 1;

    for(j = 0 ; j <= tempSide.numV ; j++)
        for(k = j+1 ; k <= tempSide.numV+1 ; k++)
            for(l = 0 ; l <= tempSide.numV ; l++)
                for(m = l+1 ; m <= tempSide.numV+1 ; m++)
                    if( (tempSide.adjMat[j][k] == 1) && (tempSide.adjMat[l][m] == 1))
                        if( ( (j<l) && (l<k) && (k<m) ) || ( (l<j) && (j<m) && (m<k) ) )
                                status = 0;

    return status;
}


/**
  * Input: sideSize (type: int).
  * Output: Number of sides of size sideSize (type: long long int).
  * This function uses all sides of sizes smaller than sideSize,
  * places them side by side and combines to generate all sides of size sideSize
  */
long long int generateSidesOfSizeDandC(int sideSize, int last)
{
    int i;
    long long int numPlanarSides = 0;

    /// Opening the file "./OutputFiles/sides_<sideSize>.txt" to write the sides
    char sideFileName[100];
    FILE *fpSides;
    sprintf(sideFileName, "./OutputFiles/sides_%g.txt", (double) sideSize);
    if((fpSides = fopen(sideFileName, "w")) == NULL)
    {
        printf("Can't open the file %s\n", sideFileName);
        exit(1);
    }

    /// Combine sides
    for(i=0; i<= sideSize/2; i++)
        numPlanarSides = combineSidesOfSize(numPlanarSides, fpSides, i, sideSize-i-1, last);

    fclose(fpSides);
    return numPlanarSides;
}


long long int combineSidesOfSize(long long int numPlanarSides, FILE* fpSides, int leftSideSize, int rightSideSize, int last)
{
    long long int i, j;
    long long int numLeftSides = numSidesOfSize[leftSideSize];
    long long int numRightSides = numSidesOfSize[rightSideSize];
    char leftSideFilename[200], rightSideFilename[100];
    FILE *fpLeftSides, *fpRightSides;
    group tempLeftSide, tempRightSide, tempSide;
    int l, m, cntr;
    int sideSize = leftSideSize + rightSideSize + 1;
    int sideStructSize = sizeof(tempSide);

    /// Opening the files "./OutputFiles/sides_<leftSideSize>.txt" to read the left sides
    sprintf(leftSideFilename, "./OutputFiles/sides_%g.txt", (double) leftSideSize);
    if((fpLeftSides = fopen(leftSideFilename, "r")) == NULL)
    {
        printf("Can't open the file %s\n", leftSideFilename);
        exit(1);
    }

    for(i = 0; i < numLeftSides ; i++)
    {
        /// Read the left side
        fread(&tempLeftSide, sideStructSize, 1, fpLeftSides);

        /// Opening the file "./OutputFiles/sides_<rightSideSize>.txt" to read the right sides
        sprintf(rightSideFilename, "./OutputFiles/sides_%g.txt", (double) rightSideSize);
        if((fpRightSides = fopen(rightSideFilename, "r")) == NULL)
        {
            printf("Can't open the file %s\n", rightSideFilename);
            exit(1);
        }

        for(j = 0; j < numRightSides ; j++)
        {
            /// Read the right side
            fread(&tempRightSide, sideStructSize, 1, fpRightSides);

            /// Combine left and right sides
            tempSide.numV = sideSize;
            /// reset the adjacency matrix
            for(l = 0 ; l <= sideSize+1 ; l++)
                for(m = 0 ; m <= sideSize+1 ; m++)
                    tempSide.adjMat[l][m] = 0;

            /// Copy the left adjacency matrix to the left
            for(l = 0 ; l <= leftSideSize+1 ; l++)
                for(m = 0 ; m <= leftSideSize+1 ; m++)
                    tempSide.adjMat[l][m] = tempLeftSide.adjMat[l][m];

            /// Copy the right adjacency matrix to the right
            for(l = 0 ; l <= rightSideSize+1 ; l++)
                for(m = 0 ; m <= rightSideSize+1 ; m++)
                    tempSide.adjMat[l + leftSideSize + 1][m + leftSideSize +1] = tempRightSide.adjMat[l][m];

            /// Add the incoming and outgoing edges incident to vertex at leftSideSize+1
            for(cntr = 3*last ; cntr < 4 ; cntr++)
            {
                if( (cntr == 0) || (cntr == 1) )
                    tempSide.adjMat[0][leftSideSize+1] = 1;
                else if( (cntr == 2) || (cntr == 3) )
                    tempSide.adjMat[0][leftSideSize+1] = 1;

                if( (cntr == 0) || (cntr == 2) )
                    tempSide.adjMat[leftSideSize+1][sideSize+1] = 1;
                else if( (cntr == 1) || (cntr == 3) )
                    tempSide.adjMat[leftSideSize+1][sideSize+1] = 1;

                /// Store the side
                tempSide.index = numPlanarSides;
                fwrite(&tempSide, sideStructSize, 1, fpSides);
                if(numPlanarSides%500 == 0)
                    printf("%lld th planar sides of size %d has been generated.\n", numPlanarSides, sideSize);
                numPlanarSides++;
            }//for(cntr) loop ends

        }//for(j) loop ends
        fclose(fpRightSides);
    }//for(i) loop ends
    fclose(fpLeftSides);
    return numPlanarSides;
}


long long int generateGroupsOfSize(int groupSize)
{
    long long int i, j, numGroups = 0;
    int numPaths, maxnumPaths = 0;

    group tempUpperSide, tempLowerSide, tempGroup;
    int sideStructSize = sizeof(tempUpperSide);
    long long int numPlanarSides = numSidesOfSize[groupSize];

    /// Opening the file "./OutputFiles/sides_<sideSize>.txt" to read the upper sides
    FILE *fpUpperSides, *fpLowerSides;
    char sideFileName[100];
    sprintf(sideFileName, "./OutputFiles/sides_%g.txt", (double) groupSize);
    if((fpUpperSides = fopen(sideFileName, "r")) == NULL)
    {
        printf("Can't open the file %s\n", sideFileName);
        exit(1);
    }

    groupWithMaxNumPathsCounter = 0;

    for(i = 0 ; i < numPlanarSides ; i++)
    {
        fread(&tempUpperSide, sideStructSize, 1, fpUpperSides);
        if(i%100 == 0)
            printf("Upper side is (%lld/%lld), with so far maxnumPaths is %d.\n", tempUpperSide.index, numPlanarSides-1, maxnumPaths);

        /// Opening the file "./OutputFiles/sides_<sideSize>.txt" to read the lower sides
        if((fpLowerSides = fopen(sideFileName, "r")) == NULL)
        {
            printf("Can't open the file %s\n", sideFileName);
            exit(1);
        }

        for(j = 0 ; j < i ; j++)
            fread(&tempLowerSide, sideStructSize, 1, fpLowerSides);

        for(j = i ; j < numPlanarSides ; j++)
        {
            fread(&tempLowerSide, sideStructSize, 1, fpLowerSides);
            //printf("Upper and lower sides are (%lld,%lld) of %lld, with so far maxnumPaths is %d\n", tempUpperSide.index, tempLowerSide.index, numPlanarSides-1, maxnumPaths);

            if(checkCompatibility(tempUpperSide, tempLowerSide) == 1)
            {
                tempGroup = createTempGroup(tempUpperSide, tempLowerSide);
                tempGroup.index = numGroups;

                /// calculates total number of paths in tempGroup
                numPaths =  computeNumPaths(tempGroup);

                if(maxnumPaths < numPaths)
                {
                    maxnumPaths = numPaths;
                    groupWithMaxNumPathsCounter = 0;
                    groupWithMaxNumPaths[groupWithMaxNumPathsCounter] = tempGroup;
                    groupWithMaxNumPathsCounter++;
                }
                else if(maxnumPaths == numPaths)
                {
                    groupWithMaxNumPaths[groupWithMaxNumPathsCounter] = tempGroup;
                    groupWithMaxNumPathsCounter++;
                }
                numGroups++;
            }//if(checkCompatibility) ends

            /// reverse the lower side around Y axis
            tempLowerSide = reverseSideWithYAxis(tempLowerSide);

            if(checkCompatibility(tempUpperSide, tempLowerSide) == 1)
            {
                tempGroup = createTempGroup(tempUpperSide, tempLowerSide);
                tempGroup.index = numGroups;

                /// calculates total number of paths in tempGroup
                numPaths =  computeNumPaths(tempGroup);

                if(maxnumPaths < numPaths)
                {
                    maxnumPaths = numPaths;
                    groupWithMaxNumPathsCounter = 0;
                    groupWithMaxNumPaths[groupWithMaxNumPathsCounter] = tempGroup;
                    groupWithMaxNumPathsCounter++;
                }
                else if(maxnumPaths == numPaths)
                {
                    groupWithMaxNumPaths[groupWithMaxNumPathsCounter] = tempGroup;
                    groupWithMaxNumPathsCounter++;
                }
                numGroups++;
            }//if(checkCompatibility) ends

        }//for(j) for lowerside ends

        fclose(fpLowerSides);
    }//for(i) for upperside ends

    fclose(fpUpperSides);
    maxNumPathsOfGroupsOfSizen = maxnumPaths;
    printf("Maximum number of patterns in a group of size %d is %d.\n", groupSize, maxnumPaths);
    printf("There are total %lld groups of size %d.\n", numGroups, groupSize);

    return numGroups;
}


/**
  * Input 1: tempUpperSide (type: group).
  * Input 2: tempLowerSide (type: group).
  * Output: If both sides have no common inner edge means they are compatible returns 1,
  * otherwise returns 0.
  */
int checkCompatibility(group tempUpperSide, group tempLowerSide)
{
    int l, m, status = 1, groupSize = tempUpperSide.numV;

    for(l = 1 ; l <= groupSize+1 ; l++)
        for(m = l+2 ; m <= groupSize ; m++)
            if( (tempUpperSide.adjMat[l][m] == 1) && (tempLowerSide.adjMat[l][m] == 1) )
                status = 0;

    return status;
}


/// returns a group generated from two sides
group createTempGroup(group tempUpperSide, group tempLowerSide)
{
    int l, m;
    group tempGroup;

    tempGroup.numV = tempUpperSide.numV;
    for(l = 0 ; l <= tempGroup.numV+1 ; l++)
        for(m = 0 ; m <= tempGroup.numV+1 ; m++)
            tempGroup.adjMat[l][m] = 0;

    for(l = 0 ; l <= tempGroup.numV ; l++)
        tempGroup.adjMat[l][l+1] = 1;

    tempGroup.adjMat[0][tempGroup.numV+1] = 1;
    for(m = 2 ; m <= tempGroup.numV ; m++)
    {
        if( (tempUpperSide.adjMat[0][m] == 1) && (tempLowerSide.adjMat[0][m] == 0) )
                tempGroup.adjMat[0][m] = 1;
        if( (tempUpperSide.adjMat[0][m] == 0) && (tempLowerSide.adjMat[0][m] == 1) )
                tempGroup.adjMat[0][m] = -1;
        if( (tempUpperSide.adjMat[0][m] == 1) && (tempLowerSide.adjMat[0][m] == 1) )
                tempGroup.adjMat[0][m] = 2;
    }
    for(l = 1 ; l <= tempGroup.numV-1 ; l++)
    {
        if( (tempUpperSide.adjMat[l][tempGroup.numV+1] == 1) && (tempLowerSide.adjMat[l][tempGroup.numV+1] == 0) )
                tempGroup.adjMat[l][tempGroup.numV+1] = 1;
        if( (tempUpperSide.adjMat[l][tempGroup.numV+1] == 0) && (tempLowerSide.adjMat[l][tempGroup.numV+1] == 1) )
                tempGroup.adjMat[l][tempGroup.numV+1] = -1;
        if( (tempUpperSide.adjMat[l][tempGroup.numV+1] == 1) && (tempLowerSide.adjMat[l][tempGroup.numV+1] == 1) )
                tempGroup.adjMat[l][tempGroup.numV+1] = 2;
    }
    for(l = 1 ; l <= tempGroup.numV ; l++)
        for(m = l+2 ; m <= tempGroup.numV+1 ; m++)
        {
            if( (tempUpperSide.adjMat[l][m] == 1) && (tempLowerSide.adjMat[l][m] == 0) )
                tempGroup.adjMat[l][m] = 1;
            if( (tempUpperSide.adjMat[l][m] == 0) && (tempLowerSide.adjMat[l][m] == 1) )
                tempGroup.adjMat[l][m] = -1;
        }

    return(tempGroup);
}


/// calculates total number of paths in tempGroup
int computeNumPaths(group tempGroup)
{
        int numVisited[MAXN];
        int j, k;

        numVisited[0] = 1;
        for(j = 1 ; j <= tempGroup.numV+1 ; j++)
            numVisited[j] = 0;

        for(j = 0 ; j <= tempGroup.numV ; j++)
            for(k = j+1 ; k <= tempGroup.numV+1 ; k++)
                if(tempGroup.adjMat[j][k] != 0)
                    numVisited[k] += numVisited[j];

        return(numVisited[tempGroup.numV+1]);
}


group reverseSideWithYAxis(group tempSide)
{
    int l, m;
    group reverseTempSide;
    reverseTempSide.numV = tempSide.numV;
    /// reverse the adjacency matrix
    for(l = 0 ; l <= tempSide.numV+1 ; l++)
        for(m = 0 ; m <= tempSide.numV+1 ; m++)
            reverseTempSide.adjMat[l][m] = tempSide.adjMat[tempSide.numV+1-m][tempSide.numV+1-l];

    return reverseTempSide;
}


void drawMaxGroups(int groupSize)
{
    char textFilename[100], latexFilename[100], timebuff[100], systemCommand[100];
    FILE *fptext, *fplatex;
    int i;
    group tempGroup;

    /// Get current time of the experiment
    time_t now = time (0);
    strftime (timebuff, 100, "\\texttt{%H:%M:%S} on \\texttt{%m-%d-%Y}", localtime (&now));

    sprintf(latexFilename, "./OutputFiles/Max_Groups_latex_%g.tex", (double) groupSize);
    if((fplatex = fopen(latexFilename, "w")) == NULL)
    {
        printf("Can't open the file %s.\n", latexFilename);
        exit(1);
    }
    fprintf(fplatex, "\\documentclass[12pt]{article} \n");
    fprintf(fplatex, "\\usepackage{tikz} \n");
    fprintf(fplatex, "\\usepackage[landscape]{geometry} \n");
    fprintf(fplatex, "\\begin{document} \n");
    fprintf(fplatex, "The number of vertices is {\\color{red} $%d$}.\\\\ \n", groupSize);
    fprintf(fplatex, "\\hrule \n\\vspace{2mm} \n");

    sprintf(textFilename, "./OutputFiles/Max_Groups_text_%g.txt", (double) groupSize);
    if((fptext = fopen(textFilename, "w")) == NULL)
    {
        printf("Can't open the file %s.\n", textFilename);
        exit(1);
    }

    for(i = 0 ; i < groupWithMaxNumPathsCounter ; i++)
    {
        tempGroup = groupWithMaxNumPaths[i];
        drawGraph(fplatex, tempGroup);
        writePaths(fplatex, tempGroup);
        fprintf(fplatex, "\\hrule \n\\vspace{2mm} \n");
        printGraph(fptext, tempGroup);
    }

    fprintf(fplatex,"File generated at %s.\\\\ \n", timebuff);
    fprintf(fplatex, "\\end{document} \n");
    fclose(fplatex);
    fclose(fptext);

    sprintf(systemCommand, "pdflatex -output-directory=./OutputFiles ./OutputFiles/Max_Groups_latex_%g.tex > log", (double) groupSize);
    system(systemCommand);
    sprintf(systemCommand, "gnome-open ./OutputFiles/Max_Groups_latex_%g.pdf", (double) groupSize);
    system(systemCommand);
}


/// draws the tempGroup in latex file pointed by FILE pointer fp
void drawGraph(FILE *fp , group tempGroup)
{
    int j,l,m;
    double r;
    double leftBoundary, rightBoundary, upBoundary, downBoundary;

    leftBoundary = 1 - 0.2;
    rightBoundary = tempGroup.numV + 0.2;
    upBoundary = ( tempGroup.numV - 1 ) / 2.0 + 0.2;
    downBoundary = - ( ( tempGroup.numV - 1 ) / 2.0 + 0.5) ;

    fprintf( fp ,"\\begin{tikzpicture} \n");

    fprintf( fp ,"\\draw [white] [thick] (%f,%f) -- (%f,%f); \n", leftBoundary, upBoundary, rightBoundary, upBoundary);
    fprintf( fp ,"\\draw [white] [thick] (%f,%f) -- (%f,%f); \n", rightBoundary, upBoundary, rightBoundary, downBoundary);
    fprintf( fp ,"\\draw [white] [thick] (%f,%f) -- (%f,%f); \n", rightBoundary, downBoundary, leftBoundary, downBoundary);
    fprintf( fp ,"\\draw [white] [thick] (%f,%f) -- (%f,%f); \n", leftBoundary, downBoundary, leftBoundary, upBoundary);


    fprintf( fp ,"\\draw [black] [thick] (1,0) -- (%d,0); \n", tempGroup.numV);
    for(j = 1 ; j<= tempGroup.numV ; j++)
        fprintf( fp ,"\\draw [fill] (%d,0) circle [radius=0.05]; \n", j);
    for(j = 1 ; j<= tempGroup.numV ; j++)
        fprintf( fp ,"\\node at (%d,-.25) {$u_{%d}$};\n", j, j);


    for(m = 2 ; m <= tempGroup.numV ; m++)
    {
        r = ( m - 1 ) / 2.0;
        if(tempGroup.adjMat[0][m] == 1)
            fprintf( fp ,"\\draw [black] [thick] [latex-] (%d,0) to [out=90,in=0] (1,%f); \n", m, r);
        else if(tempGroup.adjMat[0][m] == -1)
            fprintf( fp ,"\\draw [black] [thick] [latex-] (%d,0) to [out=270,in=0] (1,%f); \n", m, -r);
        else if(tempGroup.adjMat[0][m] == 2)
        {
            fprintf( fp ,"\\draw [black] [thick] [latex-] (%d,0) to [out=90,in=0] (1,%f); \n", m, r);
            fprintf( fp ,"\\draw [black] [thick] [latex-] (%d,0) to [out=270,in=0] (1,%f); \n", m, -r);
        }
    }

    for(l = 1 ; l <= tempGroup.numV-1 ; l++)
    {
        r = ( tempGroup.numV - l ) / 2.0;
        if(tempGroup.adjMat[l][tempGroup.numV+1] == 1)
            fprintf( fp ,"\\draw [black] [thick] [-latex] (%d,0) to [out=90,in=180] (%d,%f); \n", l, tempGroup.numV, r);
        else if(tempGroup.adjMat[l][tempGroup.numV+1] == -1)
            fprintf( fp ,"\\draw [black] [thick] [-latex] (%d,0) to [out=270,in=180] (%d,%f); \n", l, tempGroup.numV, -r);
        else if(tempGroup.adjMat[l][tempGroup.numV+1] == 2)
        {
            fprintf( fp ,"\\draw [black] [thick] [-latex] (%d,0) to [out=90,in=180] (%d,%f); \n", l, tempGroup.numV, r);
            fprintf( fp ,"\\draw [black] [thick] [-latex] (%d,0) to [out=270,in=180] (%d,%f); \n", l, tempGroup.numV, -r);
        }
    }

    for(l = 1 ; l <= tempGroup.numV-1 ; l++)
        for(m = l+2 ; m <= tempGroup.numV ; m++)
        {
            if(tempGroup.adjMat[l][m] == 1)
                fprintf( fp ,"\\draw [black] [thick] [-latex] (%d,0) to [out=90,in=90] (%d,0); \n", l, m);
            else if(tempGroup.adjMat[l][m] == -1)
                fprintf( fp ,"\\draw [black] [thick] [-latex] (%d,0) to [out=270,in=270] (%d,0); \n", l, m);
        }

    fprintf( fp ,"\\end{tikzpicture} \n\n");
}


/// writes all the paths of the tempGroup in the FILE pointed by pointer fp
void writePaths(FILE *fp, group tempGroup)
{
    int i, j, k, status;
    int temp[20];
    int prev, next, pathCount = 1;
    int t = pow(2, tempGroup.numV);

    fprintf( fp, "The paths are\\\\\n");
    fprintf( fp, "{\\color{red} $1$}: $\\emptyset$ ");

    for(i = t-1 ; i >= 1 ; i--)
    {
        status = 1;
        for(j = tempGroup.numV-1 ; j >= 0 ; j--)
        {
            k = i >> j;
            temp[tempGroup.numV-j] = (k & 1);
        }

        prev = 0;
        j = 1;
        while(j <= tempGroup.numV)
        {
            if(temp[j] == 1)
            {
                next = j;
                if(tempGroup.adjMat[prev][next] != 0)
                    prev = next;
                else
                    status = 0;
            }
            j++;
        }

        if(tempGroup.adjMat[prev][tempGroup.numV+1] == 0)
            status = 0;

        if(status == 1)
        {
            pathCount++;
            fprintf( fp, "\\hspace{5mm} ");
            fprintf( fp, "{\\color{red} $%d$}: ", pathCount);

            for(j = 1 ; j <= tempGroup.numV ; j++)
                if(temp[j] == 1)
                    fprintf( fp, "%d ", j);
        }
    }

    fprintf( fp ,"\\\\ \n");
}


/// Prints the tempGroup in a text file pointed by FILE pointer fp
void printGraph(FILE *fp , group tempGroup)
{
    int l , m;

    fprintf( fp, "     ");
    for(l = 0 ; l <= tempGroup.numV+1 ; l++)
        fprintf( fp, "v%d,  ", l);
    fprintf( fp, "\n");
    fprintf( fp, "     ");
    for(l = 0 ; l <= tempGroup.numV+1 ; l++)
        fprintf( fp, "-----");
    fprintf( fp, "\n");

    for(l = 0 ; l <= tempGroup.numV+1 ; l++)
    {
        fprintf( fp, "v%d |  ", l);
        for(m = 0 ; m <= tempGroup.numV+1 ; m++)
            fprintf( fp, "%d,   ", tempGroup.adjMat[l][m]);
        fprintf( fp, "\n");
    }
}
\end{lstlisting}

\end{document}